\documentclass[oribibl]{llncs}

\usepackage{times}
\usepackage[dvipsnames, table]{xcolor}
\usepackage{amsmath}
\usepackage{amsthm}
\usepackage{amssymb}
\usepackage{mathrsfs}  
\usepackage[all]{xy}
\usepackage{multirow}
\usepackage{caption}
\usepackage{float}
\usepackage{multicol}
\usepackage{soul}
\usepackage{enumerate}
\usepackage{hhline}
\usepackage{bbm}
\usepackage{mathtools}
\usepackage{bm}
\usepackage{ dsfont }
\usepackage{pifont}
\usepackage{enumitem}
\usepackage{tikz}
\usepackage{array}
\usepackage{pdflscape}
\usepackage{mdframed}
\usepackage{thm-restate}
\usepackage{cleveref}
\usepackage{tikz-cd}
\usepackage{pgfplots}
\usepackage[T1]{fontenc}
\usepackage[ttdefault=true]{AnonymousPro}
\usepackage{adjustbox}

\usepackage{lineno}

 \usepackage{geometry}

\newcolumntype{P}[1]{>{\centering\arraybackslash}p{#1}}

\makeatletter
\DeclareRobustCommand*{\pmzerodot}{%
  \nfss@text{%
    \sbox0{$\vcenter{}$}
    \sbox2{0}%
    \sbox4{0\/}%
    \ooalign{%
      0\cr
      \hidewidth
      \kern\dimexpr\wd4-\wd2\relax 
      \raise\dimexpr(\ht2-\dp2)/2-\ht0\relax\hbox{%
        \if b\expandafter\@car\f@series\@nil\relax
          \mathversion{bold}%
        \fi
        $\cdot\m@th$%
      }%
      \hidewidth
      \cr
      \vphantom{0}
    }%
  }%
}
\newcommand{\A}{\ensuremath{\mathcal{A}}}
\newcommand{\B}{\ensuremath{\mathcal{B}}}
\newcommand{\C}{\ensuremath{\mathcal{C}}}
\newcommand{\D}{\ensuremath{\mathcal{D}}}

\newcommand{\F}{\ensuremath{\mathcal{F}}}
\newcommand{\T}{\ensuremath{\mathcal{T}}}
\newcommand{\s}{\ensuremath{\sigma}}
\renewcommand{\S}{\ensuremath{\mathcal{S}}}
\renewcommand{\P}{\ensuremath{\mathcal{P}}}
\newcommand{\vars}{\textit{vars}}
\newcommand{\ax}{\textit{Ax}}
\newcommand{\wit}{\textit{wit}}
\newcommand{\overarrow}{\overrightarrow}

\newcommand{\dash}{\vdash}
\newcommand{\tdash}{\vdash_{\T}}
\newcommand{\wrt}{\textit{w.r.t.} }
\newcommand{\smooth}{\textbf{SM}}
\newcommand{\stainf}{\textbf{SI}}
\newcommand{\convex}{\textbf{CV}}
\newcommand{\finwit}{\textbf{FW}}
\newcommand{\strfinwit}{\textbf{SW}}
\newcommand{\mc}{\textit{mincard}}

\newcommand{\TSM}{\T^{s}_{f}}
\newcommand{\Tgeqn}{\T_{\geq n}}

\newcommand{\TsM}{\T_{f}}
\newcommand{\Tinfty}{\T_{\infty}}

\newcommand{\Teven}{\T_{\mathit{even}}^{\infty}}

\newcommand{\Tninfty}{\T_{n, \infty}}

\newcommand{\Tone}{\T_{\leq 1}}
\newcommand{\Tleqn}{\T_{\leq n}}
\newcommand{\Tneqodd}{\T^{\neq}_{\mathit{odd}}}
\newcommand{\Tmn}{\T_{m,n}}
\newcommand{\Tneqoneinfty}{\T^{\neq}_{1,\infty}}
\newcommand{\Tneqtwoinfty}{\T^{\neq}_{2,\infty}}
\newcommand{\psiv}{\psi_{\vee}}
\newcommand{\addf}[1]{(#1)_{s}}
\newcommand{\adds}[1]{(#1)^{2}}
\newcommand{\addnc}[1]{(#1)_{\vee}}
\newcommand{\Taddf}{\addf{\T}}
\newcommand{\Tadds}{\adds{\T}}
\newcommand{\Taddnc}{\addnc{\T}}
\newcommand{\Ttwo}{\T_{2,3}}

\newcommand{\Tonetwo}{\T_{1}^{\infty}}
\newcommand{\Ttwotwo}{\T_{2}^{\infty}}
\newcommand{\Toddtwo}{\T_{1}^{odd}}
\newcommand{\Psiv}{\Psi_{\vee}}
\newcommand{\Fun}{\textit{Fun}}
\newcommand{\subs}[1]{#1_{1}}
\newcommand{\subsf}[2]{#1_{#2}}
\newcommand{\plusf}[1]{s(#1)}
\newcommand{\subf}[1]{\texttt{0}(#1)}
\newcommand{\subff}[1]{#1_{\texttt{0}}}
\newcommand{\subncf}[2]{#1^{\dag}_{#2}}
\newcommand{\dagg}[1]{#1^{\dag}}
\newcommand{\subnc}[1]{\pmzerodot(#1)}
\newcommand{\plusnc}[1]{\mathfrak{s}(#1)}

\usepackage{scalerel}
\newcommand\eq{\scaleobj{0.7}{=}}
\newcommand\diff{\scaleobj{0.7}{\neq}}
\newcommand{\fof}[1]{{{#1}{#1}}}
\newcommand{\floor}[1]{\lfloor #1 \rfloor}
\newcommand{\ceil}[1]{\lceil #1 \rceil}
\newcommand{\set}[1]{\left\{{#1}\right\}}
\newcommand{\kn}{\kappa}

\usepackage{authblk}

\title{Combining Combination Properties: \\
An Analysis of Stable Infiniteness, Convexity, and Politeness}

\author{Guilherme V. Toledo$^{\text{1}}$
\and Yoni Zohar$^{\text{1}}$
\and Clark Barrett$^{\text{2}}$
}\institute{
$^{\text{1}}$Bar-Ilan University
\enskip
$^{\text{2}}$Stanford University
}

\begin{document}


\setcounter{page}{1}     

\maketitle

\begin{abstract}
We make two contributions to the study of theory combination in satisfiability modulo theories. The first is a table of examples for the combinations of the most common model-theoretic properties in theory combination, namely stable infiniteness, smoothness, convexity, finite witnessability, and strong finite witnessability (and therefore politeness and strong politeness as well). All of our examples are sharp, in the sense that we also offer proofs that no theories are available within simpler signatures. This table significantly progresses the current understanding of the various properties and their interactions.
The most remarkable example in this table is of a theory over a single sort that is polite but not strongly polite (the existence of such a theory was only known until now for two-sorted signatures).
The second contribution is a new combination theorem showing that in order to apply
polite theory combination, it is sufficient for one theory to be stably infinite and strongly finitely witnessable, thus showing that smoothness is not a critical property in this combination method. This result has the potential to greatly simplify the process of showing which theories can be used in polite combination, as showing stable infiniteness is considerably simpler than showing smoothness.%
\footnote{This work was funded in part by NSF-BSF grant numbers 2110397 (NSF) and 2020704 (BSF) and ISF grant number 619/21. 
}
\end{abstract}

\section{Introduction}

Theory combination focuses on the following problem: 
given procedures for determining the satisfiability of formulas over individual theories, can we find
a procedure for the combined theory? One of the foundational results in this field is in Nelson and Oppen's paper of 1979, \cite{10.1145/357073.357079}, where the authors show how to combine theories with disjoint signatures as long as they are both stably infinite, i.e., for every quantifier-free formula that is satisfied in the theory, there is an infinite interpretation of the theory that satisfies it.

With the introduction of stable infiniteness was born the notion of identifying model-theoretic properties that enable theory combination.
It soon became clear, however, that this first step was insufficient, since some important theories with real-world applications (like the theories of bit-vectors and finite datatypes) turned out not to be stably infinite.
Early attempts to find alternatives for stable infiniteness in theory combination included the introduction of gentle \cite{10.1007/978-3-642-04222-5_16}, shiny \cite{TinZar-RR-04}, and flexible \cite{KrsGGT-TACAS-07} theories. We focus here on  the notion of \emph{politeness}, which forms the basis for theory combination in the state-of-the-art SMT solver cvc5 \cite{DBLP:conf/tacas/BarbosaBBKLMMMN22}.

First considered 
in~\cite{ranise:inria-00000570}, polite theories were originally defined as those theories that are both smooth and finitely witnessable. Both notions are much harder to test for than stable infiniteness, but once a theory is known to be polite, it can be combined with any other theory, even non-stably-infinite ones.

A small problem in the proof of the main result of the paper was corrected in later work 
\cite{JB10-TR}.  Their paper introduces a slightly different, more strict, definition of politeness, together with a correct proof showing that polite theories can be combined with arbitrary theories. Following \cite{Casal2018}, we refer to theories satisfying the new definition as \emph{strongly} polite, which is defined as being both smooth and \emph{strongly finitely witnessable}; with that in mind, we call theories satisfying the earlier definition simply {\em polite}.

For some time, it was not known whether there exists a theory that is polite but not strongly polite.  Then, in 2021 Sheng et al. \cite{SZRRBT-21} provided an example. This suggests the need for a more thorough analysis of properties such as stable infiniteness, smoothness, finite witnessability, and strong finite witnessability, as they appear to interact with each other in sometimes surprising or unforeseeable ways.  We add to this list \emph{convexity}, which was shown to be closely related to stable infiniteness in~\cite{BDS02-FROCOS02}.

In this paper, we provide an exhaustive analysis, with examples whenever such examples are possible, of whether and how these properties can coexist. Some combinations are obviously not possible, such as a strongly finitely witnessable theory that is not finitely witnessable; the feasibility of other combinations is more elusive; for instance, it is initially unclear whether there could be a one-sorted, non-stably-infinite theory that is also not finitely witnessable (we show that this is also  impossible).  A main result is a comprehensive table describing what is known about all possible combinations of these properties.

During the course of our efforts to fill in the 
table, we were also able to improve polite
combination: by making the involved proof slightly more difficult, we can simplify the main polite theory combination result:
 we show that in order to combine theories,
it is enough for one theory to be stably infinite and strongly finitely witnessable;  there is no need for smoothness.
This result simplifies the process of qualifying a theory for polite combination, as showing stable infiniteness is considerably simpler than showing smoothness.

The paper is organized as follows. \Cref{Preliminary notions} defines the basic notions we will make use of throughout the paper. \Cref{Somerestrictions} proves several theorems showing the unfeasibility of certain combinations of properties. \Cref{Examples} describes the example theories that populate the feasible entries of the table.
\Cref{Politetheoresrerevisited} offers a new combination theorem.
And finally, \Cref{Conclusion} 
gives concluding remarks and directions for future work.\footnote{Due to space limitations, proofs are included in an appendix.}

\section{Preliminary Notions}\label{Preliminary notions}

\subsection{First-order Signatures and Structures}

A many-sorted signature $\Sigma$ is a triple formed by a countable set $\S_{\Sigma}$ of {\em sorts}, a countable set of function symbols $\F_{\Sigma}$, and a countable set of predicate symbols $\P_{\Sigma}$ which contains, for every sort $\sigma\in \S_{\Sigma}$, an equality symbol $=_{\sigma}$ 
(often denoted by $=$); 
each function symbol has an arity $\sigma_{1}\times\cdots\times\sigma_{n}\rightarrow \sigma$ and each predicate symbol an arity $\sigma_{1}\times\cdots\times \sigma_{n}$, where $\sigma_{1}, \ldots, \sigma_{n}, \sigma\in \S_{\Sigma}$ and $n\in\mathbb{N}$.  Each equality symbol $=_{\s}$ has arity $\sigma\times\sigma$.
A signature with no function or predicate symbols other than equalities is called {\em empty}.

A many-sorted signature $\Sigma$ is {\em one-sorted} if $\S_{\Sigma}$ has one element; 
we may refer to many-sorted signatures simply as signatures. Two signatures are said to be \emph{disjoint} if they share only sorts and equality symbols.

We assume for each sort in $\S_{\Sigma}$ a distinct countably infinite set of variables, and define terms, literals, and formulas (atomic or not) in the usual way. 
If $s$ is a function symbol of arity $\s\rightarrow\s$ and $x$ is a variable of sort $\s$, we define recursively the term $s^{k}(x)$, for $k\in\mathbb{N}$, as follows: $s^{0}(x)=x$, and $s^{k+1}(x)=s(s^{k}(x))$. We denote the set of free variables of sort $\sigma$ in a formula $\varphi$ by $\vars_{\sigma}(\varphi)$, and given $S\subseteq\S_{\Sigma}$, $\vars_{S}(\varphi)=\bigcup_{\sigma\in S}\vars_{\sigma}(\varphi)$ (we use $\vars(\phi)$ as shorthand for $\vars_{\S_{\Sigma}}$).

A $\Sigma$\emph{-structure} $\A$ is composed of sets $\sigma^{\A}$ for each sort $\sigma\in\S_{\Sigma}$, called the \emph{domain of $\sigma$}, equipped with interpretations $f^{\A}$ and $P^{\A}$ of the function and predicate symbols, in a way that respects their arities. 
Furthermore, $=_{\sigma}^{\A}$ must be the identity on $\sigma^{\A}$.

A $\Sigma$\emph{-interpretation} $\A$ is an extension of a $\Sigma$-structure that also interprets variables, with the value of a variable $x$ of sort $\sigma$ being an element $x^{\A}$ of $\sigma^{\A}$; we will sometimes say that an interpretation $\B$ is an interpretation on a structure $\A$ (over the same signature) to mean that $\B$ has $\A$ as its underlying structure. We write $\alpha^{\A}$ for the interpretation of the term $\alpha$ under $\A$; if $\Gamma$ is a set of terms, we define $\Gamma^{\A} = \{\alpha^{\A} : \alpha\in \Gamma\}$. We write $\A\vDash\varphi$ if $\A$ satisfies $\varphi$.
A formula $\varphi$ is called {\em satisfiable} if it is
satisfied by some interpretation $\A$. 

We shall make use of standard cardinality formulas,
given in \Cref{card-formulas}.
$\psi^{\sigma}_{\geq n}$ is only satisfied
by a structure $\A$ if $|\sigma^{\A}|$
is at least $n$,
$\psi^{\sigma}_{\leq n}$ is only satisfied
by $\A$ if $|\sigma^{\A}|$
is at most $n$,
and 
$\psi^{\sigma}_{= n}$ is only satisfied
by $\A$ if $|\sigma^{\A}|$
is exactly $n$.
In one-sorted signatures, we may drop $\s$ from the formulas, giving us $\psi_{\geq n}$, $\psi_{\leq n}$ and $\psi_{= n}$. 

\begin{figure}[t]
\begin{mdframed}
\[\psi^{\sigma}_{\geq n}=\exists\, x_{1}.\cdots \exists\, x_{m}.\: \bigwedge_{1\leq i<j\leq n}\neg(x_{i}=x_{j})
\quad\quad\psi^{\sigma}_{\leq n}=\exists\, x_{1}.\cdots \exists\, x_{n}.\:\forall\, y.\: \bigvee_{i=1}^{n}y=x_{i}
\quad\quad\psi^{\sigma}_{= n}=\psi^{\sigma}_{\geq n}\wedge \psi^{\sigma}_{\leq n}
\]
\end{mdframed}
\caption{Cardinality Formulas}
\label{card-formulas}
\end{figure}

The following lemmas are generalizations of the standard compactness and downward Skolem-L\"owenheim theorems of first-order logic to the many sorted case. They are proved in \cite{Monzano93}. 

\begin{lemma}[\cite{Monzano93}]\label{Compactness}
Let $\Sigma$ be a signature; a set of $\Sigma$-formulas $\Gamma$ is satisfiable if and only if each of its finite subsets $\Gamma_{0}$ is satisfiable.
\end{lemma}

\begin{lemma}[\cite{Monzano93}]
\label{LowenheimSkolemDownwards}
Let $\Sigma$ be a signature and $\Gamma$ a set of $\Sigma$-formulas: if $\Gamma$ is satisfiable, there exists an interpretation $\A$ which satisfies $\Gamma$ and where $\s^{\A}$ is countable whenever it is infinite.
\end{lemma}

A \emph{theory} $\T$ is a class of all $\Sigma$-structures that satisfy some set of closed formulas (formulas without free variables), called the \emph{axiomatization} of $\T$ which we denote as $\ax(\T)$; such structures will be called the \emph{models} of $\T$, a model being called {\em trivial} when $\s^{\A}$ is a singleton for some sort $\s$ in $\S_{\Sigma}$. A $\Sigma$-interpretation $\A$ whose underlying structure is in $\T$ is called a $\T$-interpretation. A formula is said to be $\T$\emph{-satisfiable} if there is a $\T$-interpretation that satisfies it; a set of formulas is $\T$-satisfiable if there is a $\T$-interpretation that satisfies each of its elements. Two formulas are $\T$\emph{-equivalent} when every $\T$-interpretation satisfies one if and only if it satisfies the other. We write $\tdash\varphi$ and say that $\varphi$ is $\T$-valid if $\A\vDash\varphi$ for every $\T$-interpretation $\A$.
Let $\Sigma_{1}$ and $\Sigma_{2}$ be disjoint signatures; by $\Sigma=\Sigma_{1}\cup\Sigma_{2}$, we mean the signature with the union of the sorts, function symbols, and predicate symbols of $\Sigma_{1}$ and $\Sigma_{2}$, all arities preserved. 
%
Given a $\Sigma_{1}$-theory $\T_{1}$ and a $\Sigma_{2}$-theory $\T_{2}$, the $\Sigma_{1}\cup\Sigma_{2}$-theory $\T=\T_{1}\oplus \T_{2}$ is the theory 
axiomatized by the union of the axiomatizations of $\T_{1}$ and $\T_{2}$.


\subsection{Model-theoretic Properties}
Let $\Sigma$ be a signature.
A $\Sigma$-theory $\T$ is said to be \emph{stably infinite} w.r.t. $S\subseteq \S_{\Sigma}$ if, for every $\T$-satisfiable quantifier-free formula $\phi$, there exists a $\T$-interpretation $\A$ satisfying $\phi$ such that, for each $\sigma\in S$, $\sigma^{\A}$ is infinite.
%
$\T$ is \emph{smooth} w.r.t. $S\subseteq \S_{\Sigma}$ when, for every quantifier-free formula $\phi$, $\T$-interpretation $\A$ satisfying $\phi$, and function $\kappa$ from $S$ to the class of cardinals such that 
$\kappa(\s)\geq|\s^{\A}|$ for every $\s\in S$, there exists a $\T$-interpretation $\B$ satisfying $\phi$ with $|\sigma^{\B}|=\kappa(\sigma)$, for every $\sigma\in S$.
%
Notice that, in light of \Cref{LowenheimSkolemDownwards}, smoothness implies stable infiniteness (taking $\kappa(\sigma)=\omega$ for every $\sigma\in S$). 
We state this in a theorem.

\begin{theorem}\label{SMimpliesSI}
Let $\Sigma$ be a signature, $S\subseteq\S_{\Sigma}$,
and $\T$ a $\Sigma$-theory. If
 $\T$ is smooth w.r.t. $S$,
 then it is also stably infinite w.r.t. $S$.
\end{theorem}

For a set of sorts $S$, finite sets of variables $V_{\sigma}$ of sort $\sigma$ for each $\s\in S$, and equivalence relations $E_{\sigma}$ on $V_{\sigma}$, the arrangement on $V=\bigcup_{\s\in S}V_{\sigma}$ induced by $E=\bigcup_{\s\in S}E_{\sigma}$, denoted by $\delta_{V}$ or $\delta_{V}^{E}$, is the formula
\[\delta_{V}=\bigwedge_{\sigma\in S}\big[\bigwedge_{xE_{\sigma}y}(x=y)\wedge\bigwedge_{x\overline{E_{\sigma}}y}\neg(x=y)\big],\]
where $\overline{E_{\sigma}}$ denotes the complement of the equivalence relation $E_{\sigma}$.

A theory $\T$ is said to be \emph{finitely witnessable} w.r.t. the set of sorts $S\subseteq \S_{\Sigma}$ when there exists a function $\wit$, called a \emph{witness}, from the quantifier-free formulas into themselves that is computable and satisfies 
for every quantifier-free formula $\phi$:
$(i)$~$\phi$ and $\exists\, \overarrow{w}.\:\wit(\phi)$ are $\T$-equivalent, where $\overarrow{w}=\vars(\wit(\phi))\setminus\vars(\phi)$; and
$(ii)$ if $\wit(\phi)$ is $\T$-satisfiable, then there exists a $\T$-interpretation $\A$ satisfying $\wit(\phi)$ such that
$\sigma^{\A}=\vars_{\sigma}(\wit(\phi))^{\A}$
for each $\sigma\in S$.
$\T$ is said to be \emph{strongly finitely witnessable} if it has a strong witness $\wit$, which has the properties of a witness with the exception of $(ii)$, satisfying instead:
$(ii')$~given a finite set of variables $V$ and an arrangement $\delta_{V}$ on $V$, if $\wit(\phi)\wedge\delta_{V}$ is $\T$-satisfiable, then there exists a $\T$-interpretation $\A$ satisfying $\wit(\phi)\wedge\delta_{V}$ such that
$\sigma^{\A}=\vars_{\sigma}(\wit(\phi)\wedge\delta_{V}\big)^{\A}$
for all $\sigma\in S$.

From the definitions, the following theorem directly follows:
\begin{theorem}\label{SFWimpliesFW}
Let $\Sigma$ be a signature, $S\subseteq\S_{\Sigma}$, and
$\T$ a $\Sigma$-theory.
If $\T$ is strongly finitely witnessable w.r.t. $S$ 
then it is also finitely witnessable w.r.t. $S$.
\end{theorem}

A theory that is both smooth and finitely witnessable w.r.t. (a set of sorts) $\S$ is said to be \emph{polite} w.r.t. $\S$; a theory that is both smooth and strongly finitely witnessable w.r.t. $\S$ is called \emph{strongly polite} w.r.t. $\S$. 
For theories over one-sorted empty signatures, we have the following theorem from \cite{SZRRBT-21}:

\begin{theorem}[\cite{SZRRBT-21}]\label{OS+P=SP}
Every one-sorted theory over the empty signature that is polite w.r.t. its only sort is strongly polite w.r.t. that sort.
\end{theorem}

A one-sorted theory $\T$ is said to be \emph{convex} if, for any conjunction of literals $\phi$ and any finite set of variables $\{u_{1}, v_{1}, ... , u_{n}, v_{n}\}$, $\tdash\phi\rightarrow \bigvee_{i=1}^{n}u_{i}=v_{i}$ implies $\tdash\phi\rightarrow u_{i}=v_{i}$, for some $i\in [1,n]$.
%

Given a one-sorted theory $\T$, its $\mc$ function takes a quantifier-free formula $\phi$ and returns the countable cardinal
$\min\{|\s^{\A}| : \text{$\A$ is a $\T$-interpretation that satisfies $\phi$}\}$.\footnote{Note that this definition was generalized in two different ways to the many-sorted case in
\cite{Casal2018} and \cite{ranise:inria-00000570}.
However, for our investigation, the single-sorted case is enough.}

Throughout this paper, we will use
$\stainf$ for stably infinite,
$\smooth$ for smooth,
$\finwit$ for finitely witnessable,
$\strfinwit$ for strongly finitely witnessable, and
$\convex$ for convex.

\section{Negative Results}\label{Somerestrictions}
If it were possible, we would present examples of every combination of properties (stable infiniteness, smoothness, convexity, finite witnessability, and strong finite witnessability) using only the one-sorted empty signature, that is, the simplest signature imaginable.

Of course, this is not always possible: smooth theories are necessarily stably infinite, and strongly finitely witnessable theories are obligatorily finitely witnessable. But there are several other connections we now proceed to show, which further restrict the combinations of properties that are possible.

In \Cref{Restriction: convexity}, we show that, under reasonable conditions, a convex theory must be stably infinite, 
while the reciprocal is also true over the empty signature. 
In \Cref{Restrictions: FW}, we show that 
over the empty one-sorted signature,
theories that are 
not stably infinite are necessarily finitely witnessable (a somewhat counter-intuitive result, since we usually look for theories that are, simultaneously, smooth and strongly finitely witnessable) and, more importantly, that 
stably-infinite and strongly finitely witnessable one-sorted theories are also strongly polite.

\subsection{Stable-infiniteness and Convexity}\label{Restriction: convexity}

Convexity is typically defined over one-sorted signatures. 
Here we offer the following generalization to arbitrary signatures.

\begin{definition}
A theory $\T$ is said to be convex \wrt a set of sorts $S\subseteq \S_{\Sigma}$ if, 
for any conjunction of literals $\phi$ and any finite set of variables $\{u_{1}, v_{1}, ... , u_{n}, v_{n}\}$ with sorts in $S$,
if $\tdash\phi\rightarrow \bigvee_{i=1}^{n}u_{i}=v_{i}$ then 
$ \tdash\phi\rightarrow u_{i}=v_{i}$,
for some $i\in [1,n]$.
\end{definition}

 If we assume, as often it is natural to do, that our theories have no trivial models, then convexity implies stable infiniteness. This is true for the one-sorted case, as proved in \cite{BDS02-FROCOS02}, but also for the many-sorted case as we show here. 
 The proof is similar, though here we need to account for several sorts at once.
 In particular, we rely on \Cref{Compactness} for the proof.

 \begin{restatable}{theorem}{btonc}
\label{Barrett's theorem on convexity}
If a $\Sigma$-theory $\T$ is convex \wrt some set $S$ of sorts and, for each $\sigma\in\S$, $\tdash \psi^{\sigma}_{\geq 2}$, then $\T$ is stably infinite \wrt $S$.
\end{restatable}


\noindent
Reciprocally, we may also obtain convexity from stable infiniteness, but only over empty signatures.

\begin{restatable}{theorem}{sietac}
\label{SI empty theories are convex}
Any theory over an empty signature that is stably infinite \wrt the set of all of its sorts is convex \wrt any set of sorts.
\end{restatable}

As we shall see in \Cref{Overview}, this result is tight:
there are theories over non-empty signatures that are
stably infinite but not convex.


\subsection{More Connections}\label{Restrictions: FW}
We next present some more connections between the various properties.
%
%
First, over the one-sorted empty signature, a theory must be either stably infinite or finitely witnessable.

\begin{restatable}{theorem}{osessifw}
\label{OS+ES+-SI=>FW}
Every one-sorted, non-stably-infinite theory $\T$ with an empty signature is finitely witnessable w.r.t. its only sort.
\end{restatable}

The following theorem shows that for one-sorted theories, strong politeness is a corollary of strong finite witnessability and stable infiniteness (rather than smoothness).

\begin{restatable}{theorem}{sisfweqs}\label{SI+SFW=S}
Every one-sorted theory that is stably infinite and strongly finitely witnessable w.r.t. its only sort is smooth, and therefore strongly polite w.r.t. that sort.
\end{restatable}

\noindent
Generalizing this theorem to the case of many-sorted signatures is left for future work.

Finally, by combining previous results, we can also get the following theorem, which relates stable infiniteness, strong finite witnessability, and convexity.

\begin{restatable}{theorem}{osessisfwcc}
\label{OS+ES+-SI+-SFW=>-C}
A one-sorted theory $\T$ with an empty signature that is neither strongly finitely witnessable nor stably infinite w.r.t. its only sort cannot be convex.
\end{restatable}

\begin{figure}[t]
\centering
    \begin{tikzpicture}[scale=0.55]
\def\rectangle{(-4,-3) rectangle (4,5.5)}
\def\firstcircle{(-1,0) coordinate (a) circle (1.5cm)}
\def\secondcircle{(1,0) coordinate (b)  circle (1.5cm)}
\def\thirdcircle{(-1,0) coordinate (c) circle (2.5cm)}
\def\fourthcircle{(1,0) coordinate (d)  circle (2.5cm)}
\def\fifthcircle{(0, 2.7) coordinate (e) circle (2.5cm)}
    \begin{scope}
\fill[gray!50] \rectangle;
\fill[white] \fifthcircle;
\fill[white] \fourthcircle;
\fill[white] \thirdcircle;
\end{scope}
    \begin{scope}
\fill[gray!50] \fifthcircle;
\fill[white] \fourthcircle;
\fill[white] \thirdcircle;
\end{scope}
\begin{scope}
    \clip \fourthcircle;
    \fill[gray!50] \fifthcircle;
\end{scope}
\begin{scope}
    \fill[white] \secondcircle;
\end{scope}
    \begin{scope}
\fill[gray!50] \thirdcircle;
\end{scope}
    \begin{scope}
\clip \fifthcircle;
\fill[white] \thirdcircle;
\end{scope}
\begin{scope}
\clip \firstcircle;
\fill[gray!50] \fourthcircle;
\end{scope}
\begin{scope}
\clip \secondcircle;
\fill[gray!50] \thirdcircle;
\end{scope}
\begin{scope}
\clip \firstcircle;
\clip \fifthcircle;
\fill[white] \secondcircle;
\end{scope}
\draw \rectangle;
\draw \firstcircle;
\draw \secondcircle;
\draw \thirdcircle;
\draw \fourthcircle;
\draw\fifthcircle;
\node[label={\scriptsize $\stainf$}] (B) at (-2.95,-0.8) {};
\node[label={\scriptsize $\finwit$}] (B) at (3,-0.8) {};
\node[label={\scriptsize $\strfinwit$}] (B) at (2,-0.8) {};
\node[label={\scriptsize $\smooth$}] (B) at (-2,-0.8) {};
\node[label={\scriptsize $\convex$}] (B) at (0,2.7) {};
    \end{tikzpicture}
    \caption{A diagram of possible combinations over a one-sorted, empty signature: gray regions are empty.}\label{venn-mono-empty}
\end{figure}

To summarize, while \Cref{Barrett's theorem on convexity} is restricted to structures with no domains of cardinality $1$, the remaining theorems of this section are not restricted to such structures.
\Cref{SI empty theories are convex} applies to empty signatures,
\Cref{SI+SFW=S} applies to one-sorted signatures,
and \Cref{OS+ES+-SI=>FW,OS+ES+-SI+-SFW=>-C} apply to signatures that are both empty and one-sorted.
Put together, we see that many combinations of properties for theories over a one-sorted empty signature are actually impossible.
This is depicted in \Cref{venn-mono-empty},
in which all areas but the white ones are empty.
For example, \Cref{OS+ES+-SI=>FW} shows
that the area outside the SI and FW circles (representing
theories that are neither stably infinite nor finitely witnessable) is empty,
as every theory (over an empty one-sorted signature) must have one of these properties.
Similarly, \Cref{OS+ES+-SI+-SFW=>-C} further shows that within the CV (convex) circle, even more is empty, namely anything outside the SI and SW circles.

\section{Positive Results}\label{Examples}\label{Empty signature}\label{Overview}

We now proceed to systematically address all possible combinations of stable-infiniteness, smoothness, finite witnessability, strong finite witnessability, and convexity.

\begin{table}[t]
\renewcommand{\arraystretch}{1}
\centering
\begin{tabular}{|P{0.6cm}|P{0.6cm}|P{0.6cm}|P{0.6cm}|P{0.6cm}|P{2cm}P{2cm}P{2cm}P{2cm}|P{0.5cm}|}
\hline
\multicolumn{5}{|c|}{} & \multicolumn{2}{c|}{Empty} & \multicolumn{2}{c|}{Non-empty} & \\
\hline
$\stainf$ & $\smooth$ & $\finwit$ & $\strfinwit$ & $\convex$ & \multicolumn{1}{c|}{One-sorted} & \multicolumn{1}{c|}{Many-sorted} &\multicolumn{1}{c|}{One-sorted} & \multicolumn{1}{c|}{Many-sorted}& $N^{\underline{o}}$\\\hline
\multirow{16}{*}{$T$} & \multirow{8}{*}{$T$}&\multirow{4}{*}{$T$}&\multirow{2}{*}{$T$}&$T$&\multicolumn{1}{c|}{$\Tgeqn$}&\multicolumn{1}{c|}{$\adds{\Tgeqn}$}& \multicolumn{1}{c|}{$\addf{\Tgeqn}$}&\multicolumn{1}{c|}{$\addf{\adds{\Tgeqn}}$}& 1\\\hhline{~~~~------}

&&&&$F$&\multicolumn{2}{c|}{\Cref{SI empty theories are convex}\cellcolor{red!15}}& \multicolumn{1}{c|}{$\addnc{\Tgeqn}$} &\multicolumn{1}{c|}{$\addnc{\adds{\Tgeqn}}$}& 2\\\hhline{~~~-------}

&&&\multirow{2}{*}{$F$}&$T$&\multicolumn{1}{c|}{\cellcolor{red!15}}&\multicolumn{1}{c|}{$\Ttwo$}&\multicolumn{1}{c|}{$\TsM$}&\multicolumn{1}{c|}{$\addf{\TsM}$}& 3\\\hhline{~~~~-*{1}{>{\arrayrulecolor{red!15}}|-}*{1}{>{\arrayrulecolor{black}}|-}---}

&&&&$F$&\multicolumn{1}{c|}{\multirow{-2}{*}{\cellcolor{red!15}\Cref{OS+P=SP}}}&\multicolumn{1}{c|}{\Cref{SI empty theories are convex}\cellcolor{red!15}}& \multicolumn{1}{c|}{$\TSM$}&\multicolumn{1}{c|}{$\addnc{\Ttwo}$}& 4\\\hhline{~~--------}

&&\multirow{4}{*}{$F$}&\multirow{2}{*}{$T$}&$T$& \multicolumn{4}{c|}{\cellcolor{red!15}}&5\\\cline{5-5}\cline{10-10}

&&&&$F$& \multicolumn{4}{c|}{\multirow{-2}{*}{\cellcolor{red!15} \Cref{SFWimpliesFW}}}& 6\\\cline{4-10}

&&&\multirow{2}{*}{$F$}&$T$&\multicolumn{1}{c|}{$\Tinfty$}&\multicolumn{1}{c|}{$\adds{\Tinfty}$}&\multicolumn{1}{c|}{$\addf{\Tinfty}$}&\multicolumn{1}{c|}{$\addf{\adds{\Tinfty}}$}&7\\\hhline{~~~~------}

&&&&$F$&\multicolumn{2}{c|}{\cellcolor{red!15}\Cref{SI empty theories are convex}}&\multicolumn{1}{c|}{$\addnc{\Tinfty}$}&\multicolumn{1}{c|}{$\addnc{\adds{\Tinfty}}$}& 8\\\hhline{~---------}

&\multirow{8}{*}{$F$}&\multirow{4}{*}{$T$}&\multirow{2}{*}{$T$}&$T$&\multicolumn{1}{c|}{\cellcolor{red!15}}&\multicolumn{1}{c|}{\textcolor{red}{Unicorn}}&\multicolumn{1}{c|}{\cellcolor{red!15}}& & 9\\\hhline{~~~~-*{1}{>{\arrayrulecolor{red!15}}|-}*{1}{>{\arrayrulecolor{black}}|-}*{1}{>{\arrayrulecolor{red!15}}|-}~*{1}{>{\arrayrulecolor{black}}|-}}

&&&&$F$&\multicolumn{1}{c|}{\multirow{-2}{*}{\cellcolor{red!15}\Cref{SI+SFW=S}}} &\multicolumn{1}{c|}{\cellcolor{red!15}\Cref{SI empty theories are convex}}&\multicolumn{1}{c|}{\multirow{-2}{*}{\cellcolor{red!15}\Cref{SI+SFW=S}}}&\multicolumn{1}{c|}{\multirow{-2}{*}{\textcolor{red}{Unicorn}}}&10\\\cline{4-10}

&&&\multirow{2}{*}{$F$}&$T$&\multicolumn{1}{c|}{$\Teven$}&\multicolumn{1}{c|}{$\adds{\Teven}$}&\multicolumn{1}{c|}{$\addf{\Teven}$}&\multicolumn{1}{c|}{$\addf{\adds{\Teven}}$}& 11\\\cline{5-8}\hhline{~~~~------}

&&&&$F$&\multicolumn{2}{c|}{\cellcolor{red!15}\Cref{SI empty theories are convex}}&\multicolumn{1}{c|}{$\addnc{\Teven}$}&\multicolumn{1}{c|}{$\addnc{\adds{\Teven}}$}& 12\\\hhline{~~--------}

&&\multirow{4}{*}{$F$}&\multirow{2}{*}{$T$}&$T$&\multicolumn{4}{c|}{\cellcolor{red!15}}&13\\\cline{5-5}\cline{10-10}

&&&&$F$&\multicolumn{4}{c|}{\multirow{-2}{*}{\cellcolor{red!15} \Cref{SFWimpliesFW}}}& 14\\\cline{4-10}

&&&\multirow{2}{*}{$F$}&$T$&\multicolumn{1}{c|}{$\Tninfty$}&\multicolumn{1}{c|}{$\adds{\Tninfty}$}&\multicolumn{1}{c|}{$\addf{\Tninfty}$}&\multicolumn{1}{c|}{$\addf{\adds{\Tninfty}}$}&15\\\hhline{~~~~------}

&&&&$F$&\multicolumn{2}{c|}{\cellcolor{red!15}\Cref{SI empty theories are convex}}&\multicolumn{1}{c|}{$\addnc{\Tninfty}$}&\multicolumn{1}{c|}{$\addnc{\adds{\Tninfty}}$}& 16\\\hline

\multirow{16}{*}{$F$}&\multirow{8}{*}{$T$}&\multirow{4}{*}{$T$}&\multirow{2}{*}{$T$}&$T$&\multicolumn{4}{c|}{\cellcolor{red!15}}& 17\\\cline{5-5}\cline{10-10}

&&&&$F$&\multicolumn{4}{c|}{\cellcolor{red!15}}& 18\\\cline{4-5}\cline{10-10}

&&&\multirow{2}{*}{$F$}&$T$&\multicolumn{4}{c|}{\cellcolor{red!15}}& 19\\\cline{5-5}\cline{10-10}

&&&&$F$&\multicolumn{4}{c|}{\multirow{-4}{*}{\cellcolor{red!15} \Cref{SMimpliesSI}}}& 20\\\hhline{~~---*{4}{|>{\arrayrulecolor{black}}-}-}

&&\multirow{4}{*}{$F$}&\multirow{2}{*}{$T$}&$T$&\multicolumn{4}{c|}{\cellcolor{red!15}}& 21\\\cline{5-5}\cline{10-10}

&&&&$F$&\multicolumn{4}{c|}{\multirow{-2}{*}{\cellcolor{red!15} Theorems \ref{SMimpliesSI} and \ref{SFWimpliesFW}}}&22\\\hhline{~~~--*{4}{|>{\arrayrulecolor{black}}-}-}

&&&\multirow{2}{*}{$F$}&$T$&\multicolumn{4}{c|}{\cellcolor{red!15}}& 23\\\cline{5-5}\cline{10-10}

&&&&$F$&\multicolumn{4}{c|}{\multirow{-2}{*}{\cellcolor{red!15} \Cref{SMimpliesSI}}}& 24\\\cline{2-10}

&\multirow{8}{*}{$F$}&\multirow{4}{*}{$T$}&\multirow{2}{*}{$T$}&$T$&\multicolumn{1}{c|}{$\Tone$}&\multicolumn{1}{c|}{$\adds{\Tone}$}&\multicolumn{1}{c|}{$\addf{\Tone}$}&\multicolumn{1}{c|}{$\addf{\adds{\Tone}}$}& 25\\\cline{5-10}

&&&&$F$&\multicolumn{1}{c|}{$\Tleqn$}&\multicolumn{1}{c|}{$\adds{\Tleqn}$}&\multicolumn{1}{c|}{$\addf{\Tleqn}$}&\multicolumn{1}{c|}{$\addf{\adds{\Tleqn}}$}& 26\\\hhline{~~~-------}

&&&\multirow{2}{*}{$F$}&$T$&\multicolumn{1}{c|}{\cellcolor{red!15}\Cref{OS+ES+-SI+-SFW=>-C}}&\multicolumn{1}{c|}{$\Toddtwo$}&\multicolumn{1}{c|}{$\Tneqodd$}& \multicolumn{1}{c|}{$\addf{\Toddtwo}$}&27\\\hhline{~~~~------}

&&&&$F$& \multicolumn{1}{c|}{$\Tmn$}&\multicolumn{1}{c|}{$\adds{\Tmn}$}&\multicolumn{1}{c|}{$\addf{\Tmn}$}&\multicolumn{1}{c|}{$\addf{\adds{\Tmn}}$}& 28\\\hhline{~~--------}

&&\multirow{4}{*}{$F$}&\multirow{2}{*}{$T$}&$T$&\multicolumn{4}{c|}{\cellcolor{red!15}}& 29\\\cline{5-5}\cline{10-10}

&&&&$F$&\multicolumn{4}{c|}{\multirow{-2}{*}{\cellcolor{red!15} \Cref{SFWimpliesFW}}}& 30\\\hhline{~~~-------}

&&&\multirow{2}{*}{$F$}&$T$&\multicolumn{1}{c|}{\cellcolor{red!15}}&\multicolumn{1}{c|}{$\Tonetwo$}&\multicolumn{1}{c|}{$\Tneqoneinfty$}&\multicolumn{1}{c|}{$\addf{\Tonetwo}$}& 31\\\hhline{~~~~-*{1}{>{\arrayrulecolor{red!15}}|-}*{1}{>{\arrayrulecolor{black}}|-}---}

&&&&$F$&\multicolumn{1}{c|}{\multirow{-2}{*}{\cellcolor{red!15}\Cref{OS+ES+-SI=>FW}}}&\multicolumn{1}{c|}{$\Ttwotwo$}&\multicolumn{1}{c|}{$\Tneqtwoinfty$}&\multicolumn{1}{c|}{$\addf{\Ttwotwo}$}& 32\\\hline
\end{tabular}
\renewcommand{\arraystretch}{1}
\vspace{1em}
\caption[Caption for LOF]{Summary of all possible combinations of theory properties. 
Red cells represent impossible combinations. 
In line $26$: $n>1$; in line $28$: $m>1$, $n>1$ and $|m-n|>1$.}
\label{tab-summary}
\end{table}

The results are summarized in \Cref{tab-summary}.
Each row 
corresponds to a possible combination of properties, as determined by the truth values in the first five columns.
For example, in the first row, the entries in the first five columns are all true, indicating that in this row, all theory examples must be
stably-infinite, smooth, finitely witnessable, strongly finitely witnessable, and convex.
The rest of the columns correspond to different possibilities for the theory signatures: either empty or non-empty, and either one-sorted or many-sorted.  Again, looking at the first row, we see four different theories listed, one for each of the signature possibilities.


Some entries in the table list theorems instead of providing example theories.  The listed theorems tell us that there do not exist any example theories for these entries.
For example,
lines $3$ and $4$ cannot provide 
examples over a one-sorted empty signature because of \Cref{OS+P=SP}.

When an example is available, its name is given in corresponding cell of the table.
The theories themselves are defined in 
\Cref{sec:theories_one_sort_no_fun,sec:theories_two_sort_no_fun,sec:theories_one_sort_with_fun,sec:theories_two_sort_with_fun}.
%
The examples on lines $25$, $27$ and $31$ must have at least one structure with a trivial domain (i.e., a domain with exactly one element) because of \Cref{Barrett's theorem on convexity}.

Lines $9$, $10$, $13$, and $14$ cover theories that are stably infinite and strongly finitely witnessable but not smooth.  We call these \emph{unicorn theories} because we could not find any such theories, nor do we believe they exist, but (ignoring the obvious cases ruled out by \Cref{SFWimpliesFW,SI empty theories are convex,SI+SFW=S}) we have no proof that they do not exist.

\begin{definition}\label{unicorntheories}
A \emph{unicorn theory} is stably infinite and strongly finitely witnessable but not smooth.
\end{definition}

%
\Cref{SI+SFW=S} shows that there are no one-sorted unicorn theories.
We believe it may be possible to provide a generalization of the upwards 
L\"owenheim-Skolem theorem to many-sorted logic in such a way that it would prove the non-existence of unicorn theories, which leads to the following conjecture:

\begin{conjecture}
There are no unicorn theories.
\end{conjecture}
 
Before defining the theories that are used in \Cref{tab-summary}, we introduce the following signatures.
\begin{definition}
\label{def:sigsofex}
$\Sigma_{1}$ is the empty one-sorted signature with sort $\sigma$, 
$\Sigma_{2}$ is the empty two-sorted signature with sorts $\sigma$ and $\sigma_{2}$,
and $\Sigma_{s}$ is the one-sorted signature with a single unary function symbol $s$.
\end{definition}

\noindent
We now describe the theories:
\Cref{sec:theories_one_sort_no_fun} describes the theories that are
over the empty one-sorted signature; 
\Cref{sec:theories_two_sort_no_fun} then continues to the next column,
describing theories over many-sorted empty signatures.
Some build on the theories of the previous column, but some are also new.
\Cref{sec:theories_one_sort_with_fun} describes the next column, one-sorted theories
over a non-empty signature.
Here, we use two constructions to generate new theories from previously introduced ones. 
One construction
 adds a function symbol to an empty signature (in a way that preserves all properties), and
the second preserves all properties but convexity, making it possible to construct non-convex examples in a uniform way.
We also present new theories when the constructions are not sufficient.
Finally, \Cref{sec:theories_two_sort_with_fun} describes theories over non-empty many-sorted signatures.\footnote{Proofs that each theory has the claimed properties can be found in the appendix.}

\subsection{Theories over the One-sorted Empty Signature}
\label{sec:theories_one_sort_no_fun}

\begin{figure}[t]
\centering
\begin{minipage}[t]{.4\textwidth}
  \centering
\renewcommand{\arraystretch}{1.8}
\centering
\begin{tabular}{c|c}
Name & Axiomatization\\
\hline
$\Tgeqn$  & $\{\psi_{\geq n}\}$\\
$\Tinfty$  & $\{\psi_{\geq k} :  k\in\mathbb{N}\}$\\
$\Teven$  & $\{\neg\psi_{=2k+1} :  k\in\mathbb{N}\}$\\
$\Tninfty$ & $\{\psi_{=n}\vee\psi_{\geq k} : k\in\mathbb{N}\}$\\
$\Tleqn$  & $\{\psi_{\leq n}\}$\\
$\Tmn$ & $\{\psi_{=m}\vee\psi_{=n}\}$\\

\end{tabular}
\renewcommand{\arraystretch}{1}

  \captionof{figure}{$\Sigma_{1}$-theories.
}
  \label{tab-theories-sigma-0}
\end{minipage}%
\begin{minipage}[t]{.6\textwidth}
  \centering
  \renewcommand{\arraystretch}{1.8}
\centering
\begin{tabular}{c|c}
Name & Axiomatization\\
\hline
$\Ttwo$ & $\{(\psi^{\s}_{=2}\wedge\psi^{\s_{2}}_{\geq k})\vee(\psi^{\s}_{\geq 3}\wedge \psi^{\s_{2}}_{\geq 3}) : k\in\mathbb{N}\}$\\
$\Toddtwo$  & $\{\psi^{\s}_{=1}\}\cup\{\neg\psi^{\s_{2}}_{=2k} : k\in\mathbb{N}\}$\\
$\Tonetwo$  & $\{\psi^{\s}_{=1}\}\cup\{\psi^{\s_{2}}_{\geq k} : k\in\mathbb{N}\}$\\
$\Ttwotwo$  & $\{\psi^{\s}_{=2}\}\cup\{\psi^{\s_{2}}_{\geq k} : k\in\mathbb{N}\}$
\end{tabular}
\renewcommand{\arraystretch}{1}
  \captionof{figure}{$\Sigma_2$-theories.}
  \label{tab-theories-sigma-2}
\end{minipage}
\end{figure}

The axiomatizations for theories over the one-sorted empty signature $\Sigma_{1}$ are given in \Cref{tab-theories-sigma-0}.
We briefly describe them here.

For each $n>0$, $\Tgeqn$ includes all structures with domains of cardinality at least $n$;
$\Tinfty$ is the theory including all structures whose domains are infinite;
$\Teven$ has structures with either an even or an infinite number of elements in their domains and was defined in \cite{SZRRBT-21}, where it was proved to be finitely witnessable, but neither smooth nor strongly finitely witnessable. The proofs justifying~\Cref{tab-summary} show additionally that it is stably infinite and convex.
$\Tninfty$ contains those structures whose domains have either exactly $n$ or an infinite number of elements;
$\Tleqn$ includes all structures with at most $n$ elements in their domains;
and for positive integers $m$ and $n$, $\Tmn$ 
has structures whose domains have either precisely $m$ elements, or precisely $n$ elements. 
%
This completes the first column of theory examples.

\begin{example}
The theory $\Tgeqn$ admits all considered properties,
while $\Tmn$ admits only finite witnessability.
\end{example}

\subsection{Theories over the Two-sorted Empty Signature}
\label{sec:theories_two_sort_no_fun}
We next introduce the theories over empty two-sorted signatures.
For many cases, we can simply add a trivial sort to one of the theories defined in \Cref{sec:theories_one_sort_no_fun}.
When this is not possible, we introduce new theories.

\subsubsection{Adding a sort to a theory}
Any $\Sigma_{1}$-theory can be used to generate a $\Sigma_2$-theory simply by adding the sort $\sigma_2$ to the signature (without changing the axiomatization).
This is formalized as follows:

\begin{definition}
\label{def:addsort}
Let $\T$ be a $\Sigma_{1}$-theory.
$\Tadds$ is the $\Sigma_{2}$-theory axiomatized by $\ax(\T)$.
\end{definition}

\begin{restatable}{lemma}{lemaddsort}
\label{lem:addsort}
Let $\T$ be a $\Sigma_{1}$-theory. Then: $\T$ is stably infinite, smooth, finitely witnessable, strongly finitely witnessable, or convex w.r.t. $\{\s\}$ if and only if $\Tadds$ is, respectively, stably infinite, smooth, finitely witnessable, strongly finitely witnessable, or convex w.r.t. $\{\s, \s_{2}\}$.
\end{restatable}

Using \Cref{def:addsort} and \Cref{lem:addsort},
we can populate many lines in the second column of examples by extending the corresponding theory from the previous column.

%
\begin{example}
$\adds{\Tgeqn}$ is a theory over two sorts, $\sigma$ and $\sigma_2$, whose structures
must have at least $n$ elements in the domain of $\sigma$ (but have no restrictions on the size of the domain of $\sigma_2$).
As seen in the first line of \Cref{tab-summary},
$\Tgeqn$ admits all the considered properties.
By \Cref{lem:addsort}, so does
$\adds{\Tgeqn}$.
\end{example}

\subsubsection{Additional theories over $\Sigma_{2}$}

On some lines, such as line $3$,
there is no theory over $\Sigma_{1}$ to extend.
In such cases, we cannot use \Cref{def:addsort} to construct
a many-sorted variant.

We introduce the theories shown in \Cref{tab-theories-sigma-2} to cover these cases.
The theory $\Ttwo$ contains two kinds of structures: (i) structures whose domains both have at least $3$ elements; and (ii) structures with exactly two elements in the domain of $\s$ and an infinite number of elements in the domain of $\s_{2}$. 
The theory ${\Toddtwo}$ has structures with exactly one element in the domain of $\s$ and either an odd or an infinite number of elements in the domain of $\s_2$. 
The theory $\Tonetwo$ is similar: it has structures with exactly one element in the domain of $\s$ and an infinite number of elements in the domain of $\s_2$. 
Finally, $\Ttwotwo$ is similar to $\Tonetwo$ except that its structures have exactly 2 elements in the domain of $\s$.

\begin{example}
    The theory $\Ttwo$ was first defined in \cite{Casal2018} and later used in \cite{SZRRBT-21}, where it was proved to be polite (and therefore smooth, stably infinite, and finitely witnessable) without being strongly polite (and therefore not strongly finitely witnessable). 
    The justification proofs for~\Cref{tab-summary} show that $\Ttwo$ is convex as well.\footnote{We thank Oded Padon for raising the question of whether there exists a theory that is polite and convex, but not strongly polite.}
\end{example}

\subsection{Theories over a One-sorted Non-empty Signature}
\label{sec:theories_one_sort_with_fun}
We continue to the next column
of \Cref{tab-summary}, with one-sorted non-empty signatures.
\Cref{sec:addf} shows how to construct
non-empty theories from one-sorted theories over the empty signature, while preserving all their properties.
In \Cref{sec:addnc}, we provide a similar construction which generates non-convex theories from the theories in the first column of examples.
And in \Cref{sec:additionalonenonempty}, we introduce additional theories not captured by the above constructions.
Two of these theories are described in more detail in
\Cref{sec:matitheories}.

\subsubsection{Extending a theory with a unary function symbol while preserving properties}
\label{sec:addf}
Whenever we have a theory over an empty signature,
we can construct a variant of it over a non-empty signature
by introducing a function symbol and interpreting it as the identity function.
This extension preserves all the properties that we consider.
This is formalized as follows.
\begin{definition}
\label{def:addfun}
Let $\Sigma_{n}$ be an empty signature with
sorts $S=\{\s_{1}, \ldots, \s_{n}\}$, 
and let $\T$ be a $\Sigma_{n}$-theory. 
The signature $\Sigma^{n}_{s}$ has sorts $S$ and a single unary function symbol
$s$ of arity $\s_{1}\rightarrow\s_{1}$, and
$\Taddf$ is the $\Sigma^{n}_{s}$-theory axiomatized by $\ax(\T)\cup\{\forall\, x.\:[s(x)=x]\}$, where $x$ is a variable of sort $\s_{1}$.
\end{definition}
\begin{restatable}{lemma}{addfunlem}
\label{lem:addfun}
For every theory $\T$ over an empty signature $\Sigma_{n}$ with sorts $S=\{\s_{1}, \ldots, \s_{n}\}$:
$\T$
is stably infinite, smooth, finitely witnessable, 
strongly finitely witnessable, or convex w.r.t. $S$ if and only if $\Taddf$ is, respectively, stably infinite, smooth, finitely witnessable, strongly finitely witnessable, or convex w.r.t. $S$.
\end{restatable}

\noindent
We use the operator $\addf{\cdot}$ in various places in \Cref{tab-summary}
in order to obtain examples in non-empty signatures from existing examples over $\Sigma_{1}$ and $\Sigma_{2}$.
\begin{example}
$\addf{\Tgeqn}$ is a one-sorted theory, whose structures
have at least $n$ elements and interpret the function symbol $s$ as the identity.
As seen in the first column of examples of \Cref{tab-summary},
$\Tgeqn$ admits all the considered properties.
By \Cref{lem:addfun}, so does
$\addf{\Tgeqn}$.
\end{example}

\subsubsection{Making a theory non-convex}
\label{sec:addnc}
The last general construction that we present aims at taking a 
theory and creating a non-convex variant of it while preserving the other 
properties we consider.
This can be done with the addition of a single unary function symbol $s$.
To define such a theory, we make use of the formula
$\psiv$
from \Cref{fig-nc}. Intuitively, $\psi_{\vee}$ states that in an interpretation $\A$ in which it holds, 
$s^{\A}(s^{\A}(a))$ must equal either $s^{\A}(a)$ or $a$ itself; in other words, either $a=s^{\A}(a)=s^{\A}(s^{\A}(a))$, $a=s^{\A}(s^{\A}(a))\neq s^{\A}(a)$,
or $a\neq s^{\A}(a)=s^{\A}(s^{\A}(a))$, as shown in \Cref{scenarios for psiv}.

\begin{figure}[t]
\begin{mdframed}
\[\psiv=\forall\, x.\:\big[\big(s^{2}(x)=x\big)\vee\big(s^{2}(x)=s(x)\big)\big]\]
\end{mdframed}
\caption{The formula $\psiv$ for non-convex theories.}
\label{fig-nc}
\end{figure}

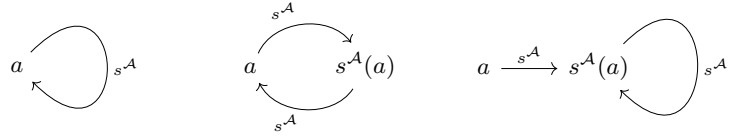
\begin{figure}[t]
\centering
\begin{minipage}{0.2\textwidth}
\begin{tikzcd}
a \arrow[loop right, distance=6em, in=315,out=45, "s^{\A}"]
\end{tikzcd}
\end{minipage}
\begin{minipage}{0.2\textwidth}
\begin{tikzcd}
a \arrow[r, bend left=60, "s^{\A}"]
& s^{\A}(a)\arrow[l, bend left=60, "s^{\A}"]
\end{tikzcd}
\end{minipage}
\begin{minipage}{0.2\textwidth}
\begin{tikzcd}
a \arrow[r, "s^{\A}"]
& s^{\A}(a)\arrow[loop right, distance=6em, in=315,out=45, "s^{\A}"]
\end{tikzcd}
\end{minipage}
\caption{Possible scenarios when $\psiv$ holds}\label{scenarios for psiv}
\end{figure}

This is especially useful for defining non-convex theories, since $(s^{2}(x)=x)\vee(s^{2}(x)=s(x))$ is valid in the theory, but neither $s^{2}(x)=x$ nor $s^{2}(x)=s(x)$ is.
Notice, of course, that non-convexity is only possible when there are at least two elements available in the domain -- otherwise, all equalities are satisfied.

\begin{definition}
\label{def:addfunconv}
Let $\T$ be a theory over an empty signature with sorts 
 $S=\{\s_{1}, \ldots, \s_{n}\}$.
 Then 
 $\Taddnc$ is the
$\Sigma^{n}_{s}$-theory axiomatized by $\ax(\T)\cup\{\psiv\}$.
\end{definition}

\begin{restatable}{lemma}{lemaddfunconv}
\label{Tvee is ...}
Let $\T$ be a theory over an empty signature $\Sigma_{n}$ with sorts $S=\{\s_{1}, \ldots, \s_{n}\}$.
Then: $\Taddnc$ is stably infinite, smooth, finitely witnessable, or strongly finitely witnessable w.r.t. $S$ if and only if $\T$ is, respectively, stably infinite, smooth, finitely witnessable, or strongly finitely witnessable w.r.t. $S$. 
 In addition, if $\T$ has a model $\A$ with $|\s_{1}^{\A}|\geq 2$, $\Taddnc$ is not convex with respect to $S$.

\end{restatable}

\begin{example}
The theory $\addnc{\Tgeqn}$ is a one-sorted theory whose structures have at least $n$ elements and interpret
the symbol $s$ in a way that satisfies $\psiv$ (i.e.,
for each element $a$ of the domain, one of the scenarios from \Cref{scenarios for psiv} holds).
As already mentioned, $\Tgeqn$ admits all properties.
According to \Cref{Tvee is ...},
the theory $\addnc{\Tgeqn}$ admits all properties but convexity.
\end{example}

\subsubsection{Additional theories over $\Sigma_s$}
\label{sec:additionalonenonempty}

\begin{figure}
\renewcommand{\arraystretch}{1.8}
\centering
\begin{tabular}{c|c}
Name & Axiomatization\\
\hline
$\TsM$ &  $\{[\psi^{=}_{\geq f_{1}(k)}\wedge \psi^{\neq}_{\geq f_{0}(k)}]\vee\bigvee_{i=1}^{k}[\psi^{=}_{=f_{1}(i)}\wedge \psi^{\neq}_{=f_{0}(i)}]: k\in\mathbb{N}\setminus\{0\}\}$\\
$\TSM$ &  $\ax(\TsM)\cup\{\psiv\}$\\
$\Tneqodd$ & $\{\psi_{=1}\vee[\neg\psi_{=2k}\wedge\forall\, x.\:\neg(s(x)=x)] : k\in\mathbb{N}\}$\\
$\Tneqoneinfty$  & $\{ \psi_{=1}\vee[\psi_{\geq k}\wedge\forall\, x.\:\neg(s(x)=s)] : k\in\mathbb{N}\}$\\
$\Tneqtwoinfty$ & $\{[\psi_{=2}\wedge\forall\, x.\:(s(x)=x)]\vee [\psi_{\geq k}\wedge \forall\, x.\:\neg(s(x)=x)] : k\in\mathbb{N}\}$\\
\end{tabular}
\renewcommand{\arraystretch}{1}
\caption{$\Sigma_s$-theories.}
\label{tab-theories-sigma-s}
\end{figure}

Whenever there is a $\Sigma_{1}$-theory with some properties, we can obtain a $\Sigma_s$ theory with the same properties using one of the techniques above.  To cover cases for which there is no corresponding $\Sigma_{1}$-theory, we use the theories presented in \Cref{tab-theories-sigma-s} and described below.

We start with $\Tneqodd$, $\Tneqoneinfty$, and $\Tneqtwoinfty$, 
deferring the discussion on $\TsM$ and $\TSM$ to \Cref{sec:matitheories}.
The theory $\Tneqodd$ has structures $\A$ with either an infinite or an odd number of elements and with the property that if $\A$ is not trivial, then $s^{\A}(a)\neq a$ for all $a\in \s^{\A}$. 
%
The theory $\Tneqoneinfty$ has all structures $\A$ that either: (i) are trivial; or (ii) have infinitely many elements and for which $s^{\A}(a)\neq a$ for each $a\in\s^{\A}$. 
Similarly, $\Tneqtwoinfty$ has structures $\A$ that either: (i) have exactly two elements and interpret $s$ as the identity; or (ii) have infinitely many elements and interpret $s$ in such a way that $s^{\A}(a)\neq a$ for all $a\in \s^{\A}$. 

\subsubsection{On the theories $\TsM$ and $\TSM$}
\label{sec:matitheories}
We now introduce the theories $\TsM$ and $\TSM$.
The importance of these theories is that both of them
are \emph{one-sorted} theories that are polite but not strongly polite (the first is also convex and the second is not).
Their existence improves on the result of \cite{SZRRBT-21}, which introduced
a \emph{two-sorted} theory that is polite but not strongly polite (namely $\Ttwo$).

For their axiomatizations, we use the formulas 
from \Cref{fig-card-s}, in which $s$ is a unary function symbol. $\psi^{\eq}_{\geq n}$ ($\psi^{\eq}_{=n}$) states that a structure $\A$ has at least (exactly) $n$ elements $a$ satisfying $s^{\A}(a)=a$; similarly, $\psi^{\diff}_{\geq n}$ ($\psi^{\diff}_{=n}$) states that a structure $\A$ has at least (exactly) $n$ elements $a$ satisfying $s^{\A}(a)\neq a$. 

\begin{figure}[t]
\begin{mdframed}
\[\psi^{\eq}_{\geq n}=\exists\, x_{1}.\cdots \exists\, x_{n}.\:\big[\bigwedge_{i=1}^{n}\big(s(x_{i})=x_{i}\big)\wedge \bigwedge_{1\leq i<j\leq n}\neg(x_{i}=x_{j})\big]\]

\[\psi^{\diff}_{\geq n}=\exists\, x_{1}.\cdots \exists\, x_{n}.\:\big[\bigwedge_{i=1}^{n}\neg\big( s(x_{i})=x_{i}\big)\wedge \bigwedge_{1\leq i<j\leq n}\neg(x_{i}=x_{j})\big]\]

\[\psi^{\eq}_{=n}=\exists\, x_{1}.\cdots \exists\, x_{n}.\:\big[\bigwedge_{i=1}^{n}\big(s(x_{i})=x_{i}\big)\wedge \bigwedge_{1\leq i<j\leq n}\neg(x_{i}=x_{j})\wedge\forall\, x.\:\big[\big(s(x)=x\big)\rightarrow\bigvee_{i=1}^{n}x=x_{i}\big]\big]\]

\[\psi^{\diff}_{=n}=\exists\, x_{1}.\cdots \exists\, x_{n}.\:\big[\bigwedge_{i=1}^{n}\neg\big(s(x_{i})=x_{i}\big)\wedge \bigwedge_{1\leq i<j\leq n}\neg(x_{i}=x_{j})\wedge\forall\, x.\:\big[\neg\big(s(x)=x\big)\rightarrow\bigvee_{i=1}^{n}x=x_{i}\big]\big].\]

\end{mdframed}
\caption{Cardinality formulas for signatures with a unary function symbol $s$.}
\label{fig-card-s}
\end{figure}

Further, the axiomatization requires a function $f$ from $\mathbb{N}$ to $\{0,1\}$ that is not computable with the property that for $k>0$, $f$ maps half of the numbers in the interval $[1,2^{k}]$
to $1$ and the other half to $0$.
The existence of such a function is formalized below. We start by defining counting functions $f_0$ and $f_1$.

\begin{definition}
Let $f:\mathbb{N}\setminus\{0\}\rightarrow\{0,1\}$.
For $i\in\{0,1\}$ and $n\in\mathbb{N}$, $f_{i}(n)$ is defined by: 
$f_{i}(n)=|f^{-1}(i)\cap[1,n]|$.
\end{definition}

\noindent
Intuitively,
$f_0(n)$ counts how many numbers between $1$ and $n$ (inclusive) are mapped by $f$ to $0$ and $f_{1}(n)$ counts how many are mapped to 1. Because $f(n)$ always equals $0$ or $1$, it is easy to see that for every $n>0$, $n=f_{1}(n)+f_{0}(n)$.

\begin{restatable}{lemma}{lemmatifexists}
\label{lem:mati-f-exists}
There exists a function $f:\mathbb{N}\setminus\{0\}\rightarrow\{0,1\}$ such that $f(1)=1$ with the following two properties:
\begin{enumerate}
\item $f$ is not computable;
\item for every $k\in\mathbb{N}\setminus\{0\}$, $f_{0}(2^{k})=f_{1}(2^{k})$.
\end{enumerate}

\end{restatable}

\begin{example}
The constant function that assigns $0$ to all natural numbers
satisfies neither the first nor the second condition of \Cref{lem:mati-f-exists}.
The function that assigns $0$ to even numbers and $1$ to odd 
numbers satisfies the second condition, but not the first.
Of course, any non-computable function satisfies the first condition. An example could be found by a function that
returns $1$ if the Turing machine that is encoded by the given number halts and $0$ otherwise, under some encoding.
Finding a function that admits both conditions is more challenging.
\end{example}

Let $f$ be some function with the properties listed in \Cref{lem:mati-f-exists}.  We can now define $\TsM$ 
over $\Sigma_s$ (note that $f$ itself is not a part of the signature, but is rather used to help define the axioms of $\TsM$).
$\TsM$ consists of those structures $\A$ that either (i) have a finite cardinality $n$, with $f_{1}(n)$ elements satisfying $s^{\A}(a)=a$, and $f_{0}(n)$
elements satisfying $s^{\A}(a)\neq a$ (and thus $\A$ satisfies $\psi^{=}_{\geq f_{1}(k)}\wedge \psi^{\neq}_{\geq f_{0}(k)}$ for $k\leq n$, and $\psi^{=}_{= f_{1}(n)}\wedge \psi^{\neq}_{= f_{0}(n)}$ and hence $\bigvee_{i=1}^{k}[\psi^{=}_{=f_{1}(i)}\wedge \psi^{\neq}_{=f_{0}(i)}]$ for all $k\geq n$); or (ii) have infinitely many elements, with infinitely many elements satisfying each condition, $s^{\A}(a)=a$ and $s^{\A}(a)\neq a$ (and thus $\A$ satisfies $\psi^{=}_{\geq f_{1}(k)}\wedge \psi^{\neq}_{\geq f_{0}(k)}$ for all $k\in\mathbb{N}$). Note that the description is well-defined because an element must always satisfy either $s^{\A}(a)=a$ or $s^{\A}(a)\neq a$, but never both or neither of these.
%
The theory $\TSM$ is similar to $\TsM$, but in addition to $\ax(\TsM)$ its structures must also satisfy $\psiv$.

\begin{remark}
The construction of $\TSM$ from $\TsM$ is very similar to the
general construction of 
\Cref{def:addfunconv}.
However, the corresponding result,
\Cref{Tvee is ...}, according to which all properties
but convexity
are preserved by this operation, is only shown 
in \Cref{Tvee is ...}
for cases where the original signature is empty, which is not the case for
$\TsM$. 
Obtaining $\TSM$ from $\TsM$ is not done by adding a function symbol, but rather by changing the axiomatization of the already existing
function symbol.  While we do prove that $\TSM$ has the required properties,
a general result in the style of \Cref{Tvee is ...} for arbitrary signatures,
with the ability to preserve an existing function symbol instead of adding a new one, is left for future work.
\end{remark}

\begin{example}
Let $\A_{n}$ be a $\Sigma_s$-model with domain $\{a_{1}, \ldots , a_{n}\}$ such that:
\[s^{\A_{n}}(a_{i})=\begin{cases*}
    a_{i} & if $1\leq i\leq f_{1}(n)$;\\
    a_{1} & if $f_{1}(n)<i\leq n$\\
\end{cases*}\]
(where the second condition may be void if $n=1$). Then $\A_{n}$ is a model of both $\TsM$ and $\TSM$.

If $\kappa$ is an infinite cardinal, let $\A_{\kappa}$ be a $\Sigma_s$-model with domain $A\cup\{a_{n} : n\in\mathbb{N}\setminus\{0\}\}$ (where $A$ is a set of cardinality $\kappa$ disjoint from $\{a_{n} : n\in\mathbb{N}\setminus\{0\}\}$) such that $s^{\A_{\kappa}}(a_{i})=a_{i}$ for each $i\in \mathbb{N}\setminus\{0\}$, and $s^{\A_{\kappa}}(a)=a_{1}$ for each $a\in A$.  Then $\A_{\kappa}$ is a model of both $\TsM$ and $\TSM$.
\end{example}

To show that $\TsM$ is smooth and finitely witnessable, we construct, given a $\TsM$-interpretation.
another $\TsM$-interpretation by (possibly) adding
two disjoint sets of elements to the interpretation,
one whose elements will satisfy
$s(a)=a$, and one whose elements will satisfy $s(a)\neq a$.

To show that it is not strongly finitely witnessable, we use the following lemmas, which are interesting in their own right.
According to the first, the $\mc$
function of $\TsM$ is not computable.
 
\begin{restatable}{lemma}{mctsmnc}
\label{mcofTsMisnotcomputable}
The $\mc$ function of $\TsM$ is not computable.
\end{restatable}

The second lemma that is needed in order to prove that $\TsM$ is not strongly finitely witnessable, is quite surprising.
As it turns out, 
for quantifier-free formulas,
the set of $\TsM$-satisfiable formulas coincides with the set of satisfiable formulas.
That is, even though the definition of $\TsM$ is very complex,
it induces the same satisfiability relation,
over quantifier-free formulas, 
as the simplest theory possible -- the theory axiomatized by the empty set (or, equivalently, all valid first-order sentences).

\begin{restatable}{lemma}{thmtrisattsm}
\label{Decidability of TsM}
Every quantifier-free $\Sigma_{s}$-formula that is satisfiable is $\TsM$-satisfiable.
\end{restatable}
\noindent
Note that \Cref{Decidability of TsM} does not hold
for quantified formulas in general. For example,
the formula $\forall\, x. \: s(x) \neq x$
is satisfiable but not $\TsM$-satisfiable:
because $f(1)=1$, every $\TsM$-interpretation
$\A$
must have at least one element $a$ with $s^{\A}(a)=a$.



The arguments for $\TSM$ are very similar, and require
minor changes in the corresponding proofs for $\TsM$.


\begin{remark}
We remark on the connection between the results regarding
$\TsM$ and $\TSM$, 
and those of \cite{DBLP:conf/lpar/CasalR13}.
What we show here is that 
$\TsM$ ($\TSM$) is 
polite but
not strongly polite.
Figure 1 of \cite{DBLP:conf/lpar/CasalR13} summarizes the relations
between these two properties for the one-sorted case.
It shows that 
polite theories that are 
axiomatized by a
universal set of axioms, and whose quantifier-free
satisfiability problem is decidable,
are strongly polite.
While $\TsM$ is decidable for quantifier-free formulas
(this is a corollary of \Cref{Decidability of TsM}),
its presentation here is definitely not as a universal theory.
On the other hand, \cite{DBLP:conf/lpar/CasalR13} also shows that
decidable polite theories for which checking if a finite interpretation
belongs to the theory is decidable are also strongly polite.
However, it is undecidable,
given an interpretation, to check whether it belongs
to $\TsM$ (and $\TSM$): such an algorithm would lead to an algorithm to compute $f$ as well. 
Thus, the theories $\TsM$ and $\TSM$ are polite, but do not
meet the criteria for strong politeness from \cite{DBLP:conf/lpar/CasalR13}.
And indeed, they are not strongly polite.
\end{remark}

\subsection{Theories Over Many-sorted Non-empty Signatures}
\label{sec:theories_two_sort_with_fun}
For the last column of \Cref{tab-summary}, all possible theories
can be obtained from theories that were already defined, using a combination of 
 \Cref{def:addfun,def:addsort,def:addfunconv},
 and so there is no need to present additional theories specifically for many-sorted non-empty signatures.

 \begin{example}
Line $1$ includes in its last column the theory $\addf{\adds{\Tgeqn}}$,
which is obtained from 
$\adds{\Tgeqn}$
using \Cref{def:addfun},
where the latter theory is obtained from
$\Tgeqn$
using \Cref{def:addsort}.
This theory admits all properties, including convexity.
To obtain a non-convex variant, the theory
$\addnc{\adds{\Tgeqn}}$ is constructed in a similar fashion,
using \Cref{def:addfunconv} instead of \Cref{def:addfun}.

\end{example}

With many-sorted non-empty signatures, we can always find an example for each combination of properties, except for those that are trivially impossible due to \Cref{SMimpliesSI,SFWimpliesFW} (i.e., theories that are strongly finitely witnessable but not finitely witnessable and theories that are smooth but not stably infinite).
This is nicely depicted by \Cref{venn-all}.
\Cref{SMimpliesSI,SFWimpliesFW} are represented in this figure by the location
of the circles: the circle for smooth theories is entirely inside the circle
for stably infinite theories, and similarly for strongly finitely witnessable and finitely witnessable theories.
Then, for every region in this figure, the right-most column of \Cref{tab-summary} has an example, the sole exception being the region that represents unicorn theories.

\begin{figure}[t]
\centering
    \begin{tikzpicture}[scale=0.8]
\def\firstcircle{(-1,0) coordinate (a) circle (1.5cm)}
\def\secondcircle{(1,0) coordinate (b)  circle (1.5cm)}
\def\thirdcircle{(-1,0) coordinate (c) circle (2.5cm)}
\def\fourthcircle{(1,0) coordinate (d)  circle (2.5cm)}
\def\fifthcircle{(0, 2.7) coordinate (e) circle (2.5cm)}
    \begin{scope}
\clip \fourthcircle;
\fill[blue!50] \firstcircle;
\end{scope}
\begin{scope}
\clip \thirdcircle;
\fill[red!50] \secondcircle;
    \end{scope}
    \begin{scope}
\clip \secondcircle;
\fill[violet!50] \firstcircle;
    \end{scope}
\draw \firstcircle;
\draw \secondcircle;
\draw \thirdcircle;
\draw \fourthcircle;
\draw\fifthcircle;
\node[fill=white,label={$\stainf$}] (B) at (-3.4,1.3) {};
\node[fill=white,label={$\finwit$}] (B) at (3.4,1.3) {};
\node[fill=white,label={\textcolor{red}{$\stainf$ and $\strfinwit$}}] (B) at (1.7,-2.2) {};
\node[fill=white,label={\textcolor{blue}{Polite}}] (B) at (-1.7,-2.2) {};
\node[fill=white,label={$\strfinwit$}] (B) at (2.7,0.3) {};
\node[fill=white,label={$\smooth$}] (B) at (-2.6,0.5) {};
\node[fill=white,label={\textcolor{violet}{Strongly polite}}] (B) at (0,-2.7) {};
\node[fill=white,label={$\convex$}] (B) at (0,2.7) {};
    \end{tikzpicture}
    \caption{A diagram of the various notions studied in this paper.}\label{venn-all}
\end{figure}
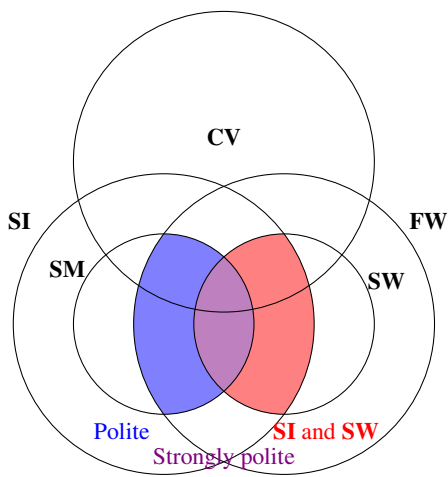

\begin{remark}
For non-empty signatures, we chose to include functions rather than predicates.
This is not essential as we can replace function symbols by predicate symbols by including the sort of the result of the function as the last component of the arity of the predicate, 
and then adding an axiom that forces the predicate to be a function.
\end{remark}

\section{Polite Combination without Smoothness}\label{Countable and uncountable smoothness}\label{Politetheoresrerevisited}

Polite combination of theories was introduced in \cite{ranise:inria-00000570}.
There, it was claimed that in order to combine a theory $\T$ with any other theory using polite combination, it suffices for $\T$
to be smooth and finitely witnessable (that is, polite).
Later, in \cite{JB10-TR}, this condition was corrected, and it was shown that in fact a stronger requirement is needed from $\T$: it has
to be smooth and strongly finitely witnessable (that is, strongly polite) to be applicable for the combination method.

Given that weakening strong finite witnessability to finite witnessability results in a condition that does not suffice,
it is natural to ask whether there is any other way to weaken the required conditions for polite combination.
Rather than weakening strong finite witnessability to finite witnessability, here
we consider another option: weakening the smoothness condition to stable infiniteness.
Thus, the main result of this section is that polite combination
can be done for theories that are stably infinite and strongly finitely witnessable, even if they are not smooth.

Our contribution can be understood by viewing \Cref{venn-all},
ignoring the circle that represents convexity (a property unrelated to the current section).
\cite{JB10-TR} shows that polite combination can be done for the purple region, which represents
smooth and strongly finitely witnessable theories.
\cite{JB10-TR} also presented an example showing that expanding the same combination method to the blue region, which represents 
smooth and finitely witnessable theories, results in an error.
Here we instead expand polite theory combination to the red region, which represents stably infinite and strongly finitely witnessable theories.
Now, the red region, if not empty, is only populated by unicorn theories (see \Cref{Overview}).
If such theories do not exist, then the result follows immediately.
Until this is settled, however, we provide a direct proof,
regardless of the existence of unicorn theories.

The next theorem  
shows that polite theory combination can be done for theories that are not necessarily strongly polite
(smooth and strongly finitely witnessable), but rather that are simply stably infinite and strongly finitely witnessable.

\begin{restatable}{theorem}{newcombteo}
\label{newcombinationtheorem}
Let $\Sigma_{1}$ and $\Sigma_{2}$ be disjoint signatures with sorts $S_{1}$ and $S_{2}$; let $\T_{1}$ be a $\Sigma_{1}$-theory, $\T_{2}$ be a $\Sigma_{2}$-theory, and $\T=\T_{1}\oplus \T_{2}$; and let $\phi_{1}$ be a quantifier-free $\Sigma_{1}$-formula and $\phi_{2}$ a quantifier-free $\Sigma_{2}$-formula.

Assume that $\T_{2}$ is stably-infinite and strongly finitely witnessable w.r.t. $S=S_{1}\cap S_{2}$, with strong witness $\wit$. 
Let $\psi=\wit(\phi_{2})$,  
$V_{\s}=\vars_{\s}(\psi)$ for every $\s\in S$
and $V=\bigcup_{\s\in S}\vars_{\s}(\psi)$.
Then the following are equivalent:
\begin{enumerate}
\item $\phi_{1}\wedge\phi_{2}$ is $\T$-satisfiable;
\item there exists an arrangement $\delta_{V}$ over $V$ such that $\phi_{1}\wedge\delta_{V}$ is $\T_{1}$-satisfiable and $\psi\wedge\delta_{V}$ is $\T_{2}$-satisfiable.
\end{enumerate}
\end{restatable}

\noindent
This result greatly improves the state-of-the-art in polite theory combination, which requires proving that one of the theories is both smooth and strongly finitely witnessable.
Thanks to this theorem, proving smoothness can be replaced by proving stable infiniteness,
which is typically a much easier task.

\section{Conclusion}\label{Conclusion}

As mentioned, there are two main contributions offered in this paper, both associated with the theme of theory combination. In \Cref{Overview}, we provide a table with examples for almost all the combinations of stable infiniteness, smoothness, convexity, finite witnessability, and strong finite witnessabillity known not to be impossible.
\Cref{Somerestrictions} provides theorems proving the sharpness of the examples provided.
The second contribution is a new combination theorem, according to which polite theory combination can be done without smoothness, provided we have instead stable infiniteness.

Many ideas for future work rise from the studies here presented. A first direction would be to settle the question of whether unicorn theories exist: if they do not, a proof would probably involve an interesting generalization of the upward L\"owenheim-Skolem theorem for many-sorted logic and would imply that strongly polite theories are just simply stably-infinite and strongly finitely witnessable theories,
thus greatly simplifying the proof of \Cref{newcombinationtheorem}; if unicorn theories do exist, one wonders if they can be combined in some meaningful way. 
Another direction of future work involves considering other model-theoretic properties in our table, such as shininess, gentleness, flexibility, and so on, as well as the effect of taking proper subsets of sorts for signatures containing more than one sort.

\newpage

\bibliographystyle{plain}
\bibliography{bib}

\newpage
\appendix

\newgeometry{left=25mm,bottom=25mm, top=25mm, right=25mm}

\section{Additional Notions for the Appendix}
We provide additional notions that will be used in the appendix.

If we replace each occurrence of a variable $x$ in a formula $\varphi$ by $y$, we denote the resulting formula by $\varphi[x/y]$; if we replace the variables $x_{1}$ through $x_{n}$ of $\varphi$ by, respectively, $y_{1}$ through $y_{n}$,
the resulting formula is denoted by $\varphi[x_{1}/y_{1}, \cdots , x_{n}/y_{n}]$.

The following result 
is a corollary of exercise $10.5$ in Section $10$ of \cite{BradleyManna2007}.

\begin{theorem}\label{uninterpretedfunctionsis convex}
    If $\T$ is a theory over a one-sorted signature with only one unary function, and if $\T$ is axiomatized by the empty set (that is, its functions are \emph{uninterpreted}), then $\T$ is convex.
\end{theorem}

\section{Proof of \Cref{Barrett's theorem on convexity}}

\btonc*

\begin{proof}
Suppose $\T$ is not stably-infinite \wrt $S$: then there exists a quantifier-free formula $\phi^{\prime}$ that is $\T$-satisfiable, but every $\T$-interpretation $\A$ that satisfies $\phi^{\prime}$ has $\s^{\A}$ finite for some $\s\in S$. Since every quantifier-free formula is equivalent to a disjunction of conjunctions of literals called its \emph{Disjunctive Normal Form} ($DNF$), being a conjunction of literals called a {\em cube}, 
we state that a cube $\phi$ in the $DNF$ of $\phi^{\prime}$ also has this property: suppose that this is not actually true, and so every $\T$-satisfiable cube in the $DNF$ of $\phi^{\prime}$ is satisfied by a $\T$-interpretation $\A$ with $\s^{\A}$ infinite for every $\s\in S$; since $\phi^{\prime}$, 
being equivalent to a disjunction of these cubes, must also be satisfied by each of these $\T$-interpretations $\A$, we reach a contradiction. So, to summarize, there is a conjunction of literals $\phi$ that is $\T$-satisfiable, but every $\T$-interpretation $\A$ that satisfies $\phi$ must have $\min\{|\sigma^{\A}| : \sigma\in S\}$ finite.

By taking $\overarrow{x}=\vars(\phi)$, we have that any $\T^{\prime}$-interpretations $\B$, for $\T^{\prime}$ the theory with axiomatization $\ax^{\prime}=\ax(\T)\cup \{\exists\, \overarrow{x}.\:\phi\}$, must have $\min\{|\sigma^{\B}| : \sigma\in S\}$ finite. 
Therefore,
the set
\[\ax^{\prime}\cup\{\psi^{\s}_{\geq k} : \sigma\in S, k\in\mathbb{N}\}\]
is unsatisfiable, since the latter set of formulas states that the domain of sort $\sigma$ of an interpretation is infinite, for each $\sigma\in S$. By \Cref{Compactness}, there must exist some $K_{\sigma}$ for each $\sigma\in S$ and a finite subset of $\ax^{\prime}$ that, together with the set
$\{\psi^{\s}_{\geq K_{\s}} : \sigma\in S\}$
is unsatisfiable (notice this is the case since the formula
$\psi^{\s}_{\geq k^{\prime}}\rightarrow \psi^{\s}_{\geq k}$ is satisfied if $k^{\prime}\geq k$, being therefore enough to take $K_{\sigma}$ as the maximum of these indices). Of course, this means $\ax^{\prime}\cup \{\psi^{\s}_{\geq K_{\s}} : \sigma\in S\}$ is unsatisfiable, and so if $\B$ is a $\T^{\prime}$-interpretation with $|\tau^{\B}|\geq K_{\tau}$ for every sort $\tau\in S\setminus\{\sigma\}$, then $|\s^{\B}|<K_{\s}$.

Let $z_{\sigma,1}$ through $z_{\sigma,K_{\s}}$ be fresh variables of sort $\sigma$ for every $\s\in S$: we wish now to show that 
\[\dash_{\T}\phi\rightarrow \bigvee_{\sigma\in S}\bigvee_{1\leq i<j\leq K_{\s}}z_{\sigma,i}=z_{\sigma, j}\]
but $\not\dash_{\T}\phi\rightarrow z_{\s, i}=z_{\s, j}$ for any $\s\in S$ and $1\leq i<j\leq K_{\s}$, contradicting the fact that $\T$ is supposed to be convex, and thus allowing us to reach the conclusion that $\T$ is stably-infinite.

So, suppose $\C$ is a $\T$-interpretation that satisfies $\phi$ (and so is a $\T^{\prime}$-interpretation), and we show that it must also satisfy $\bigvee_{\sigma\in S}\bigvee_{1\leq i<j\leq K_{\s}}z_{\sigma,i}=z_{\sigma, j}$ by contradiction: suppose that $\C$ is able not to satisfy $\bigvee_{1\leq i<j\leq K_{\tau}}z_{\tau, i}=z_{\tau, j}$ for all $\tau\in S\setminus\{\s\}$, and so has at least $K_{\tau}$ elements of each of these sorts. From our analysis of $\T^{\prime}$, this means that $\s^{\C}$ must have less than $K_{\s}$ elements, and by the pigeonhole 
principle $\C$ satisfies $\bigvee_{1\leq i<j\leq K_{\s}}z_{\s, i}=z_{\s, j}$, and thus $\bigvee_{\sigma\in S}\bigvee_{1\leq i<j\leq K_{\s}}z_{\sigma,i}=z_{\sigma, j}$ as we had previously stated.

But, by one of the hypothesis of the lemma, any $\T$-interpretation has more than one element in the domain of sort $\sigma$, for each $\sigma\in S$; so, if $\C$ in addition satisfies $\phi$, for each $\s\in S$ and pair of elements $1\leq i<j\leq K_{\s}$, given that $z_{\s,i}$ and $z_{\s,j}$ are variables that do not occur in $\phi$, we can construct a $\T$-interpretation $\C_{\s, i,j}$ that only differs from $\C$ in the values given to $z_{\s,i}$ and $z_{\s,j}$, where we have instead $z_{\s,i}^{\C_{\s, i,j}}\neq z_{\s,j}^{\C_{\s, i,j}}$. Since $\C$ and $\C_{\s,i,j}$ agree on the values given to the variables in $\phi$, and $\phi$ is satisfied by $\C$, we have that $\C_{\s,i,j}$ satisfies $\phi$ but not $z_{\s,i}=z_{\s,j}$, meaning
\[\not\dash_{\T}\phi\rightarrow z_{\s,i}=z_{\s,j}\]
for each $\s\in S$ and $1\leq i<j\leq K_{\s}$, thus finishing the proof.
\end{proof}

\section{Proof of \Cref{SI empty theories are convex}}
\sietac*

\begin{proof}
Fix a theory $\T$ over the empty signature that is stably-infinite w.r.t. the set of all of its sorts $S$. Suppose $\phi$ is a conjunction of literals, and $u_{1}, v_{1}, ... , u_{n},v_{n}$ are variables with sorts in some subset of $S$ (say, $u_{i}$ and $v_{i}$ both have sort $\s_{i}$), such that 
$\tdash \phi\rightarrow\bigvee_{i=1}^{n}u_{i}=v_{i}$.

Now, it is easy to see we cannot have $\tdash \phi\rightarrow(u_{i}=v_{i})$ for $u_{i}$ or $v_{i}$ not in $\phi$: in fact, take a $\T$-interpretation $\A$ which satisfies $\phi$, and since $\T$ is stably-infinite, there exists an infinite $\T$-interpretation $\A^{\prime}$ which also satisfies $\phi$; by changing its value on $u_{i}$, respectively $v_{i}$, to a value different from $v_{i}^{\A^{\prime}}$, respectively $u_{i}^{\A^{\prime}}$, we obtain a third $\T$-intepretation $\A^{\prime\prime}$ that still satisfies $\phi$ but not $u_{i}=v_{i}$. So we may restrict ourselves to the pairs $(x_{1}, y_{1}), ... , (x_{m}, y_{m})$ among $\{(u_{1}, v_{1}), ... , (u_{n}, v_{n})\}$ where both $x_{i}$ and $y_{i}$ are variables in $\phi$ while still having $\dash_{\T}\phi\rightarrow\bigvee_{i=1}^{m}x_{i}=y_{i}$; we will keep on denoting the sort of $x_{j}$ and $y_{j}$ by $\s_{j}$, for simplicity.

Now, assuming for contradiction that $\not\tdash \phi\rightarrow(x_{i}=y_{i})$ for all $1\leq i\leq m$, there are $\T$-interpretations $\B_{i}$ which satisfy $\phi$ but not $x_{i}=y_{i}$; without loss of generality, by the stable-infiniteness of $\T$, we may assume that all $\B_{i}$ are infinite on all their sorts, and by \Cref{LowenheimSkolemDownwards} we may assume as well that $\B_{1}$ through $\B_{m}$ have the same countable domain (say $\mathbb{N}$) for all sorts. That means, since we have no function or predicate symbols, that all interpretations $\B_{i}$ are over the same model $\B$ of $\T$ that is countable in each domain.

Now, for any $\T$-interpretation $\B_{\lambda}$ on $\B$, we define an equivalence relation $E^{\s}_{\lambda}$ on the variables of $\phi$ of sort $\s$ by making $xE^{\s}_{\lambda}y$ iff $x^{\B_{\lambda}}=y^{\B_{\lambda}}$; we also define an equivalence relation $xE^{\s}y$ iff $xE^{\s}_{\lambda}y$ for all $\T$-interpretations $\B_{\lambda}$ that satisfy $\phi$. Because of the interpretations $\B_{i}$, we have that $x_{i}\overline{E^{\s_{i}}}y_{i}$ for each $i$, 
where $\overline{E^{\s_{i}}}$ is the complement of the relation $E^{\s_{i}}$. Now, we state that it is possible to define an interpretation $\B^{\prime}$ on $\B$ such that $xE^{\s}y$ if, and only if, $x^{\B^{\prime}}=y^{\B^{\prime}}$, for variables $x$ and $y$ in $\phi$: in addition, $\B^{\prime}$ is a $\T$-interpretation that satisfies $\phi$, while not satisfying $\bigvee_{i=1}^{m}x_{i}=y_{i}$, what leads to a contradiction.

It is rather easy to define $\B^{\prime}$: it only needs to map all variables in an equivalence class of $E^{\s}$ to the same element, while mapping variables in a different equivalence class to a different element; this is clearly possible since $\B$ has countably many elements in each domain. Of course, then $x^{\B^{\prime}}=y^{\B^{\prime}}$ iff $xE^{\s}y$. Furthermore, since we are over the empty signature, $\phi$ is a conjunction of equalities and disequalities: if $x$ and $y$ are of sort $\s$ and $x=y$ (respectively $\neg(x=y)$) is one of the literals of $\phi$, then for any $\T$-interpretation $\B_{\lambda}$ that satisfies $\phi$, we must have that it also satisfies $x=y$ (respectively $\neg(x=y)$), and therefore $xE^{\s}y$ (respectively $x\overline{E^{\s}}y$); this means $x^{\B^{\prime}}=y^{\B^{\prime}}$ ($x^{\B^{\prime}}\neq y^{\B^{\prime}}$), and so $\B^{\prime}$ indeed satisfies $\phi$.

Finally, one lands at a contradiction: since $\B^{\prime}$ does not satisfy $x_{i}=y_{i}$ for any $1\leq i\leq m$, it cannot possibly satisfy $\bigvee_{i=1}^{m}x_{i}=y_{i}$, although it does satisfy $\phi$. The conclusion must be that, for some $i$ between $1$ and $m$, $\tdash \phi\rightarrow x_{i}=y_{i}$.
\end{proof}

\section{Proof of \Cref{OS+ES+-SI=>FW}}

\begin{restatable}{lemma}{IMSI}
\label{Infinite Model => SI}
If $\Sigma$ is a signature without function or predicate symbols, $S$ is a set of sorts, and $\T$ is a theory with a model $\A$ where all domains are infinite, $\T$ is stably-infinite w.r.t. $S$.
\end{restatable}

\begin{proof}
Let $\phi$ be a quantifier-free $\Sigma$-formula and $\B$ a $\T$-interpretation that satisfies $\phi$. For each $\s\in S$, $\vars_{\s}(\phi)$ is finite, and so is $\vars_{\s}(\phi)^{\B}$, meaning there is a subset $C(\s)$ of $\s^{\A}$ with the same cardinality as $\vars_{\s}(\phi)^{\B}$; let $h_{\s}:\vars_{\s}(\phi)^{\B}\rightarrow C(\s)$ be bijections for each $\s\in S$.

We define an interpretation $\A^{\prime}$ on $\A$ such that, for every $x\in\vars_{\s}(\phi)$, $x^{\A^{\prime}}=h_{\s}(x^{\B})$, and for every variable $x$ not in $\phi$ we may define $x^{\A^{\prime}}$ arbitrarily. Now, let $x=y$ be an atomic subformula of $\phi$, with $x$ and $y$ of sort $\s$: since $h_{\s}$ is a bijection, $x^{\A^{\prime}}=y^{\A^{\prime}}$ iff $x^{\B}=y^{\B}$; since all atomic subformulas of $\phi$ receive precisely the same truth-value in either $\A^{\prime}$ or $\B$, and a quantifier-free formula's truth value is entirely determined by the truth-values of its atomic subformulas, we get that $\A^{\prime}$ satisfies $\phi$, and so $\T$ is stably-infinite w.r.t. $S$.   
\end{proof}

\begin{restatable}{lemma}{simaxmod}
\label{SI=>Maximum model}
If a one-sorted theory over the empty signature is not stably-infinite, it has a model of maximum finite cardinality.
\end{restatable}

\begin{proof}
Let $\T$ be a one-sorted, not stably-infinite theory over the empty signature; $\T$ cannot have infinite models because of \Cref{Infinite Model => SI}; 
and it has finite models since, otherwise, it would be vacuously stably-infinite. So, suppose that $\T$ has models of arbitrarily finite size: since 
$\Gamma=\{\psi_{\geq n} : n\in\mathbb{N}\}$ 
is only satisfied by structures with an 
infinite domain, we have that $\ax(\T)\cup\Gamma$ is unsatisfiable. 
By \Cref{Compactness}, there must exist finite subsets $\ax_{0}\subseteq \ax(\T)$ and $\Gamma_{0}\subseteq \Gamma$ such that $\ax_{0}\cup \Gamma_{0}$ is 
unsatisfiable. Let $N$ be the largest index of a formula $\psi_{\geq n}$ showing up in $\Gamma_{0}$, and we can derive from the fact that $\psi_{\geq j}$ 
implies $\psi_{\geq i}$ for $j>i$ that $\ax_{0}\cup\{\psi_{\geq N}\}$ is unsatisfiable, meaning that there are no structures that satisfy $\ax_{0}$ with 
more than $N$ elements in their domains. But models of $\ax(\T)$ are also models of $\ax_{0}$, what implies that $\ax(\T)$ has no models with more than 
$N$ elements, contradicting our hypothesis. 
\end{proof}

\osessifw*

\begin{proof}
By \Cref{SI=>Maximum model}, a one-sorted, not stably-infinite theory over the empty signature must have a model of maximum size. Let $M$ be that finite cardinality; we can then prove that, for a quantifier-free formula $\phi$ and fresh variables $x_{1}$ through $x_{M}$, 
$\wit(\phi)=\phi\wedge \bigwedge_{i=1}^{M}x_{i}=x_{i}$
is a witness. 
We begin by noticing that $\exists\,\overarrow{x}.\:\wit(\phi)$ and $\phi$ are $\T$-equivalent, where $\overarrow{x}=\vars(\wit(\phi))\setminus\vars(\phi)$. Now, assume that the $\T$-interpretation $\A$ satisfies $\wit(\phi)$: since the maximum size of a model of $\T$ is $M$, we have $|\s^{\A}|\leq M$. 
We define another $T$-interpretation $\A'$ 
by changing the value of $\A$ only on the variables of $\overarrow{x}$ so that $x_{i}\mapsto x_{i}^{\A^{\prime}}$ is surjective. $\wit(\phi)$ remains valid in $\A'$ and, in addition, $\vars(\wit(\phi))^{\A^{\prime}}=\s^{\A^{\prime}}$.
\end{proof}

\section{Proof of \Cref{SI+SFW=S}}
\sisfweqs*

\begin{proof}
Suppose that $\T$ is stably-infinite and strongly finitely witnessable w.r.t. its only sort $\sigma$; let $\A$ be a $\T$-interpretation which satisfies a quantifier-free formula $\phi$, and take a cardinal $\kappa(\sigma)\geq |\sigma^{\A}|$.

If $\kappa(\sigma)$ is infinite: since $\T$ is assumed to be stably-infinite, there exists a $\T$-interpretation $\B$ that satisfies $\phi$ with $\sigma^{\B}$ infinite; then, by Lowenheim-Skolem's theorem for one-sorted logic (upwards if $\kappa(\sigma)>|\sigma^{\B}|$, downwards if $\kappa(\sigma)<|\sigma^{\B}|$), there exists a $\T$-interpretation $\C$ with $|\sigma^{\C}|=\kappa(\sigma)$ that satisfies $\varphi$.

So assume $\kappa(\sigma)$ is finite, and call it $m$: of course, in this case, $|\sigma^{\A}|$ must also be finite. Let $\wit$ be the strong witness of $\T$, and we know that, since $\phi$ and $\exists\,\overarrow{x}.\:\wit(\phi)$ are $\T$-equivalent (for $\overarrow{x}=\vars(\wit(\phi))\setminus\vars(\phi)$), there exists a $\T$-interpretation $\A^{\prime}$, with same underlying structure as that of $\A$, that satisfies $\wit(\phi)$. Let $W=\vars(\wit(\phi))$, with the equivalence relation $F$ over $W$ such that $xFy$ iff $x^{\A^{\prime}}=y^{\A^{\prime}}$ (and the corresponding arrangement 
being $\delta_{W}$), and let $n$ be the number of elements in $W^{\A^{\prime}}$, that is, the 
number of equivalence classes in $W/F$. 

If $m=n$, we have that $m=\kappa(\s)\geq |\s^{\A}|=|\s^{\A^{\prime}}|\geq |W^{\A^{\prime}}|=n$, meaning $W^{\A^{\prime}}$ actually equals $\s^{\A^{\prime}}$ and therefore $\A^{\prime}$ is already a $\T$-interpretation that satisfies $\wit(\phi)$ with $\s^{\A^{\prime}}=\vars_{\s}(\wit(\phi))^{\A^{\prime}}$; otherwise, take a set $U$ of $m-n$ 
fresh variables, and define the equivalence relation $E$ on $V=W\cup U$ such that $xEy$ iff $xFy$ or 
$x=y$, which then has $m$ equivalence classes, with corresponding arrangement $\delta_{V}$ on $V$. 

The formula $\wit(\phi)\wedge\delta_{V}$ is satisfied by an infinite $\T$-interpretation: to see that, notice that since $\wit(\phi)\wedge\delta_{W}$ is quantifier-free and satisfied by $\A^{\prime}$, and $\T$ is stably-infinite, there must exist an infinite $\T$-interpretation $\B$ that satisfies $\wit(\phi)\wedge\delta_{W}$; since $\wit(\phi)\wedge\delta_{V}$ is equivalent to
\[\wit(\phi)\wedge\delta_{W}\wedge\bigwedge_{y\in U}\bigwedge_{x\in V\setminus\{y\}}\neg(x=y),\]
and since $\B$ is infinite and the variables in $U$ are fresh, we can change the value of $\B$ only on $U$, 
thus creating a $T$-interpretation $\B'$ that satisfies
$\wit(\phi)\wedge\delta_{V}$.

Since the quantifier-free formula $\wit(\phi)\wedge\delta_{V}$ is $\T$-satisfiable, there must exist a $\T$-interpretation $\C$ that satisfies $\wit(\phi)\wedge\delta_{V}$ with $\sigma^{\C}=V^{\C}$. Since $V/E$ has $m$ equivalence classes, from the fact that $\C$ satisfies $\delta_{V}$ it follows that $|\sigma^{\C}|=m$; and since $\C$ satisfies $\wit(\phi)$, it also satisfies $\exists\, \overarrow{x}.\: \wit(\phi)$ (for $\overarrow{x}=\vars(\wit(\phi))\setminus\vars(\phi)$), and given that the latter formula and $\phi$ are $\T$-equivalent, we get that $\C$ satisfies $\phi$, and so $\T$ is smooth.
\end{proof}

\section{Proof of \Cref{OS+ES+-SI+-SFW=>-C}}
\osessisfwcc*

\begin{proof}
Assume that $\T$ is one-sorted and convex without being strongly finitely witnessable nor being stably-infinite, over the empty signature. By \Cref{SI=>Maximum model}, we know that $\T$ must have a model of maximum finite cardinality, say $M$. We can now state that $M>1$, meaning the maximum model of $\T$ has more than one element, since otherwise $\T$ would consist only of the trivial model, and one can easily prove in that case that $\wit(\phi)=\phi$ is a strong witness: if $\phi\wedge\delta_{V}$ is $\T$-satisfiable, it is satisfied in the $\T$-interpretation $\A$ with only one element in its domain, and of course $\vars(\phi\wedge\delta_{V})^{\A}=\s^{\A}$.

But now we can reach a contradiction: suppose $\phi$ is a conjunction of literals that is tautological, such as $x=x$; we have that 
$\dash_{\T}\phi\rightarrow\bigvee_{i=1}^{M+1}x=y_{i}$
since, by the pigeonhole principle, the disjunction on the right of the implication must always be true in $\T$-interpretations. But we do not have $\dash_{\T}\phi\rightarrow x=y_{i}$ for any $1\leq i\leq M+1$, since $\phi$ is a tautology and there are $\T$-interpretations on a structure with $M>1$ elements. Therefore $\T$ is not convex, against our hypothesis that it in fact was.
\end{proof}

\section{Proof of \Cref{lem:addsort}}

\begin{definition}
    A witness $wit$ is called {\em variable-dependent} if 
    there is a function $\chi$ from sets of variables to formulas such that
    for every formula $\phi$, $wit(\phi)=\phi\wedge \chi(\vars(\phi))$.
\end{definition}

\begin{restatable}{lemma}{witarevariabledependent}
\label{lem:witarevariabledependent}
Every theory $\T$ defined over an empty signature $\Sigma$ that is finitely witnessable (respectively strongly finitely witnessable) w.r.t. $S\subseteq \S_{\Sigma}$, has a witness (strong witness) that is variable-dependent.
\end{restatable}
\color{black}

\begin{proof}
Take a quantifier-free formula $\phi$. Let $V=\vars(\phi)$, take the set $E(V)$ of equivalence relations on $V$, and for a $E\in E(V)$, let $\delta_{V}^{E}$ be the arrangement induced by $E$ on $V$; we then define
\[\chi(V)=\bigvee_{E\in E(V)}\wit(\delta_{V}^{E})\quad\text{and}\quad \wit_{0}(\phi)=\phi\wedge\chi(V),\] 
and we state that, if $\wit$ is a witness, respectively a strong witness, $\wit_{0}$ is also a witness, respectively a strong witness. 
\begin{enumerate}
\item We start by showing that, if $\overarrow{x}=\vars(\wit_{0}(\phi))\setminus\vars(\phi)$, $\phi$ and $\exists\,\overarrow{x}.\:\wit_{0}(\phi)$ are $\T$-equivalent; of course, if $\exists\,\overarrow{x}.\:\wit_{0}(\phi)=\phi\wedge\exists\,\overarrow{x}.\:\chi(V)$ (what we can do since $\phi$ contain none of the variables in $\overarrow{x}$) is satisfied by a $\T$-interpretation $\A$, so is $\phi$, so let us focus on the other direction instead.

So assume the $\T$-interpretation $\A$ satisfies $\phi$, and let $E_{0}$ be the equivalence on $V$ such that $xE_{0}y$ iff $x^{\A}=y^{\A}$, meaning $\A$ satisfies $\delta_{V}^{E_{0}}$. Then $\A$ must also satisfy $\exists\,\overarrow{y}.\:\wit(\delta_{V}^{E_{0}})$, for $\overarrow{y}=\vars(\wit(\delta_{V}^{E_{0}})\setminus\vars(\phi)$, and thus there exists a $\T$-interpretation $\A^{\prime}$, differing from $\A$ at most on $\overarrow{y}$, that satisfies $\wit(\delta_{V}^{E_{0}})$ (and thus $\exists\,\overarrow{y}.\:\wit(\delta_{V}^{E_{0}})$ and $\delta_{V}^{E_{0}}$ as well).

Because we are in an empty signature, all atomic subformulas of $\phi$ are equalities of variables $x=y$: and this formula is satisfied in $\A^{\prime}$ iff $xE_{0}y$, what happens in turn iff $x^{\A}=y^{\A}$; this means the atomic subformulas of $\phi$ receive the same truth-values in $\A$ and $\A^{\prime}$, and since $\phi$ is quantifier-free we get that $\A^{\prime}$ satisfies $\phi$. Of course, $\A^{\prime}$ also satisfies $\wit(\delta_{V}^{E_{0}})$, and thus satisfies $\phi\wedge\chi(V)$, meaning $\A$ satisfies $\exists\,\overarrow{y}.\:\phi\wedge\chi(V)$; since the variables in $\overarrow{y}$ are a subset of the variables in $\overarrow{x}$, we get $\A$ satisfies $\exists\,\overarrow{x}.\:\phi\wedge\chi(V)$. Thus $\phi$ and $\exists\,\overarrow{x}.\:\wit_{0}(\phi)$ are indeed $\T$-equivalent.

\item \begin{enumerate}
\item Suppose that $\wit$ is a witness, and let $\A$ be a $\T$-interpretation that satisfies $\wit_{0}(\phi)$. Let $E_{0}$ be
the equivalence on $V$ such that $xE_{0}y$ iff $x^{\A}=y^{\A}$, and we know that $\A$ also satisfies $\delta_{V}^{E_{0}}$, and thus 
$\exists\,\overarrow{y}.\:\wit(\delta_{V}^{E_{0}})$, for 
$\overarrow{y}=\vars(\wit(\delta_{V}^{E_{0}}))\setminus\vars(\delta_{V}^{E_{0}})$ (contained in $\overarrow{x}=\vars(\wit_{0}(\phi))\setminus\vars(\phi)$). Changing the value of $\A$ at most on the variables $\overarrow{y}$, we obtain a second $\T$-interpretation $\A^{\prime}$ that 
satisfies $\wit(\delta_{V}^{E_{0}})$; but of course $\A^{\prime}$ also satisfies $\exists\,\overarrow{y}.\:\wit(\delta_{V}^{E_{0}})$ and thus $\delta_{V}^{E_{0}}$. 

Now, since $\wit$ is a witness, there must exist a $\T$-interpretation $\B$ that satisfies $\wit(\delta_{V}^{E_{0}})$ with

\[\s^{\B}=\vars_{\s}(\wit(\delta_{V}^{E_{0}}))^{\B}\]
for each 
$\s\in S$. But, because $\B$ satisfies $\wit(\delta_{V}^{E_{0}})$, it satisfies $\exists\,\overarrow{y}.\:\wit(\delta_{V}^{E_{0}})$ and thus $\delta_{V}^{E_{0}}$; again,
since we are on the empty signature, this means the atomic subformulas of $\phi$ receive the same truth-value in either $\A$ or $\B$, and so $\B$ satisfies $\phi$ (and, since $\B$ also satisfies $\wit(\delta_{V}^{E_{0}})$ and thus $\bigvee_{E\in E(V)}\wit(\delta_{V}^{E})$, $\B$ then satisfies $\phi\wedge\chi(V)$). 
Furthermore, $\vars_{\s}(\wit(\delta_{V}^{E_{0}}))\subseteq \vars_{\s}(\wit_{0}(\phi))$, meaning
$\s^{\B}=\vars_{\s}(\wit_{0}(\phi))^{\B}$ and, therefore, $\wit_{0}$ is also a witness.

\item Suppose now $\wit$ is a strong witness, $U$ is a set of variables, $F$ is an equivalence on $U$, $\delta^{F}_{U}$ is the corresponding arrangement, and $\A$ is a $\T$-interpretation that satisfies $\wit_{0}(\phi)\wedge\delta^{F}_{U}$. 

Let $E_{0}$ be the equivalence relation on $V$ such that $xE_{0}y$ iff $x^{\A}=y^{\A}$: this means that $\A$ satisfies $\delta_{V}^{E_{0}}$, and thus 
$\exists\,\overarrow{y}.\:\wit(\delta_{V}^{E_{0}})$, for $\overarrow{y}=\vars(\wit(\delta_{V}^{E_{0}}))\setminus\vars(\delta_{V}^{E_{0}})$; as we did in the case that $\wit$ was only a witness, we have that there is a $\T$-interpretation $\A^{\prime}$, 
differing from $\A$ at most on $\overarrow{y}$, that satisfies $\wit(\delta_{V}^{E_{0}})$ (and, of course, $\A^{\prime}$ also satisfies $\exists\,\overarrow{y}.\:\wit(\delta_{V}^{E_{0}})$ and therefore $\delta_{V}^{E_{0}}$).

Define then the set of variables  $W=U\cup V$ with the equivalence $G$ such that $xGy$ iff $x^{\A}=y^{\A}$, and notice that $\A$ satisfies $\delta_{W}^{G}$, a formula that implies both $\delta_{V}^{E_{0}}$ and $\delta_{U}^{F}$. 

Since $\A$ satisfies $\wit(\delta_{V}^{E_{0}})\wedge\delta_{W}^{G}$, and $\wit$ is a strong witness, there must exist a $\T$-interpretation $\B$ that satisfies $\wit(\delta_{V}^{E_{0}})\wedge\delta_{W}^{G}$ with 
\[\s^{\B}=\vars_{\s}(\wit(\delta_{V}^{E_{0}})\wedge\delta_{W}^{G})^{\B}\]
for each $\s\in S$. 

Because $\B$ satisfies $\delta_{W}^{G}$, it must satisfy $\delta_{V}^{E_{0}}$, meaning $\B$ and $\A$ satisfy precisely the same atomic subformulas of $\phi$, and thus $\B$ satisfies $\phi$ (since that is a quantifier-free formula), and 
consequently $\wit_{0}(\phi)$ (since $\B$ also satisfies $\wit(\delta_{V}^{E_{0}})$ and thus $\chi(V)$); furthermore, because $\B$ satisfies $\delta_{W}^{G}$, it also satisfies $\delta_{U}^{F}$. Finally, 
notice the variables of sort $\s$ of $\wit(\delta_{V}^{E_{0}})\wedge\delta_{W}^{G}$ are a subset of the set of variables of sort $\s$ of $\wit_{0}(\phi)\wedge\delta_{U}^{F}$: first of all,
because $\wit(\delta_{V}^{E_{0}})$ is a subformula of $\chi(V)$ and thus $\wit_{0}(\phi)$; and second, because the variables of $\delta_{W}^{G}$ ($W=U\cup V$) are also
variables of $\delta_{U}^{F}$ and $\phi$ (and thus $\wit_{0}(\phi)$). So we have that $\s^{\B}=\vars_{\s}(\wit_{0}(\phi)\wedge\delta_{U}^{F})^{\B}$, proving $\wit$ is indeed a strong witness.\qedhere
\end{enumerate}
\end{enumerate}
\end{proof}

For simplicity of notation, the symbols $\A$ and $\B$ will be reserved, in the proofs of Lemmas \ref{technical result on adding a sort} and \ref{lem:addsort}, for $\Sigma_{1}$-interpretations, while $\C$ and $\D$ will denote $\Sigma_{2}$-interpretations.
Still in the following results, given a $\Sigma_{2}$-interpretation $\C$, we denote by $\subs{\C}$ the $\Sigma_{1}$-interpretation with $\s^{\subs{\C}}=\s^{\C}$, and $x^{\subs{\C}}=x^{\A}$ for every variable $x$ of sort $\s$. 

\begin{lemma}\label{technical result on adding a sort}
Take a $\Sigma_{2}$-interpretation $\C$. It is then true that, for any $\Sigma_{1}$-formula $\varphi$, $\subs{\C}$ satisfies $\varphi$ if, and only if, $\C$ satisfies $\varphi$.
\end{lemma}

\begin{proof}
We prove that $\subs{\C}$ satisfies $\varphi$ iff $\C$ satisfies $\varphi$ by induction on the complexity of $\varphi$.
\begin{enumerate}
\item Suppose $\varphi$ is $x=y$, for $x$ and $y$ variables of sort $\s$; then, since $x^{\subs{\C}}=x^{\C}$ and $y^{\subs{\C}}=y^{\C}$, $\subs{\C}$ satisfies $\varphi$ iff $\C$ does so.
\item Suppose $\varphi=\neg\psi$ or $\varphi=\psi\vee\xi$, the proof for the connectives $\wedge$ and $\rightarrow$ being analogous. In the first case, $\subs{\C}$ satisfies $\varphi$ iff it does not satisfy $\psi$, what happens by induction hypothesis iff $\C$ does not satisfy $\psi$, equivalent to $\C$ satisfying $\varphi$.

In the second case, $\subs{\C}$ satisfies $\varphi$ iff it satisfies either $\psi$ or $\xi$, what happens by induction hypothesis iff $\C$ satisfies either $\psi$ or $\xi$; of course, this is equivalent to $\C$ satisfying $\varphi$.
\item Finally, suppose $\varphi=\exists\, x.\:\psi$, the case for the quantifier $\forall$ being very similar. $\subs{\C}$ satisfies $\varphi$ iff there exists a 
second $\Sigma_{1}$-interpretation $\A$, differing from $\subs{\C}$ at most on $x$, such that $\A$ satisfies $\psi$. Taking the $\Sigma_{2}$-interpretation $\D$ that differs from $\C$ at most on $x$, where $x^{\D}=x^{\A}$, one has $\subs{\D}=\A$, and so it is clear by induction hypothesis that 
$\A$ satisfies $\psi$ iff $\D$ satisfies $\psi$. This is equivalent to the fact that $\C$ satisfies $\varphi$.
\end{enumerate}
\end{proof}

It will be also important, in what is to come, to consider, given a quantifier-free $\Sigma_{2}$-formula $\phi$ and an 
interpretation $\C$ that satisfies it, a quantifier-free $\Sigma_{1}$-formula $\subsf{\phi}{\C}$ obtained by: replacing each 
atomic subformula $u=v$ on $\phi$, where both $u$ and $v$ are of sort $\s_{2}$, by 
either a tautology in $\Sigma_{1}$ (such as $x=x$, for $x$ of sort $\s$ already in $\phi$), if $u^{\C}=v^{\C}$, or a 
contradiction in $\Sigma_{1}$ (such as $\neg(x=x)$), if $u^{\C}\neq v^{\C}$; notice that $\subs{\C}$ satisfies $\subsf{\phi}{\C}$ iff $\C$ satisfies $\phi$.

\lemaddsort*

\begin{proof}
As expected, we must use \Cref{technical result on adding a sort} again and again. 
\begin{enumerate}
\item If $\T$ is stably-infinite, take a quantifier-free $\Sigma_{2}$-formula $\phi$ and a $\Tadds$-interpretation $\C$ that satisfies $\phi$; we then have that $\subs{\C}$ is a $\T$-interpretation 
(since, if $\C$ satisfies $\psi$, for a $\psi$ in $\ax(T)$,
$\subs{\C}$ certainly satisfies $\psi$ as well) that satisfies $\subsf{\phi}{\C}$. Since $\T$ is assumed to be stably-infinite, 
there exists an infinite $\T$-interpretation $\A$ that satisfies $\subsf{\phi}{\C}$. By then picking a $\Tadds$-interpretation $\D$ such that: 
$\subs{\D}=\A$; $|\s_{2}^{\D}|\geq\omega$; and $u^{\D}=v^{\D}$ iff $u^{\C}=v^{\C}$ for variables $u$ and $v$ of sort $\s_{2}$, we get that $\D$ is infinite in both domains and satisfies $\phi$.

Reciprocally, suppose now that $\Tadds$ is stably-infinite, and then take a quantifier-free $\Sigma_{1}$-formula 
$\phi$ and a $\T$-interpretation $\A$ that satisfies $\phi$; it follows that any $\Tadds$-interpretation $\C$ with 
$\subs{\C}=\A$ must satisfy $\phi$, and from the fact that $\Tadds$ is stably-infinite, there exists a $\Tadds$-interpretation $\D$, infinite on both 
domains, that satisfies $\phi$. Then the $\T$-interpretation $\subs{\D}$ has an infinite domain and satisfies $\phi$, proving $\T$ is stably-infinite.

\item Start assuming $\T$ is smooth; then take a quantifier-free $\Sigma_{2}$-formula $\phi$, a $\Tadds$-interpretation 
$\C$ that satisfies $\phi$, and cardinals $\kappa(\s)\geq |\s^{\C}|$ and $\kappa(\s_{2})\geq |\s_{2}^{\C}|$. We 
know that $\subs{\C}$ is a $\T$-interpretation that satisfies $\subsf{\phi}{\C}$, and then there exists a $\T$-interpretation $\A$ that satisfies 
$\subsf{\phi}{\C}$ with $|\s^{\A}|=\kappa(\s)$. Taking a $\Tadds$-interpretation $\D$ with: $\subs{\D}=\A$; 
$|\s_{2}^{\D}|=\kappa(\s_{2})$; and $u^{\D}=v^{\D}$, for all variables $u$ and $v$ of sort $\s_{2}$, iff $u^{\C}=v^{\C}$, we get that $\D$ 
satisfies $\phi$, and $|\s^{\D}|=\kappa(\s)$ and $|\s_{2}^{\D}|=\kappa(\s_{2})$.

Reciprocally, if $\Tadds$ is smooth, take a quantifier-free $\Sigma_{1}$-formula $\phi$, a $\T$-interpretation $\A$ that 
satisfies $\phi$, and a cardinal $\kappa\geq |\s^{\A}|$. Any $\Tadds$-
interpretation $\C$ with $\subs{\C}=\A$ must satisfy $\phi$, and then there must exist a $\Tadds$-interpretation $\D$, since 
this theory is smooth, with $|\s^{\D}|=\kappa$, $|\s_{2}^{\D}|=|\s_{2}^{\C}|$ and that 
satisfies $\phi$. Of course $\subs{\D}$ is then a $\T$-interpretation with $|\s^{\subs{\D}}|=\kappa$ and that satisfies $\phi$.

\item
Suppose $\Tadds$ is finitely witnessable, with witness $\wit$, and we proceed to prove $\T$ is also finitely witnessable, with witness 
\[\wit_{1}(\phi)=\wit(\phi)[u_{1}/x,\ldots u_{n}/x],\]
where $\vars_{\s_{2}}(\wit(\phi))=\{u_{1}, \ldots, u_{n}\}$, and $x$ is any variable in $\phi$ (say, the first, so that it is unique). Start by picking a 
$\Sigma_{1}$-formula $\phi$ and a $\T$-interpretation $\A$ that satisfies $\phi$; we have that the $\Tadds$-interpretation $\C$, with $\subs{\C}=\A$ and 
$|\s_{2}^{\C}|=1$, satisfies $\phi$. Since $\wit$ is supposed to be a witness for $\Tadds$, $\C$ must also satisfy $\exists\,\overarrow{x}.\:\wit(\phi)$, for 
$\overarrow{x}=\vars(\wit(\phi))\setminus\vars(\phi)$; then there is an interpretation $\C^{\prime}$, differing from $\C$ at most on $\overarrow{x}$, that satisfies $\wit(\phi)$, and thus $\subs{\C^{\prime}}$ satisfies
$\wit_{1}(\phi)$ (given that, for any two $u,v\in \vars_{\s_{2}}(\wit(\phi))$, $u^{\C}=v^{\C}$). Since $\subs{\C^{\prime}}$ differs from $\A$ at most on the 
values given to $\overarrow{x}$, $\A$ satisfies $\exists\,\overarrow{x}.\:\wit_{1}(\phi)$. Reciprocally, if $\A$ is a $\T$-interpretation that satisfies 
$\exists\,\overarrow{x}.\:\wit_{1}(\phi)$, the same $\C$ as above satisfies $\exists\,\overarrow{x}.\:\wit(\phi)$ and thus $\phi$, meaning $\subs{\C}=\A$ satisfies $\phi$. With both cases, we have proved $\phi$ and 
$\exists\,\overarrow{x}.\:\wit_{1}(\phi)$ are $\T$-equivalent.

Now, suppose $\A$ is a $\T$-interpretation that satisfies $\wit_{1}(\phi)$, and so the $\Tadds$-interpretation $\C$ with $\subs{\C}=\A$ and $|\s_{2}^{\C}|=1$ satisfies $\wit(\phi)$. There must exist then a $\Tadds$-interpretation $\D$ that satisfies $\wit(\phi)$, and thus $\phi$, with $\s^{\D}=\vars_{\s}(\wit(\phi))^{\D}$ and $\s_{2}^{\D}=\vars_{\s_{2}}(\wit(\phi))^{\D}$; then $\subs{\D}$ is a $\T$-interpretation that satisfies $\phi$ with $\s^{\subs{\D}}=\vars_{\s}(\wit(\phi))^{\subs{\D}}$. Changing the value given by $\subs{\D}$ to the variables on $\overarrow{x}$ so $\subs{\D}^{\prime}$ satisfies $\wit_{1}(\phi)$, from the fact that $\vars_{\s}(\wit(\phi))=\vars_{\s}(\wit_{1}(\phi))$ we obtain that $\T$ is finitely witnessable.

Reciprocally, assume $\T$ is finitely witnessable, and thanks to \Cref{lem:witarevariabledependent} we may take a witness (here, $\phi$ is a quantifier-free $\Sigma_{1}$-formula) $\wit(\phi)=\phi\wedge\psi(\vars(\phi))$, where $\psi(\vars(\phi))$ is a formula that depends only on $\vars(\phi)$; we wish to prove that then $\Tadds$ is also finitely witnessable, with witness (and here, $\phi$ is a quantifier-free $\Sigma_{2}$-formula)
\[\wit_{2}(\phi)=\phi\wedge\psi(\vars_{\s}(\phi))\wedge (u=u),\]
with $u$ a variable of sort $\s_{2}$. Start by taking a quantifier-free $\Sigma_{2}$-formula $\phi$ and a $\Tadds$-interpretation $\C$ that satisfies $\phi$, and we have that $\subs{\C}$ satisfies $\subsf{\phi}{\C}$, and thus $\exists\,\overarrow{x}.\:\wit(\subsf{\phi}{\C})=\subsf{\phi}{\C}\wedge\exists\,\overarrow{x}.\:\psi(\vars(\subsf{\phi}{\C}))$, for 
\[\overarrow{x}=\vars(\wit(\subsf{\phi}{\C}))\setminus\vars(\subsf{\phi}{\C})=\vars(\wit_{2}(\phi))\setminus\big(\vars(\phi)\cup\{u\}\big),\]
since $\vars(\subsf{\phi}{\C})=\vars_{\s}(\phi)$. There must exist an interpretation $\A$, differing from $\subs{\C}$ at most on $\overarrow{x}$, that satisfies $\subsf{\phi}{\C}\wedge\psi(\vars(\subsf{\phi}{\C}))$; take then the interpretation $\D$, differing from $\C$ at most on $x$ in $\overarrow{x}$, where $x^{\D}=x^{\A}$, and that satisfies $\phi\wedge\psi(\vars(\subsf{\phi}{\C}))\wedge u=u$. Since $\D$ differs from $\C$ at most on the variables of $\overarrow{x}$, we are done with one direction of the biconditional. But, if we assume the $\Tadds$-interpretation $\C$ satisfies $\exists\,\overarrow{x}.\:\wit_{2}(\phi)$ (i.e., $\phi\wedge(u=u)\wedge\exists\,\overarrow{x}.\:\psi(\vars_{\s}(\phi))$), we already have it satisfies $\phi$, and so $\phi$ and $\exists\,\overarrow{x}.\:\wit_{2}(\phi)$ are $\Tadds$-equivalent.

Now, if the $\Tadds$-interpretation $\C$ satisfies $\wit_{2}(\phi)$, $\subs{\C}$ satisfies $\subsf{\phi}{\C}\wedge\psi(\vars_{\s}(\phi))=\wit(\subsf{\phi}{\C})$, and thus there exists a $\T$-interpretation $\A$ that satisfies $\subsf{\phi}{\C}\wedge\psi(\vars_{\s}(\phi))$ with $\s^{\A}=\vars_{\s}(\subsf{\phi}{\C})^{\A}=\vars_{\s}(\phi)^{\A}$. Taking the $\Tadds$-interpretation $\D$ with $\subs{\D}=\A$ and $|\s_{2}^{\D}|=1$, we see that not only $\D$ satisfies $\phi\wedge\psi(\vars_{\s}(\phi))\wedge (u=u)=\wit_{2}(\phi)$, but also has the property that $\s^{\D}=\vars_{\s}(\wit_{2}(\phi))^{\D}$ and $\s_{2}^{\D}=\vars_{\s_{2}}(\wit_{2}(\phi))^{\D}$, since $\s_{2}^{\D}$ contains only one element and $\vars_{\s_{2}}(\wit_{2}(\phi))^{\D}$ contains at least $u^{\D}$.

\item Suppose $\Tadds$ is now strongly finitely witnessable, with strong witness $\wit$, and we now proceed to show that 
\[\wit_{1}(\phi)=\wit(\phi)[u_{1}/x,\ldots u_{n}/x],\]
where $\vars_{\s_{2}}(\wit(\phi))=\{u_{1}, \ldots, u_{n}\}$, and $x$ is any variable in $\phi$, is a strong witness for $\T$. We know that $\phi$ and $\exists\,\overarrow{x}.\:\wit_{1}(\phi)$, for $\overarrow{x}=\vars_{\s}(\wit_{1}(\phi))\setminus\vars_{\s}(\phi)$, are $\T$-equivalent from our proof that, if $\Tadds$ is finitely witnessable, so is $\T$. So take a set of variables (of sort $\s$) $V$, an arrangement $\delta_{V}$ on $V$ and a $\T$-interpretation $\A$ that satisfies $\wit_{1}(\phi)\wedge\delta_{V}$; the $\Tadds$-interpretation $\C$ with $\subs{\C}=\A$ and $|\s_{2}^{\C}|=1$ then satisfies $\wit(\phi)\wedge\delta_{U}$, where $\delta_{U}$ is the arrangement on $U=V\cup\vars_{\s_{2}}(\phi)$ extending $\delta_{V}$ such that all $u_{i}$ are equal. Given $\wit$ is a strong witness for $\Tadds$, there exists a $\Tadds$-interpretation $\D$ that satisfies $\wit(\phi)\wedge\delta_{U}$ with $\s^{\D}=\vars_{\s}(\wit(\phi)\wedge\delta_{U})^{\D}$ and $\s_{2}^{\D}=\vars_{\s_{2}}(\wit(\phi)\wedge\delta_{U})^{\D}$. Of course we must then have $|\s_{2}^{\D}|=1$, and then $\subs{\D}$ is a $\T$-interpretation that satisfies $\wit_{1}(\phi)\wedge\delta_{V}$ with $\s^{\subs{\D}}=\vars_{\s}(\wit_{1}(\phi)\wedge\delta_{V})^{\D}$.

Reciprocally, assume $\T$ is strongly finitely witnessable, with a strong witness  $\wit(\phi)=\phi\wedge\psi(\vars(\phi))$, where $\psi(\vars(\phi))$ is a formula that depends only on $\vars(\phi)$, and we shall prove that
\[\wit_{2}(\phi)=\phi\wedge\psi(\vars_{\s}(\phi))\wedge (u=u),\]
with $u$ a variable of sort $\s_{2}$, is a strong witness for $\Tadds$. Of course, we do not need to prove $\phi$ and $\exists\,\overarrow{x}.\:\wit_{2}(\phi)$ are $\Tadds$-equivalent, for $\overarrow{x}=\vars(\wit_{2}(\phi))\setminus\vars(\phi)$: this was done in the finitely witnessable case. So take a set of variables $V$ of sorts $\s$ and $\s_{2}$, an arrangement $\delta_{V}$ on $V$, and a $\Tadds$-interpretation $\C$ that satisfies $\wit_{2}(\phi)\wedge\delta_{V}$; we must then write 
\[\delta_{V}=\delta_{W}\wedge\delta_{U}\wedge\bigwedge_{x\in W}\bigwedge_{u\in U}\neg(x=u),\]
for $W$ the subset of $V$ of variables of sort $\s$, and $U=V\setminus W$ all variables of sort $\s_{2}$. We then have that $\subs{\C}$ satisfies $\subsf{\phi}{\C}\wedge\psi(\vars_{\s}(\phi))\wedge\delta_{W}$, equal to $\wit(\subsf{\phi}{\C})\wedge\delta_{W}$; so there is a $\T$-interpretation $\A$, with $\s^{\A}=\vars_{\s}(\wit(\subsf{\phi}{\C})\wedge\delta_{W})^{\A}$, that satisfies $\wit(\subsf{\phi}{\C})\wedge\delta_{W}$. Taking then the $\Tadds$-interpretation $\D$ with $\subs{\D}=\A$, $\s_{2}^{\D}=\vars_{\s_{2}}(\wit_{2}(\phi)\wedge \delta_{V})^{\C}$ and $u^{\D}=u^{\C}$ for every variable $u$ of sort $\s_{2}$, we have that: $\D$ satisfies $\phi\wedge\psi(\vars_{\s}(\phi))\wedge\delta_{U}$ since $\subs{\D}$ satisfies $\subsf{\phi}{\C}\wedge\psi(\vars_{\s}(\phi))\wedge\delta_{W}$; $\D$ satisfies $u=u$ and $\delta_{U}$, since $\C$ satisfies $\delta_{U}$; and $\s^{\D}=\vars_{\s}(\wit_{2}(\phi)\wedge \delta_{V})^{\D}$ and $\s_{2}^{\C}=\vars_{\s_{2}}(\wit_{2}(\phi)\wedge \delta_{V})^{\D}$, what finishes proving $\Tadds$ is strongly finitely witnessable.

\item Suppose $\T$ is convex, let $\phi$ be a cube in $\Sigma_{2}$ and assume that 
\[\dash_{\Tadds}\phi\rightarrow\bigvee_{i=1}^{n}x_{i}=y_{i}\vee\bigvee_{j=1}^{m}u_{j}=v_{j},\]
where the $x_{i}$ and $y_{i}$ are of sort $\s$, and the $u_{j}$ and $v_{j}$ are of sort $\s_{2}$; because $\phi$ is a conjunction of literals, and each literal
can only have variables of one sort, we may write $\phi=\phi_{1}\wedge\phi_{2}$, where $\phi_{1}$ has only variables of 
sort $\s$, and $\phi_{2}$ has only variables of sort $\s_{2}$.

It follows that $\dash_{\Tadds}\phi_{1}\rightarrow\bigvee_{i=1}^{n}x_{i}=y_{i}$ or $\dash_{\Tadds}\phi_{2}\rightarrow\bigvee_{j=1}^{m}u_{j}=v_{j}$: indeed, suppose this were not true, and so there exist 
$\Tadds$-interpretations $\C_{1}$ and $\C_{2}$ where, respectively, $\phi_{1}\rightarrow\bigvee_{i=1}^{n}x_{i}=y_{i}$ and 
$\phi_{2}\rightarrow\bigvee_{j=1}^{m}u_{j}=v_{j}$ are not satisfied; hence $\C_{1}$ satisfies $\phi_{1}$ but not 
$\bigvee_{i=1}^{n}x_{i}=y_{i}$, and $\C_{2}$ satisfies $\phi_{2}$, but not $\bigvee_{j=1}^{m}u_{j}=v_{j}$.
So define the $\Tadds$-interpretation $\D$ where: $\s^{\D}=\s^{\C_{1}}$, $\s_{2}^{\D}=\s_{2}^{\C_{2}}$, $x^{\D}=x^{\C_{1}}$ for every variable $x$ of sort $\s$, and
$u^{\D}=u^{\C_{2}}$ for every variable $u$ of sort $\s_{2}$. Because $\D$ agrees with $\C_{1}$ on the variables of sort $\s$, it 
satisfies $\phi$ but not $\bigvee_{i=1}^{n}x_{i}=y_{i}$; and because $\D$ agrees with $\C_{2}$ on the 
variables of sort $\s_{2}$, it satisfies $\phi_{2}$ but not $\bigvee_{j=1}^{m}u_{j}=v_{j}$. Of course, $\D$ then satisfies $\phi$ but 
not $\bigvee_{i=1}^{n}x_{i}=y_{i}\vee\bigvee_{j=1}^{m}u_{j}=v_{j}$, leading to a contradiction.

If we have $\dash_{\Tadds}\phi_{2}\rightarrow\bigvee_{j=1}^{m}u_{j}=v_{j}$, because $\Tadds$ has an axiomatization of formulas with no variables of sort $\s_{2}$, one gets 
$\dash_{\T^{\prime}}\phi_{2}\rightarrow\bigvee_{j=1}^{m}u_{j}=v_{j}$, for $\T^{\prime}$ the theory over the signature with only one sort $\s_{2}$ and
no symbols, axiomatized by the empty set; since $\T^{\prime}$ is convex, according to \Cref{uninterpretedfunctionsis convex}, this means $\dash_{\T^{\prime}}\phi_{2}\rightarrow u_{j}=v_{j}$ for some $1\leq j\leq m$, and 
thus $\dash_{\Tadds}\phi\rightarrow u_{j}=v_{j}$, in whcich case we would be done. 

Suppose then that $\dash_{\Tadds}\phi_{1}\rightarrow\bigvee_{i=1}^{n}x_{i}=y_{i}$. Since 
$\phi_{1}\rightarrow\bigvee_{i=1}^{n}x_{i}=y_{i}$ has no variables of sort $\s_{2}$, we obtain $\dash_{\T}\phi_{1}\rightarrow\bigvee_{i=1}^{n}x_{i}=y_{i}$, and since this theory 
is convex, $\dash_{\T}\phi_{1}\rightarrow x_{i}=y_{i}$, for some $1\leq i \leq n$, and thus $\dash_{\Tadds}\phi\rightarrow x_{i}=y_{i}$.

Reciprocally, assume this time $\Tadds$ is convex, and let $\phi$ be a conjunction of literals such that $\dash_{\T}\phi\rightarrow\bigvee_{i=1}^{n}x_{i}=y_{i}$, where $x_{i}$ and $y_{i}$, for $i\in[1,n]$, are variables of sort $\s$. It follows that $\dash_{\Tadds}\phi\rightarrow\bigvee_{i=1}^{n}x_{i}=y_{i}$, and so $\dash_{\Tadds}\phi\rightarrow x_{i}=y_{i}$ for some $1\leq i\leq n$, meaning $\dash_{\T}\phi\rightarrow x_{i}=y_{i}$.\qedhere

\end{enumerate}

\end{proof}

\section{Proof of \Cref{lem:addfun}}

In the proofs of Lemmas \ref{defining Ts from T} and \ref{lem:addfun}, given the enormous number of involved interpretations in the demonstrations, we agree that $\A$ and $\B$ shall be $\Sigma_{n}$-interpretations (often these will be in addition $\T$-interpretations), while $\C$ and $\D$ will denote $\Sigma^{n}_{s}$-interpretations (that will be, many times, $\Taddf$-interpretations as well).

\begin{lemma}\label{defining Ts from T}
Let $\Sigma_{n}$ be the empty signature with $n$ sorts $S=\{\s_{1}, \ldots, \s_{n}\}$, and $\Sigma^{n}_{s}$ the signature with sorts $S$ and only one function symbol $s$ of arity $\s_{1}\rightarrow\s_{1}$. Given a $\Sigma_{n}$-interpretation $\A$, we define a $\Sigma^{n}_{s}$-interpretation $\plusf{\A}$ by making:
\begin{enumerate}
\item $\s^{\plusf{\A}}=\s^{\A}$ for each $\s\in S$;
\item $s^{\plusf{\A}}(a)=a$ for all $a\in \s_{1}^{\A}$; 
\item and $x^{\plusf{\A}}=x^{\A}$ for every variable $x$.
\end{enumerate}
Reciprocally, given any $\Sigma^{n}_{s}$-interpretation $\C$, we may consider the $\Sigma_{n}$-interpretation $\subf{\C}$ with:
\begin{enumerate}
\item $\s^{\subf{\C}}=\s^{\C}$ for each $\s\in S$;
\item and $x^{\subf{\C}}=x^{\C}$, for every variable $x$.
\end{enumerate}
Finally, given a $\Sigma^{n}_{s}$-formula $\varphi$, we repeatedly replace each occurrence of $s(x)$ in $\varphi$ by $x$ until we obtain a $\Sigma_{n}$-formula
$\subff{\varphi}$. Then, it is true that a $\Sigma^{n}_{s}$-interpretation $\C$ that satisfies $\forall\, x.\:[s(x)=x]$ (where $x$ is of sort $\s_{1}$) also satisfies $\varphi$ iff $\subf{\C}$ satisfies $\subff{\varphi}$; of course, given that for any $\Sigma_{n}$-interpretation $\A$, $\subf{\plusf{\A}}=\A$, $\A$ satisfies a $\Sigma_{n}$-formula $\varphi$ iff $\plusf{\A}$ satisfies $\varphi$ (since $\subff{\varphi}=\varphi$).
\end{lemma}

\begin{proof}
We prove the lemma by induction on the complexity of a formula.
\begin{enumerate}
\item If $\varphi$ is an atomic formula, it is either of the form $s^{i}(x)=s^{j}(y)$, for $x$ and $y$ of sort $\s_{1}$, or of the form $u=v$ for $u$ and $v$ of a sort in $S\setminus\{\s_{1}\}$. In the first case, $\C$ satisfies $\varphi$ iff $x^{\C}=y^{\C}$ since $(s^{\C})^{i}(x^{\C})=x^{\C}$ and $(s^{\C})^{j}(y^{\C})=y^{\C}$; of course, this happens iff $\subf{\A}$ satisfies $x=y$, which is precisely $\subff{\varphi}$. In the second one, it is clear that $\C$ satisfies this $\varphi$ iff $\subf{\A}$ satisfies $\subff{\varphi}=\varphi$.
\item If $\varphi=\neg\psi$ or $\varphi=\psi\vee\xi$, $\C$ satisfies $\varphi$ iff, respectively, it does not satisfy $\psi$, or it satisfies either $\psi$ or $\xi$; by induction hypothesis, this happens iff, respectively, $\subf{\A}$ does not satisfy $\subff{\psi}$, equivalent to it satisfying $\subff{\varphi}$, or $\subf{\C}$ satisfies either $\subff{\psi}$ or $\subff{\xi}$, again equivalent to it satisfying $\subff{\varphi}$. The cases of the connectives $\wedge$ and $\rightarrow$ are entirely analogous.
\item $\C$ satisfies $\varphi=\exists\, x.\:\psi$ iff there exists a $\Sigma^{n}_{s}$-interpretation $\D$, different from $\C$ at most at $x$, that satisfies $\psi$; by induction hypothesis, and from the fact that $\D$ must also satisfy $\forall\, x.\:[s(x)=x]$, we get that $\D$ satisfies $\psi$ iff $\subf{\D}$ satisfies $\subff{\psi}$, equivalent to $\subf{\A}$ satisfying $\subff{\varphi}=\exists\, x.\:\subff{\psi}$. The case of the quantifier $\forall$ is completely analogous.\qedhere
\end{enumerate}
\end{proof}

\addfunlem*

\begin{proof}
We heavily and repeatedly use \Cref{defining Ts from T} throughout this proof. Start by noticing that $\C$ is a $\Taddf$-interpretation iff $\subf{\C}$ is a $\T$-interpretation, given that for any formula $\varphi$ in $\ax(\T)$, $\varphi=\subff{\varphi}$. 

\begin{enumerate}
\item
Suppose $\T$ is stably-infinite: given a quantifier-free $\Sigma^{n}_{s}$-formula $\phi$ and a $\Taddf$-interpretation $\C$ 
that satisfies $\phi$, we have that $\subf{\C}$ is a $\T$-interpretation that satisfies $\subff{\phi}$; since $\subff{\phi}$ is also quantifer-free, there must exist a 
$\T$-interpretation $\A$, with all infinite domains, that satisfies $\subff{\phi}$, and then $\plusf{\A}$ is a $\Taddf$-interpretation, infinite in all of its domains, that satisfies $\phi$. 

Reciprocally, 
suppose $\Taddf$ is stably-infinite, let $\phi$ be a quantifier-free $\Sigma_{n}$-formula and $\A$ a $\T$-interpretation that satisfies $\phi$; since $\phi=\subff{\phi}$, $\plusf{\A}$ satisfies $\phi$, and 
there must exist a $\Taddf$-interpretation $\C$ that satisfies $\phi$ and has all infinite domains. Of course, $\subf{\C}$ then satisfies $\phi$.

\item 
Suppose $\T$ is smooth: given a quantifier-free $\Sigma^{n}_{s}$-formula $\phi$, a $\Taddf$-interpretation $\C$ that satisfies $\phi$, and a function $\kappa$ from $S$ to the class of cardinals such that
$\kappa(\s)\geq|\s^{\C}|$ for each $\s\in S$, we have that $\subf{\C}$ is a $\T$-interpretation that satisfies $\subff{\phi}$. Given that 
$\T$ is supposed to be smooth, $\subff{\phi}$ is quantifier-free and $|\s^{\subf{\C}}|=|\s^{\C}|\leq \kappa(\s)$ for each $\s\in S$, there exists a $\T$-interpretation $\A$ that satisfies $\subff{\phi}$ with $|\s^{\A}|=\kappa(\s)$ (again, for every $\s\in S$); of course, 
$\plusf{\A}$ is then a $\Taddf$-interpretation with $|\s^{\plusf{\A}}|=|\s^{\A}|=\kappa(\s)$, for any $\s\in S$, that satisfies $\phi$. 

Reciprocally, suppose $\Taddf$ is smooth: then, for any quantifier-free $\Sigma_{n}$-formula $\phi$, $\T$-interpretation
$\A$ that satisfies $\phi$, and function $\kappa$ from $S$ to the class of cardinals such that $\kappa(\s)\geq |\s^{\A}|$ for all $\s\in S$, we have that $\plusf{\A}$ satisfies $\phi$, since 
$\phi=\subff{\phi}$; given that $|\s^{\plusf{\A}}|=|\s^{\A}|\leq\kappa(\s)$, one obtains there must exist a $\Taddf$-interpretation $\C$ that satisfies $\phi$ with 
$|\s^{\C}|=\kappa(\s)$ for every $\s\in S$. And then, $\subf{\C}$ is a $\T$-interpretation that satisfies $\phi$ with $|\s^{\subf{\C}}|=\kappa(\s), \forall \s\in S$.

\item
Suppose now $\T$ is finitely witnessable, with witness $\wit$: we shall prove that $\Taddf$ is also finitely witnessable, with witness 
\[\wit_{s}(\phi)=\phi\wedge\wit(\subff{\phi}).\]
To start with, given a quantifier-free $\Sigma^{n}_{s}$-formula $\phi$, $\subff{\phi}$ is also quantifier-free, and thus $\subff{\phi}$ and $\exists\,\overarrow{x}.\:\wit(\subff{\phi})$ are $\T$-equivalent, for 
$\overarrow{x}=\vars(\wit(\subff{\phi}))\setminus\vars(\subff{\phi})=\vars(\wit(\phi))\setminus\vars(\phi)$.
If the $\Taddf$-interpretation $\C$ satisfies $\phi$, $\subf{\C}$ satisfies $\subff{\phi}$ and thus $\exists\,\overarrow{x}.\:\wit(\subff{\phi})$; since $\subff{(\exists\,\overarrow{x}.\:\wit(\subff{\phi}))}=\exists\,\overarrow{x}.\:\wit(\subff{\phi})$, we have that $\C$ satisfies $\exists\,\overarrow{x}.\:\wit(\subff{\phi})$ and thus
$\exists\,\overarrow{x}.\:\wit_{s}(\phi)=\exists\,\overarrow{x}.\:(\phi\wedge\wit(\subff{\phi}))$,
given that $\phi$ has none of the variables in $\overarrow{x}$. Of course, if the $\Taddf$-interpretation $\C$ satisfies $\exists\,\overarrow{x}.\:\wit_{s}(\phi)$, it must satisfy $\phi$, and so the two formulas are $\Taddf$-equivalent.

Suppose now that a $\Taddf$-interpretation $\C$ satisfies $\wit_{s}(\phi)$, and we have that $\subf{\C}$ is a $\T$-interpretation that satisfies $\subff{(\wit_{s}(\phi))}=\subff{\phi}\wedge\wit(\subff{\phi})$; from the facts that $\T$ is finitely witnessable, with witness $\wit$, and $\subf{\C}$ 
satisfies $\wit(\subff{\phi})$, it follows that there exists a $\T$-interpretation $\A$ that satisfies $\wit(\subff{\phi})$, and thus $\subff{\phi}$, with $\s^{\A}=\vars_{\s}(\wit(\subff{\phi}))^{\A}$ for each $\s\in S$. Then, since 
$\subff{(\wit(\subff{\phi}))}=\wit(\subff{\phi})$, $\plusf{\A}$ is a $\Taddf$-interpretation that satisfies $\wit(\subff{\phi})$; since $\A$ also satisfies $\subff{\phi}$, $\plusf{\A}$ satisfies $\phi$ and thus $\wit_{s}(\phi)=\phi\wedge\wit(\subff{\phi})$ as well; and, 
given that $\s^{\plusf{\A}}=\s^{\A}$ and $\vars_{\s}(\wit(\subff{\phi}))=\vars_{\s}(\phi\wedge\wit(\subff{\phi}))$ (both for any $\s\in S$), we get $\s^{\plusf{\A}}=\vars(\wit_{s}(\phi))^{\plusf{\A}}$ for all $\s\in S$, proving $\Taddf$ is indeed finitely witnessable. 

Reciprocally, suppose $\Taddf$ is finitely witnessable with witness $\wit$, and we want to prove that $\T$ is finitely witnessable with witness $\wit_{0}(\phi)=\subff{(\wit(\phi))}$. We start by taking a quantifier-free $\Sigma_{n}$-formula $\phi$, and since $\phi$ is a quantifier-free $\Sigma^{n}_{s}$-formula as 
well, $\phi$ and $\exists\,\overarrow{x}.\:\wit(\phi)$ are $\Taddf$-equivalent, where $\overarrow{x}=\vars(\wit(\phi))\setminus\vars(\phi)$. So, suppose the $\T$-interpretation $\A$ satisfies $\phi=\subff{\phi}$: then $\plusf{\A}$ satisfies $\phi$ and thus $\exists\,\overarrow{x}.\:\wit(\phi)$, meaning $\A$ satisfies 
\[\subff{(\exists\,\overarrow{x}.\:\wit(\phi))}=\exists\,\overarrow{x}.\:\subff{\wit(\phi)}=\exists\,\overarrow{x}.\:\wit_{0}(\phi).\]
If the $\T$-interpretation $\A$ satisfies $\exists\,\overarrow{x}.\:\wit_{0}(\phi)$, $\plusf{\A}$ satisfies $\exists\,\overarrow{x}.\:\wit(\phi)$ and hence $\phi$, meaning $\A$ satisfies $\phi$ and proving that $\phi$ and $\exists\,\overarrow{x}.\:\wit_{0}(\phi)$ are $\T$-equivalent.

Now suppose that $\A$ is a $\T$-interpretation that satisfies $\wit_{0}(\phi)$, and thus $\plusf{\A}$ is a $\Taddf$-interpretation that satisfies $\wit(\phi)$, since $\wit_{0}(\phi)=\subff{(\wit(\phi))}$. Given that $\Taddf$ is finitely witnessable with witness $\wit$, there exists a $\Taddf$-interpretation $\C$ that satisfies $\wit(\phi)$ with 
$\s^{\C}=\vars_{\s}(\wit(\phi))^{\C}$ for any $\s\in S$. It follows that $\subf{\C}$ is a $\T$-interpretation that satisfies $\subff{(\wit(\phi))}=\wit_{0}(\phi)$ with $\s^{\subf{\C}}=\vars(\wit_{0}(\phi))^{\subf{\C}}$, $\forall \s\in S$ (since $\s^{\subf{\C}}=\s^{\C}$ and $\vars_{\s}(\wit(\phi))=\vars_{\s}(\wit_{0}(\phi))$, both for all $\s\in S$).

\item
Now, let us look at strong finite witnessability. If $\T$ is strongly finitely witnessable with witness $\wit$, we state $\Taddf$ is also strongly finitely witnessable with witness $\wit_{s}(\phi)=\phi\wedge\wit(\subff{\phi})$, and we already know that $\phi$ and $\exists\,\overarrow{x}.\:\wit_{s}(\phi)$ are $\Taddf$-equivalent from our discussion above about finite witnessability. So let $V$ be a set of variables, $\delta_{V}$ an arrangement on $V$, and $\C$ a $\Taddf$-interpretation that satisfies $\wit_{s}(\phi)\wedge\delta_{V}$. Since 
$\subff{(\wit_{s}(\phi)\wedge\delta_{V})}=\subff{\phi}\wedge\wit(\subff{\phi})\wedge\delta_{V}$,
$\subf{\C}$ satisfies $\wit(\subff{\phi})\wedge\delta_{V}$; it follows, from the fact that $\T$ has as strong witness $\wit$, that there is a $\T$-interpretation $\A$ that satisfies $\wit(\subff{\phi})\wedge\delta_{V}$ (and hence $\subff{\phi}$ as well) with $\s^{\A}=\vars_{\s}(\wit(\subff{\phi})\wedge\delta_{V})^{\A}$, for each and any $\s\in S$. $\plusf{\A}$ is, therefore, a $\Taddf$-interpretation that satisfies $\phi\wedge\wit(\subff{\phi})\wedge\delta_{V}$, what amounts to $\wit_{s}(\phi)\wedge\delta_{V}$, with $\s^{\plusf{\A}}=\vars_{\s}(\wit_{s}(\phi)\wedge\delta_{V})^{\plusf{\A}}$ once one notices that $\s^{\plusf{\A}}=\s^{\A}$ and $\vars_{\s}(\wit_{s}(\phi)\wedge\delta_{V})=\vars_{\s}(\wit(\subff{\phi})\wedge\delta_{V})$, both of these for all sorts $\s$ in $S$. So $\Taddf$ is indeed strongly finitely witnessable.

Reciprocally, assume $\Taddf$ is strongly finitely witnessable with strong witness $\wit$, and we will prove that $\T$ is also strongly finitely witnessable with strong witness 
\[\wit_{0}(\phi)=\subff{(\wit(\phi))}.\]
Of course, from our discussion about finite witnessability we already know that $\phi$ and $\exists\,\overarrow{x}.\:\wit_{0}(\phi)$ are $\T$-equivalent, where $\overarrow{x}=\vars(\wit_{0}(\phi))\setminus\vars(\phi)$. So take a set of variables $V$, an arrangement $\delta_{V}$ on $V$, and a $\T$-interpretation $\A$ that satisfies $\wit_{0}(\phi)\wedge\delta_{V}$. Since $\subff{(\wit(\phi)\wedge\delta_{V})}$ equals $\wit_{0}(\phi)\wedge\delta_{V}$, we have that $\plusf{\A}$ satisfies $\wit(\phi)\wedge\delta_{V}$, and so there exists a $\Taddf$-interpretation $\C$ that satisfies $\wit(\phi)\wedge\delta_{V}$ with $\s^{\C}=\vars_{\s}(\wit(\phi)\wedge\delta_{V})^{\C}$ for every $\s$ in $S$. Then $\subf{\C}$ satisfies $\wit_{0}(\phi)\wedge\delta_{V}$ and, since $\s^{\subf{\C}}=\s^{\C}$ and $\vars_{\s}(\wit(\phi)\wedge\delta_{V})=\vars_{\s}(\wit_{0}(\phi)\wedge\delta_{V})$ for any $\s$ in $S$, $\s^{\subf{\C}}=\vars_{\s}(\wit_{0}(\phi)\wedge\delta_{V})^{\C}$ again for any $\s$ in $S$, what finishes the proof.

\item Finally, suppose $\T$ is convex. Let $\phi$ be a conjunction of literals in $\Sigma^{n}_{s}$ (notice $\subff{\phi}$ is a conjunction of literals in $\Sigma_{n}$), and assume that $\dash_{\Taddf}\phi\rightarrow\bigvee_{i=1}^{n}x_{i}=y_{i}$ (implying $\dash_{\T}\subff{\phi}\rightarrow\bigvee_{i=1}^{n}x_{i}=y_{i}$): it follows that $\dash_{\T}\subff{\phi}\rightarrow x_{i}=y_{i}$ for some $1\leq i\leq n$, and since $\subff{(x_{i}=y_{i})}=(x_{i}=y_{i})$, that $\dash_{\Taddf}\phi\rightarrow x_{i}=y_{i}$, proving $\Taddf$ is convex. 

Reciprocally, assume $\Taddf$ is convex, and let $\phi$ be a cube in $\Sigma_{n}$ (and so $\subff{\phi}=\phi$) such that $\dash_{\T}\phi\rightarrow\bigvee_{i=1}^{n}x_{i}=y_{i}$ (and thus $\dash_{\Taddf}\phi\rightarrow\bigvee_{i=1}^{n}x_{i}=y_{i}$); we must then have $\dash_{\Taddf}\phi\rightarrow x_{i}=y_{i}$ for some $1\leq i\leq n$, and again by using that $\subff{(x_{i}=y_{i})}=(x_{i}=y_{i})$ we get $\dash_{\T}\phi\rightarrow x_{i}=y_{i}$, proving $\T$ is convex.\qedhere
\end{enumerate}
\end{proof}

\section{Proof of \Cref{Tvee is ...}}

In the proof of Lemmas \ref{Defining Tvee from T} and \ref{Tvee is ...}, we will use the same convention found in the proofs of Lemmas \ref{defining Ts from T} and \ref{lem:addfun}: $\A$ and $\B$ will be $\Sigma_{n}$-interpretations (often $\T$-interpretations in addition), while $\C$ and $\D$ will be $\Sigma^{n}_{s}$-interpretations (often $\Taddnc$-interpretations as well). Also important, on what is to follow, given a $\Sigma_{s}^{n}$ formula $\phi$ and a $\Sigma_{s}^{n}$-interpretation $\C$ that satisfies $\phi$, is to consider the formula $\subncf{\phi}{\C}$. To define this formula, let $\vars_{\s_{1}}(\phi)=\{z_{1}, \ldots, z_{n}\}$, $M_{i}$ be the maximum of $j$ such that $s^{j}(z_{i})$ appears in $\phi$, and $\{y_{i,j} : 1\leq i\leq n, 0\leq j\leq M_{i}+2\}$ be fresh variables of sort $\s_{1}$.
\begin{enumerate}
\item We define $\dagg{\phi}$ by replacing each atomic subformula $s^{j}(z_{i})=s^{q}(z_{p})$ of $\phi$ by $y_{i,j}=y_{p,q}$;
\item $\delta_{V}$ is the arrangement induced on the set of variables $V=\{y_{i,j} : 1\leq i\leq n, 0\leq j\leq M_{i}+2\}$ by making $x$ related to $y$ iff $x^{\C}=y^{\C}$.
\end{enumerate}
We then make $\subncf{\phi}{\C}=\dagg{\phi}\wedge\delta_{V}$.

\begin{lemma}\label{Defining Tvee from T}
Let $\Sigma_{n}$ be the empty signature with $n$ sorts $S=\{\s_{1}, \ldots, \s_{n}\}$, and $\Sigma_{s}^{n}$ be the signature with sorts $S$ and only one function symbol $s$ of arity $\s_{1}\rightarrow\s_{1}$. Given a $\Sigma_{n}$-interpretation $\A$, we define a $\Sigma_{s}^{n}$-interpretation $\plusnc{A}$ by making:
\begin{enumerate}
\item $\s^{\plusnc{A}}=\s^{\A}$ for each $\s\in S$;
\item $s^{\plusnc{A}}(a)=a$ for all $a\in \s_{1}^{\A}$;
\item and $x^{\plusnc{A}}=x^{\A}$ for all variables $x$.
\end{enumerate}
Reciprocally, given any $\Sigma_{s}^{n}$-interpretation $\C$, we may consider the $\Sigma_{n}$-interpretation $\subnc{\C}$ with:
\begin{enumerate}
\item $\s^{\subnc{\C}}=\s^{\C}$ for each $\s\in S$;
\item and $x^{\subnc{\C}}=x^{\C}$ for all variables $x$.
\end{enumerate}
Then, it is true that a $\Sigma_{n}$-interpretation $\A$ satisfies a $\sigma_{n}$-formula $\varphi$ iff $\plusnc{\A}$ satisfies $\varphi$.
\end{lemma}

\begin{proof}
See \Cref{defining Ts from T}.
\end{proof}

\lemaddfunconv*

\begin{proof} 
Let $\A$ be a $\Sigma_{n}$-structure, and $\C$ a $\Sigma_{s}^{n}$-structure: of course, from \Cref{Defining Tvee from T}, $\A$ is a model of $\T$ iff $\plusnc{\A}$ is a model of $\Taddnc$; and $\C$ is a model of $\Taddnc$ iff it satisfies $\psiv$ and $\subnc{\C}$ is a model of $\T$. This last observation is true since, 
although $\plusnc{\subnc{\C}}$ 
may differ from $\C$ in the function assigned to $s$, any formula on the axiomatization of $\T$ is free of symbols, thus receiving the same value in both $\C$ and $\plusnc{\subnc{\C}}$ (and therefore $\subnc{\C}$).
\begin{enumerate}
\item Suppose that $\T$ is stably-infinite, let $\phi$ be a quantifier-free $\Sigma_{s}^{n}$-formula and $\C$ a $\Taddnc$-interpretation that satisfies 
$\phi$. Then the $\Taddnc$-interpretation $\C^{\prime}$, obtained from $\C$ by making $y_{i,j}^{\C^{\prime}}=(s^{\C})^{j}(z_{i}^{\C})$, satisfies the 
quantifier and symbol-free formula $\subncf{\phi}{\C}$, meaning $\subnc{\C^{\prime}}$ satisfies $\subncf{\phi}{\C}$. Since $\T$ is assumed to be stably-infinite, there is an infinite (on all domains) $\T$-interpretation $\A$ that satisfies 
$\subncf{\phi}{\C}$. By defining the $\Taddnc$-interpretation $\D$ such that: $\s^{\D}=\s^{\A}$ for all $\s\in S$; for all $a=y_{i,j}^{\A}\in V^{\A}$, $s^{\D}(a)=y_{i, j+1}^{\A}$, and for all $a\in \s_{1}^{\A}\setminus V^{\A}$, 
$s^{\D}(a)=a$; and, for all variables $x$, $x^{\D}=x^{\A}$, we see that $\D$ is a $\Taddnc$-interpretation, infinite on all domains, that satisfies $\phi$.

Reciprocally, assume $\Taddnc$ is stably-infinite, take a quantifier-free  $\Sigma_{n}$-formula $\phi$ and a $\T$-interpretation $\A$ that satisfies $\phi$, meaning $\plusnc{\A}$ satisfies $\phi$. There must then exist a $\Taddnc$-interpretation $\C$, infinite in all domains, that satisfies $\phi$, and 
since $\phi$ is free of symbols, $\subnc{\C}$ satisfies $\phi$ and has $\s^{\subnc{\C}}$ infinite for all $\s\in S$.

\item Suppose now that $\T$ is smooth, let $\phi$ be a quantifier-free $\Sigma_{s}^{n}$-formula, $\C$ a $\Taddnc$-interpretation that satisfies 
$\phi$, and $\kappa$ a function from $S$ to the class of cardinals such that $\kappa(\s)\geq |\s^{\C}|$ for each $\s\in S$. We take the $\Taddnc$-interpretation $\C^{\prime}$ (obtained from $\A$ by changing 
$y_{i,j}^{\C^{\prime}}$ so that it equals $(s^{\C})^{j}(z_{i}^{\C})$), which satisfies the quantifier and symbol-free formula $\subncf{\phi}{\C}$, meaning 
$\subnc{\C^{\prime}}$ satisfies $\subncf{\phi}{\C}$. Using $\T$ is smooth, there must exist a $\T$-interpretation $\A$ that satisfies $\subncf{\phi}{\C}$ and 
$|\s^{\A}|=\kappa(\s)$ for each $\s\in S$. We finally define a $\Taddnc$-interpretation $\D$ such that: $\s^{\D}=\s^{\A}$ for all $\s\in S$; for all 
$a=y_{i,j}^{\A}\in V^{\A}$, $s^{\D}(a)=y_{i, j+1}^{\A}$, and for all $a\in \s_{1}^{\A}\setminus V^{\A}$, $s^{\D}(a)=a$; and, for all variables $x$, 
$x^{\D}=x^{\A}$. Given that $\D$ satisfies an atomic subformula of $\phi$ iff $\C$ does so, it satisfies $\phi$; it is also a $\Taddnc$-interpretation such that $|\s^{\D}|=\kappa(\s)$ for each $\s\in S$.

Reciprocally, suppose $\Taddnc$ is smooth, and then take a quantifier-free $\Sigma_{n}$-formula $\phi$, a $\T$-interpretation $\A$ that satisfies $\phi$ and a function $\kappa$ from $S$ to the class of cardinals such that $\kappa(\s)\geq |\s^{\A}|$ for all $\s\in S$; it follows that $\plusnc{\A}$ satisfies $\phi$, and since $\Taddnc$ is smooth there must exist a $\Taddnc$-interpretation $\C$ that satisfies $\phi$ with $|\s^{\C}|=\kappa(\s)$ for each $\s\in S$. Since $\phi$ has no function symbols, $\subnc{\C}$ is then a $\T$-interpretation that satisfies $\phi$ with $|\s^{\subnc{\C}}|=\kappa(\s)$ for any $\s\in S$, finishing the proof that $\T$ is also smooth.

\item We start by assuming $\Taddnc$ is finitely witnessable, with witness $\wit$, and then state that 
\[\wit_{0}(\phi)=\phi\wedge\dagg{\wit(\phi)}\wedge\bigwedge_{i=1}^{n}(y_{i,0}=z_{i})\wedge\Psiv(\overarrow{y})\wedge\Fun(\overarrow{y})\]
is a witness for $\T$, where: $\vars_{\s_{1}}(\wit(\phi))=\{z_{1}, \ldots, z_{n}\}$; $M_{i}$ is the maximum of $j$ such that $s^{j}(z_{i})$ appears in $\wit(\phi)$; $\overarrow{y}=\{y_{i,j} : 1\leq i\leq n, 1\leq j\leq M_{i}+2\}$ are fresh variables of sort $\s_{1}$; $\dagg{\wit(\phi)}$ is obtained from $\wit(\phi)$, in this case, by replacing $s^{j}(z_{i})=s^{q}(z_{p})$ for $y_{i,j}=y_{p,q}$, where $1\leq i\leq n$ and $0\leq j\leq M_{i}$;
\[\Psiv(\overarrow{y})=\bigwedge_{i=1}^{n}\big[(y_{i,2}=y_{i,1})\vee(y_{i,2}=y_{i,0})\big]\]
and
\[\Fun(\overarrow{y})=\bigwedge_{i=1}^{n}\bigwedge_{p=1}^{n}\bigwedge_{j=0}^{M_{i}}\bigwedge_{q=0}^{M_{p}}\big[(y_{i,j}=y_{p,q})\rightarrow(y_{i,j+1}=y_{p,q+1})\big].\]
We start by taking a quantifier-free $\Sigma_{n}$ formula $\phi$ and a $\T$-interpretation $\A$ that satisfies $\phi$, and then $\plusnc{\A}$ satisfies $\phi$ and thus $\exists\,\overarrow{x}.\:\wit(\phi)$, where
\[\overarrow{x}=\vars(\wit(\phi))\setminus\vars(\phi)=\vars(\wit_{0}(\phi))\setminus[\vars(\phi)\cup\overarrow{y}].\]
There must exist a $\Taddnc$-interpretation $\C$, differing from $\plusnc{\A}$ at most on $\overarrow{x}$, that satisfies $\wit(\phi)$ and in addition $\phi$, since $\plusnc{\A}$ does satisfy $\phi$ and $\C$ does not 
change the value of the variables in $\phi$. Taking then the $\T$-interpretation $\B$ that differs from $\subnc{\C}$ at most on $\overarrow{x}\cup\overarrow{y}$ such that, for $x$ in $\overarrow{x}$ one has 
$x^{\B}=x^{\C}$, and for $y_{i,j}$ in $\overarrow{y}$ one has $y_{i,j}^{\B}=(s^{\C})^{j}(z_{i}^{\C})$, it is easy to see that: $\B$ satisfies $\phi\wedge\dagg{\wit(\phi)}$, 
since $\C$ satisfies $\phi\wedge\wit(\phi)$; and $\B$ also satisfies $\Psiv(\overarrow{y})$ and $\Fun(\overarrow{y})$ given that 
$z_{i}^{\B}=y_{i,0}^{\B}=\cdots=y_{i,M_{i}+2}^{\B}$ for each $1\leq i\leq n$. Hence, $\A$ satisfies $\exists\,\overarrow{x}.\exists\,\overarrow{y}.\:\wit_{0}(\phi)$. Now, if the $\T$-interpretation $\A$ satisfies instead $\exists\,\overarrow{x}.\exists\, \overarrow{y}.\:\wit_{0}(\phi)$, given that the 
variables in $\overarrow{x}$ and $\overarrow{y}$ do not occur in $\phi$, it follows that $\A$ satisfies $\phi$, and so $\phi$ and $\exists\,\overarrow{x}.\exists\, \overarrow{y}.\:\wit_{0}(\phi)$ are $\T$-equivalent.

Now, suppose $\A$ is a $\T$-interpretation that satisfies $\wit_{0}(\phi)$; $\plusnc{\A}$ must then satisfy $\wit(\phi)$, and so there exists a $\Taddnc$-interpretation $\C$ that satisfies $\wit(\phi)$ with $\s^{\C}=\vars_{\s}(\wit(\phi))^{\C}$ for each $\s\in S$. We then take the $\T$-interpretation $\B$, differing from $\subnc{\C}$ at most on $\overarrow{y}$, such that $y_{i,j}^{\B}=(s^{\C})^{j}(z_{i}^{\C})$. This way, $\B$ satisfies $\phi$ because $\C$ does so; it satisfies $\dagg{\wit(\phi)}$ and $\bigwedge_{i=1}^{n}y_{i,0}=z_{i}$ by its very definition; and it satisfies $\Psiv(\overarrow{y})$ and $\Fun(\overarrow{y})$ because, respectively, $\C$ satisfies $\psiv$ and $s^{\C}$ is a function. Furthermore, for any $\s\in S$, $\vars_{\s}(\wit(\phi))\subset\vars_{\s}(\wit_{0}(\phi))$, and so $\s^{\B}=\vars_{\s}(\wit_{0}(\phi))^{\B}$.


Reciprocally, assume $\T$ is finitely witnessable, with a witness $\wit(\phi)=\phi\wedge\psi(\vars(\phi))$, where $\psi(\vars(\phi))$ is a formula that depends only on the variables of $\phi$; there is always such a witness as proved in \Cref{lem:witarevariabledependent}. We state that $\Taddnc$ is also finitely witnessable, with witness 
\[\wit_{s}(\phi)=\phi\wedge\psi(\vars(\phi)\cup \overarrow{y})\wedge\bigwedge_{i=1}^{n}\bigwedge_{j=0}^{M_{i}+2}[y_{i,j}=s^{j}(z_{i})]\wedge\Psiv(\overarrow{y})\wedge\Fun(\overarrow{y}),\] 
where now $\phi$ may be a $\Sigma_{s}^{n}$-formula, $\vars_{\s_{1}}(\phi)=\{z_{1},\ldots, z_{n}\}$, $M_{i}$ is the maximum of $j$ such that $s^{j}(z_{i})$ occurs in $\phi$, and $\overarrow{y}=\{y_{i,j} : 1\leq i\leq n, 0\leq j\leq M_{i}+2\}$ are fresh variables of sort $\s_{1}$. Take a 
$\Sigma_{s}^{n}$-formula $\phi$ and a $\Taddnc$-interpretation $\C$ that satisfies $\phi$. If we take the interpretation $\A$ that 
differs from $\subnc{\C}$ on $\overarrow{y}$, where $y_{i,j}^{\A}=(s^{\C})^{j}(z_{i}^{\C})$, we have that it satisfies $\Psiv(\overarrow{y})$ (since $\C$ satisfies $\psiv$), $\Fun(\overarrow{y})$ (since $s^{\C}$ is a function in $\C$), and $\subncf{\phi}{\C}$, and thus $\exists\,\overarrow{x}.\:\wit(\subncf{\phi}{\C})=\subncf{\phi}{\C}\wedge\exists\,\overarrow{x}.\:\psi(\vars(\subncf{\phi}{\C}))$ for
\[\overarrow{x}=\vars(\wit(\subncf{\phi}{\C}))\setminus\vars(\subncf{\phi}{\C})=\vars(\wit_{s}(\phi))\setminus\big[\vars(\phi)\cup \overarrow{y}\big];\] 
there is then an interpretation $\A^{\prime}$, differing from $\A$ at most on $\overarrow{x}$,
that satisfies $\subncf{\phi}{\C}\wedge\psi(\vars(\subncf{\phi}{\C}))$. We now take the $\Taddnc$-interpretation $\D$ that differs from $\C$ at most on $\overarrow{x}\cup\overarrow{y}$, with $\subnc{\D}=\A^{\prime}$, and we have that $\D$ satisfies: $\phi$, since $\C$ satisfies $\phi$, and 
$\D$ only differs from $\C$ at most on 
$\overarrow{x}\cup\overarrow{y}$, none of these variables present in $\phi$; 
$\psi(\vars(\phi)\cup\overarrow{y})$, since this formula is satisfied by $\A^{\prime}$ and has no function symbols; $\Psiv(\overarrow{y})$ and $\Fun(\overarrow{y})$, since both are satisfied by $\A$ and 
$\D$ only differs on the values given for variables from $\A$ on $\overarrow{x}$; and $\bigwedge_{j=0}^{M_{i}+2}[y_{i,j}=s^{j}(z_{i})]$, since
\[y_{i,j}^{\D}=y_{i,j}^{\A^{\prime}}=y_{i,j}^{\A}=(s^{\C})^{j}(z_{i}^{\C})=(s^{\D})^{j}(z_{i}^{\D}).\]
This means, of course, that $\C$ satisfies $\exists\,\overarrow{x}.\exists\,\overarrow{y}.\:\wit_{s}(\phi)$, as we needed to show. Reciprocally, if the $\Taddnc$-interpretation $\C$ satisfies $\exists\,\overarrow{x}.\exists\,\overarrow{y}.\:\wit_{s}(\phi)$, it is obvious that it satisfies $\phi$ (since the variables in $\overarrow{x}\cup\overarrow{y}$ do not occur in $\phi$), meaning $\phi$ and $\exists\,\overarrow{x}.\exists\,\overarrow{y}.\:\wit_{s}(\phi)$ are $\Taddnc$-equivalent.

So, suppose $\C$ is a $\Taddnc$-interpretation that satisfies $\wit_{s}(\phi)$, meaning $\A$, differing from $\subnc{\C}$ at most on $y_{i,j}\in \overarrow{y}$ ,where $y_{i,j}^{\A}=(s^{\C})^{j}(z_{i}^{\C})$, satisfies 
\[\subncf{\phi}{\C}\wedge\psi(\vars(\subncf{\phi}{\C})\cup\overarrow{y})\wedge\Psiv(\overarrow{y})\wedge\Fun(\overarrow{y})=\wit(\subncf{\phi}{\C}\wedge\Psiv(\overarrow{y})\wedge\Fun(\overarrow{y})).\]
There is, therefore, a $\T$-interpretation $\B$ that satisfies $\wit(\subncf{\phi}{\C}\wedge\Psiv(\overarrow{y})\wedge\Fun(\overarrow{y}))$ with $\s^{\B}=\vars_{\s}(\wit_{s}(\phi))^{\B}$, for each $\s\in S$, the set 
$\vars_{\s}(\wit_{s}(\phi))$ happening to be the same as $\vars_{\s}(\wit(\subncf{\phi}{\C}\wedge\Psiv(\overarrow{y})\wedge\Fun(\overarrow{y})))$. We then build a $\Taddnc$-interpretation $\D$ with: $\s^{\D}=\s^{\B}$ for every 
$\s\in S$; $x^{\D}=x^{\B}$ for every variable $x$; and $s^{\D}(y_{i,j}^{\D})=y_{i, j+1}^{\D}$ for every $1\leq i\leq n$ and $0\leq j\leq M_{i}$, and $s^{\D}(a)=a$ for all other elements of $\s_{1}^{\D}$. We indeed 
have that $s^{\D}$ is a function since $\B$ satisfies $\Fun(\overarrow{y})$, and $\D$ satisfies $\psiv$ because $\B$ satisfies $\Psiv(\overarrow{y})$. It follows that, not only $\D$ satisfies $\wit_{s}(\phi)$, but also has the property that $\s^{\D}=\vars_{\s}(\wit_{s}(\phi))^{\D}$, proving $\wit_{s}$ is indeed a witness for $\Taddnc$.

\item Assume now that $\Taddnc$ is strongly finitely witnessable, with a strong witness $\wit$, and we will prove 
\[\wit_{0}(\phi)=\phi\wedge\dagg{\wit(\phi)}\wedge\bigwedge_{i=1}^{n}(y_{i,0}=z_{i})\wedge\Psiv(\overarrow{y})\wedge\Fun(\overarrow{y})\]
is a strong witness for $\T$, where we again have: $\vars_{\s_{1}}(\wit(\phi))=\{z_{1}, \ldots, z_{n}\}$; $M_{i}$ is the maximum of $j$ such that $s^{j}(z_{i})$ appears in $\wit(\phi)$;  $\overarrow{y}=\{y_{i,j} : 1\leq i\leq n, 1\leq j\leq M_{i}+2\}$ are fresh variables of sort $\s_{1}$; and $\dagg{\wit(\phi)}$ is obtained from $\wit(\phi)$, in this case, by replacing $s^{j}(z_{i})=s^{q}(z_{p})$ for $y_{i,j}=y_{p,q}$, where $1\leq i\leq n$ and $0\leq j\leq M_{i}$. As we already 
know, $\phi$ and $\exists\,\overarrow{x}.\exists\,\overarrow{y}.\:\wit_{0}(\phi)$ are $\T$-equivalent, for $\overarrow{x}=\vars(\wit(\phi))\setminus\vars(\phi)$; so take a set of variables $V$, an arrangement $\delta_{V}$ on $V$, and a $\T$-interpretation $\A$ that satisfies $\wit_{0}(\phi)\wedge\delta_{V}$. It follows that $\plusnc{\A}$ satisfies $\wit(\phi)\wedge\delta_{W}$, where 
$W=V\cup\vars(\wit_{0}(\phi))$ and $\delta_{W}$ is the arrangement induced on $W$ by making $x$ related to $y$ iff $x^{\plusnc{\A}}=y^{\plusnc{\A}}$; since $\wit$ is a strong witness for $\Taddnc$, there exists a $\Taddnc$-interpretation $\C$ that satisfies $\wit(\phi)\wedge\delta_{W}$ and $\s^{\C}=\vars_{\s}(\wit(\phi)\wedge\delta_{W})^{\C}$ for each $\s\in S$. Now, consider the 
$\Taddnc$-interpretation $\D$ with $\s^{\D}=\s^{\C}$ for each $\s\in S$, $x^{\D}=x^{\C}$ for each variable $x$, and $s^{\D}$ equal to $s^{\C}$ except on the elements 
\[\{y_{i,j}^{\B} : 1\leq i\leq n, 1\leq j\leq M_{i}+2\}\cup\{(s^{\C})^{j}(z_{i}^{\C}) : 1\leq i\leq m, 1\leq j\leq M_{i}+2\},\]
where we define instead: $s^{\D}(y_{i,j}^{\C})=y_{i,j+1}^{\C}$, for $1\leq i\leq n$ and $j\leq M_{i}+1$; and, if $s^{\D}((s^{\C})^{j}(z_{i}^{\C}))$ has not yet been defined, we simply make it equal to $(s^{\C})^{j}(z_{i}^{\C})$ (notice $s^{\D}(y_{i, M_{i}+2}^{\C})$ must have, under these conditions, already been defined, since either $y_{i, M_{i}+2}^{\C}=y_{i, M_{i}}^{\C}$ or $y_{i, M_{i}+2}^{\C}=y_{i, M_{i}+1}^{\C}$). Since $\plusnc{\A}$ satisfies $\Fun(\overarrow{y})$,
if $\C$ satisfies $y_{i,j}=y_{p,q}$ it also satisfies $y_{i,j+1}=y_{p, q+1}$, meaning $s^{\D}$ is indeed a well-defined function; furthermore, given that $\plusnc{\A}$ also satisfies $\Psiv(\overarrow{y})$, we have that $\C$ either satisfies $y_{i,2}=y_{i,1}$ or $y_{i,2}=y_{i,1}$, and so $\D$ satisfies 
$\psiv$, being therefore a $\Taddnc$-interpretation. It is somewhat clear that $\D$ satisfies $\phi$, $\bigwedge_{i=1}^{n}y_{i,0}=z_{i}$, $\Fun(\overarrow{y})$ and $\Psiv(\overarrow{y})$, and we have to prove that 
is satisfies $\dagg{\wit(\phi)}$, so let $s^{j}(z_{i})=s^{q}(z_{p})$ be one of its atomic formulas; we have $(s^{\D})^{j}(z_{i}^{\D})=y_{i,j}^{\D}$ and $(s^{\D})^{q}(z_{i}^{\D})=y_{p,q}^{\D}$ from the definition of $s^{\D}$. Furthermore, 
$\plusnc{\A}$ satisfies $y_{i,j}=y_{p,q}$ iff $\C$ and thus $\D$ satisfy the same formula; this means $\plusnc{\A}$ satisfies $s^{j}(z_{i})=s^{q}(z_{p})$ iff $\D$ does so, and therefore $\D$ satisfies $\dagg{\wit(\phi)}$, since $\plusnc{\A}$ 
does so and this formula is quantifier-free. Given that $\s^{\C}=\vars_{\s}(\wit(\phi)\wedge\delta_{W})^{\C}$, and $\s^{\D}=\vars_{\s}(\wit(\phi)\wedge\delta_{W})^{\D}=\vars_{\s}(\wit_{0}(\phi))^{\D}$, and thus $\wit_{0}$ is indeed a strong witness for $\T$.

Reciprocally, suppose now $\T$ is strongly finitely witnessable, with a witness $\wit(\phi)=\phi\wedge\psi(\vars(\phi))$, where $\psi(\vars(\phi))$ is a formula that depends only on the variables of $\phi$. We state that then $\Taddnc$ is also finitely witnessable, with witness 
\[\wit_{s}(\phi)=\phi\wedge\psi(\vars(\phi)\cup \overarrow{y})\wedge\bigwedge_{i=1}^{n}\bigwedge_{j=0}^{M_{i}+2}[y_{i,j}=s^{j}(z_{i})]\wedge\Psiv(\overarrow{y})\wedge\Fun(\overarrow{y}),\] 
where now $\phi$ may be a $\Sigma_{s}^{n}$-formula, $\vars_{\s_{1}}(\phi)=\{z_{1},\ldots, z_{n}\}$, $M_{i}$ is the maximum of $j$ such that $s^{j}(z_{i})$ occurs in $\phi$, and $\overarrow{y}=\{y_{i,j} : 1\leq i\leq n, 0\leq j\leq M_{i}+2\}$ are fresh variables of sort $\s_{1}$. As this is the same as the witness in our proof that $\Taddnc$ must be also finitely witnessable, if $\T$ is so, we already know $\phi$ and 
$\exists\,\overarrow{x}.\exists\,\overarrow{y}.\:\wit_{s}(\phi)$ are $\Taddnc$-equivalent, where $\overarrow{x}=\vars(\wit(\phi))\setminus\vars(\phi)$. So, take a set of variables $V$, an arrangement $\delta_{V}$ over $V$, and a $\Taddnc$-interpretation $\C$ that satisfies $\wit_{s}(\phi)\wedge\delta_{V}$; let $W$ be 
$V\cup\overarrow{x}\cup\overarrow{y}\cup\vars(\phi)$, $\delta_{W}$ be the arrangement on $W$ such that $x$ is related to $y$ iff $x^{\C}=y^{\C}$, and we have that $\C$ satisfies $\wit_{s}(\phi)\wedge\delta_{W}$. Then $\subnc{\C}$ satisfies 
\[\subncf{\phi}{\C}\wedge \psi(\vars(\subncf{\phi}{\C})\cup\overarrow{y})\wedge\Psiv(\overarrow{y})\wedge\Fun(\overarrow{y})\wedge\delta_{W}=\wit(\subncf{\phi}{\C}\wedge\Psiv(\overarrow{y})\wedge\Fun(\overarrow{y}))\wedge\delta_{W},\]
so there must exist a $\T$-interpretation $\A$, given that $\wit$ is a strong witness, such that $\s^{\A}=\vars_{\s}(\wit_{s}(\psi)\wedge\delta_{V})^{\A}$ for each $\s\in S$ (notice the variables in $\wit(\subncf{\phi}{\C}\wedge\Psiv(\overarrow{y})\wedge\Fun(\overarrow{y}))\wedge\delta_{W}$ are the same as the ones in $\wit_{s}(\phi)\wedge\delta_{W}$, which in turn are the same as the ones in $\wit_{s}(\phi)\wedge\delta_{V}$), and $\A$ satisfies $\wit(\subncf{\phi}{\C}\wedge\Psiv(\overarrow{y})\wedge\Fun(\overarrow{y}))\wedge\delta_{V}$. We finally define a $\Taddnc$-interpretation $\D$ such that: $\s^{\D}=\s^{\A}$ for every $\s\in S$; $x^{\D}=x^{\A}$ for every variable $x$; and $s^{\D}(y_{i,j}^{\D})=y_{i, j+1}^{\D}$ for every $1\leq i\leq n$ and
$0\leq j\leq M_{i}+1$, and $s^{\D}(a)=a$ for all other elements of $\s_{1}^{\D}$ (again, notice $s^{\D}(y_{i,M_{i}+2}^{\D})$ must have been defined under these conditions). We indeed have that $s^{\D}$ is a function, since $\A$ satisfies $\Fun(\overarrow{y})$, and $\D$ satisfies $\psiv$, because $\A$ satisfies $\Psiv(\overarrow{y})$.Of course $\D$ satisfies 
$\wit_{s}(\phi)\wedge\delta_{V}$, and in addition has the property that $\s^{\D}=\vars_{\s}(\wit_{s}(\phi)\wedge\delta_{V})^{\D}$ for each $\s\in S$, proving $\wit_{s}$ is a strong witness for $\Taddnc$.

\end{enumerate}

Finally, we deal with convexity. Consider the conjunction of literals $\phi=(y=s(x))\wedge(z=s(y))$, where $x$, $y$ and $z$ are of sort $\s_{1}$, and we have that 
\[\dash_{\Taddnc} \phi\rightarrow (x=y)\vee(x=z)\vee(y=z).\]
To see that, suppose we have a $\Taddnc$-interpretation $\C$ that satisfies $\phi$, but neither $x=y$ nor $x=z$; therefore, $s^{\C}(a)\neq a$ and $s^{\C}(s^{\C}(a))\neq a$, where $a=x^{\C}$. Since $s^{\C}(s^{\C}(a))\neq a$, we must have $s^{\C}(s^{\C}(a))=s^{\C}(a)$, meaning that $\C$ satisfies $y=z$.

However, we do not have that $\Taddnc$ entails $\phi\rightarrow(x=y)$, $\phi\rightarrow(x=z)$ or $\phi\rightarrow(y=z)$, as we assume $\T$ possesses a model $\A$ with $|\s_{1}^{\A}|\geq 2$; say $a, b\in \s_{1}^{\A}$. To understand why, consider the interpretation $\C$ and $\D$ of $\Taddnc$ with: $\s^{\C}=\s^{D}=\s^{\A}$ for every $\s\in S$; $s^{\C}(a)=b$, $s^{\C}(b)=b$, $s^{\D}(a)=b$,  $s^{\D}(b)=a$ and $s^{\C}(c)=s^{\D}(c)=c$ for each $c\in\s_{1}^{\A}\setminus\{a,b\}$; and
\[x^{\C}=a\quad\text{and}\quad y^{\C}=z^{\C}=b,\quad\text{and}\quad x^{\D}=z^{\D}=a\quad\text{and}\quad y^{\D}=b.\]
We see that both $\Taddnc$-interpretations satisfy $\phi$; however, $\C$ does not satisfy $x=y$ and $x=z$, while $\D$ does not satisfy $y=z$ (and also $x=y$).

\end{proof}

\section{Proof of \Cref{lem:mati-f-exists}}

\begin{definition}
\label{k of n}
Let $n>2$.
$\kn(n)$ is defined to be the unique number $k$ such that
$2^{k+1} +1 \leq n \leq 2^{k+2}$.
\end{definition}

\begin{lemma}
\label{k of n is well-defined}
$\kn$ is a well-defined function from $\mathbb{N}\setminus\set{0,1}$ to
$\mathbb{N}$.
\end{lemma}

\begin{proof}
We prove that for each $n>2$ there exists a unique $k$ such that
$2^{k+1} + 1 \leq n \leq 2^{k+2}$.
\begin{description}
\item[Existence]: by induction on $n$. For $n=3$,
take $k=0$ and then
$2^{1}+1\leq 3\leq 2^2$.
For $n>2$,
by the induction hypothesis, there exists a unique
$k'$ such that $2^{k'+1}+1\leq \floor{\dfrac{n}{2}}\leq 2^{k+2}$.
In particular,
$2^{k'+1}<\floor{\dfrac{n}{2}}$, and so
$2^{k+1}=2^{k'+2}<2\cdot\floor{\dfrac{n}{2}}\leq n$, which means
$2^{k+1}+1\leq n$.
Now, if $n$ is even then
$n=2\cdot\floor{\dfrac{n}{2}}$
and then
$n=2\cdot\floor{\dfrac{n}{2}}\leq 2\cdot 2^{k'+2}=2^{k'+3}=2^{k+2}$.
If $n$ is odd then
$n=2\cdot\floor{\dfrac{n}{2}}+1$ and then
$2\cdot\floor{\dfrac{n}{2}}\leq 2\cdot 2^{k'+2}=2^{k'+3}=2^{k+2}$.
Assume for contradiction that
$2^{k+2}< n$. So
$2^{k+2}<2\cdot \floor{\dfrac{n}{2}}+1$, which means that 
$2^{k+2}\leq 2\cdot \floor{\dfrac{n}{2}}$.
But we also have 
$2\cdot \floor{\dfrac{n}{2}}\leq 2^{k+2}$.
This means 
$2\cdot \floor{\dfrac{n}{2}} =  2^{k+2}$ and in particular
$ \floor{\dfrac{n}{2}}= 2^{k+1}$.
But $ \floor{\dfrac{n}{2}}\geq 2^{k+1}+1$, which is a contradiction.

\item[Uniqueness]: if there are $k,k'$ such that
$2^{k+1} +1\leq n \leq 2^{k+2}$ and
$2^{k'+1} +1\leq n \leq 2^{k'+2}$ and $k\neq k'$,
we obtain a contradiction as follows.
W.l.g. assume $k<k'$, and so
$k+1\leq k'$ which means that
$2^{k+2}=2\cdot 2^{k+1}\leq 2\cdot 2^{k'}=2^{k'+1}$. 
Thus we obtain:
$2^{k+1} +1 \leq n \leq 2^{k+2}\leq
2^{k'+1} < 2^{k'+1} + 1\leq n$, which is a
contradiction.
\end{description}

\end{proof}

\begin{remark}
$\kn(n)$ is simply $\ceil{log_{2} n} -2$
\end{remark}

\begin{definition}\label{defining f from F}
Given a function $F:\mathbb{N}\setminus\{0\}\rightarrow\{0,1\}$ with $F(1)=1$, the function 
$\fof{F}:\mathbb{N}\setminus\{0\}\rightarrow\{0,1\}$ is defined by: 
$\fof{F}(1)=F(1)=1$, $\fof{F}(2)=0$,
and, for every $n\in\mathbb{N}\setminus\set{0,1,2}$,
\[\fof{F}(n)=\begin{cases*}
F(n-2^{\kn(n)}) & for $2^{\kn(n)+1}+1\leq n\leq 2^{\kn(n)+1}+2^{\kn(n)}$\\
1 & for $2^{\kn(n)+1}+2^{\kn(n)}+1\leq n\leq 2^{\kn(n)+2}+2^{\kn(n)}-\fof{F}_{1}(2^{\kn(n)+1}+2^{\kn(n)})$;\\
0 & for $2^{\kn(n)+2}+2^{\kn(n)}-\fof{F}_{1}(2^{\kn(n)+1}+2^{\kn(n)})+1\leq n\leq 2^{\kn(n)+2}$.
\end{cases*}\]
where, for $i\in\{0,1\}$,
\[\fof{F}_{i}(n)=|\{m : 1\leq m\leq n\quad\text{and}\quad \fof{F}(m)=i\}|,\].
\end{definition}

\begin{lemma}
\label{fof is well defined}
For every 
$F:\mathbb{N}\setminus\{0\}\rightarrow\{0,1\}$ with $F(1)=1$, we have that 
$\fof{F}$ is a (well-defined) function from $\mathbb{N}\setminus\set{0}\rightarrow\set{0,1}$.
\end{lemma}

\begin{proof}
For each $n\in\mathbb{N}\setminus\set{0}$,
if $n\leq 2$ then $\fof{F}(n)$ is clearly well-defined.
For $n>2$,
$2^{\kn(n)+1}+1 \leq n\leq 2^{\kn(n)+2}$ for a unique $\kn(n)$.
This holds by \Cref{k of n,k of n is well-defined}.
And clearly once $\kn(n)$ is fixed, the definition of the function
distinguishes 3 distinct cases that exactly cover that range.
It is left to make sure that in the first of these cases,
$F$ is defined over $n-2^\kn(n)$.
This holds as
$n\geq 2^{\kn(n)+1}+1>2^{\kn(n)}$, so
$n-2^{\kn(N)}\geq 1$.
\end{proof}

\begin{lemma}
\label{fof is nice}
For every 
$F:\mathbb{N}\setminus\{0\}\rightarrow\{0,1\}$ with $F(1)=1$,
$\fof{F}_{1}(2^{k+1})=2^{k}$ for every $k\in\mathbb{N}$.
\end{lemma}

\begin{proof}
By induction on $k$.
For $k=0$, we have that
$\fof(F)(2^0)=\fof(F)(1)=1$ and 
$\fof(F)(2^1)=\fof(F)(2)=0$.
Thus $\fof(F)_{1}(2^{1})=\fof(F)_{1}(2)=1=2^{0}$.

For $k=1$, we show that $\fof(F)(4)=2$.
Recall that 
$\fof(F)(2^0)=\fof(F)(1)=1$ and 
$\fof(F)(2^1)=\fof(F)(2)=0$.
It is left to compute $\fof(F)(3)$, and $\fof(F)(4)$ (which depend on it).
Clearly, $\kn(3)=0$, and
$2^{1}+1\leq 3\leq 2^1+2^0$, thus
$\fof(F)(3)=F(3-1)=F(2)$.
We distinguish two cases:
\begin{enumerate}
\item If $F(2)=1$ then $\fof(F)(3)=1$ and so $\fof(F)_{1}(3)=2$.
In this case,
$2^{2}+2^{0}-\fof(F)_{1}(2^1+2^0)+1=
4+1-\fof(F)_{1}(3)+1=
4\leq 4 \leq 2^2
$, and so
$\fof(F)(4)=0$ and therefore
$\fof(F)_{1}(4)=2$.
\item If $F(2)=0$, then $\fof(F)(3)=0$and so $\fof(F)_{1}(3)=1$.
In this case,
$2^{1}+2^{0}+1\leq 4 \leq
4=
4+1-1=
4+1-\fof(F)_{1}(3)=
2^{2}+2^{0}-\fof(F)_{1}(2^{1}+2^0)$.
Therefore, 
$\fof(F)(4)=1$ and thus
$\fof(F)_{1}(4)=2$.
\end{enumerate}

Now assume $\fof{F}_{1}(2^{k+1})=2^{k}~(\ast)$.
We prove that
$\fof{F}_{1}(2^{k+2})=2^{k+1}$.
Denote $|\set{2^{k+1}+1\leq n\leq 2^{k+2}\mid f(n)=1}|$ by 
$M$.
Clearly, 
$\fof{F}_{1}(2^{k+2})=\fof{F}_{1}(2^{k+1})+M$.
By $(\ast)$, we have
$\fof{F}_{1}(2^{k+2})=2^{k}+M$.
Since we would like to prove that 
$\fof{F}_{1}(2^{k+2})=2^{k+1}$, it suffices to show that
$M=2^{k}$.

Let $m$ be the number of numbers $n$ from $2^{k+1}+1$ to $2^{k+1}+2^{k}$ (inclusive) such that $F(n-2^{k})=1$.
Notice that for each such $n$, we have that $\kn(n)=k$.
Thus, $\fof{F}_{1}(2^{k+1}+2^{k})=\fof{F}_{1}(2^{k+1})+m=2^{k}+m$.
Hence by the definition of $\fof{F}$, the number of numbers $n$
from $2^{k+1}+2^{k}+1$ to $2^{k+2}+2^{k}-\fof{F}_{1}(2^{k+1} + 2^{k})$ is
$(2^{k+2}+2^{k}-\fof{F}_{1}(2^{k+1} + 2^{k})) - (2^{k+1}+2^{k}+1) + 1=
(2^{k+2}+2^{k}-2^{k}-m) - (2^{k+1}+2^{k}+1) + 1=
2^{k+2}-m - 2^{k+1} - 2^{k} - 1 + 1 =
2^{k}(4 - 2 - 1) - m= 
2^{k}\cdot 1 - m=
2^{k} - m$.
Also, from the definition of $\fof{F}$, we have that the number of numbers $n$ from $2^{k+2}+2^{k}-\fof{F}_{1}(2^{k+1}+2^{k})+1$ to
$2^{k+2}$ (inclusive) is $0$.
In total, we get that 
$M=m+2^{k}-m+0=2^{k}$.
\end{proof}

Intuitively, suppose $\fof{F}$ has been defined for $n$ up to a certain power of two $2^{k+1}$. Because we want $\fof{F}$ to be equal to $F$ (not necessarily given the same arguments, but only in terms of the sequence of $0$'s and $1$'s they output) as often as possible, we define $\fof{F}(n)=F(n-2^{k})$ for the following quarter of the next power of two, that is $2^{k+1}+1\leq n\leq 2^{k+1}+2^{k}$. Notice that we are continuing from where we stopped last: that is, the last value for which we defined $\fof{F}$ as a function of $F$ was $m=2^{(k-1)+1}+2^{k-1}$, when $\fof{F}(m)$ equalled $F(2^{k}+2^{k-1}-2^{k-1})=F(2^{k})$, while the next value of $n$ is $n=2^{k+1}+1$, when we get $\fof{F}(n)$ equal to $F(2^{k+1}+1-2^{k})=F(2^{k}+1)$. 

However, after having done that, we must address that, at the same time, we want the number of times that $\fof{F}$ equals one is, more or less, the same as the number of times it equals zero; we do that by demanding that $\fof{F}$ equals $1$ for numbers up to $2^{k+2}$ exactly $2^{k+1}$ times, that is, $\fof{F}_{1}(2^{k+2})=2^{k+1}$. In order to accomplish that, we make $\fof{F}(n)$ equal to $1$ for $n>2^{k+1}+2^{k}$ until $\fof{F}_{1}(n)$ becomes $2^{k+1}$; after that, we let $\fof{F}$ simply equal zero until the next power of two.  

For an example, take a certain $F:\mathbb{N}\setminus\{0\}\rightarrow\{0,1\}$ with $F(1)=1$; say, with $F(1)=1$, $F(2)=0$, $F(3)=0$, $F(4)=0$, $F(5)=0$, $F(6)=1$, $F(7)=0$ and $F(8)=1$. We can then draw the elucidating diagram in \Cref{diagram for f}.

\begin{figure}[t]\label{diagram for f}
\centering
\adjustbox{scale=0.7,center}{%
\begin{tikzcd}[row sep=0.8em]
    n &[-2em] 1 &[-2em] 2 &[-2em] 3 &[-2em] 4 &[-2em] 5 &[-2em] 6 &[-2em] 7 &[-2em] 8 &[-2em] 9 &[-2em] 10 &[-2em] 11 &[-2em] 12 &[-2em] 13 &[-2em] 14 &[-2em] 15 &[-2em] 16 &[-2em] \cdots\\
    F(n) & 1 \arrow[d] & 0\arrow[dr] & 0\arrow[drr] & 0\arrow[drr] & 0\arrow[drrrr] & 1\arrow[drrrr] & 0\arrow[drrrr] & 1\arrow[drrrr] & \cdots &  &  &  &  &  &  &  &\\
    \fof{F}(n) & 1 & \textcolor{ForestGreen}{0} & 0 & \textcolor{red}{1} & 0 & 0 & \textcolor{red}{1} & \textcolor{red}{1} & 0 & 1 & 0 & 1 & \textcolor{red}{1} & \textcolor{red}{1} & \textcolor{ForestGreen}{0} & \textcolor{ForestGreen}{0} & \cdots\\
    \fof{F}_{1}(n) & 1 & 1 & 1 & 2 & 2 & 2 & 3 & 4 & 4 & 5 & 5 & 6 & 7 & 8 & 8 & 8 & \cdots\\
\end{tikzcd}}
\caption{A diagram for the function $f$.}
\end{figure}
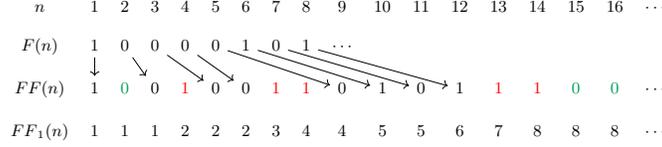

\begin{lemma}\label{ceil function}
Given a function $F:\mathbb{N}\setminus\{0\}\rightarrow\{0,1\}$ with $F(1)=1$, if we define $f$ to be $\fof{F}$,
then
for every $n\geq 2$ it is true that $F(n)=f(n+2^{\kn(n)}+1)$.
\end{lemma}

\begin{proof}
Since 
$2^{\kn(n)+1}+1\leq n \leq 2^{\kn(n)+2}$, we also have
$2^{(\kn(n)+1})+1\leq n + 2^{\kn(n)+1}\leq 2^{(\kn(n)+1)+1}+2^{\kn(n)+1}\leq 2^{(\kn(n)+1)+2}$.
Hence $\kn(n + 2^{\kn(n)+1})=\kn(n)+1$.
By the definition of $\fof{F}$, we therefore have
$F(n)=F(n+2^{\kn(n)+1}-2^{\kn(n)+1)}=\fof{F}(n+2^{\kn(n)+1})$.

\end{proof}

\lemmatifexists*

\begin{proof}
We start by taking a non-computable function $F:\mathbb{N}\setminus\{0\}\rightarrow\{0,1\}$ with $F(1)=1$ (for example: given an enumeration $\{T_{n}:n\in\mathbb{N}\setminus\{0\}\}$ of all Turing machines, $F(n)=1$ iff $T_{n}$ halts, where we can change the enumeration so that the first Turing machine halts).

Using \Cref{defining f from F}, we can define $f:\mathbb{N}\setminus\{0\}\rightarrow\{0,1\}$ as equal to $\fof{F}$. 
This function is not computable since otherwise, by \Cref{ceil function}, $F$ would be computable (as $\kn$ is computable as well). And, from \Cref{defining f from F}, if we make $k\geq 2$ equal to $k^{\prime}+1$ for $k^{\prime}\in\mathbb{N}\setminus\{0\}$, 
$f_{1}(2^{k})=\fof{F}_{1}(2^{k^{\prime}+1})=2^{k^{\prime}}$ (from \Cref{fof is nice}), 
and $f_{0}(2^{k})=2^{k}-f_{1}(2^{k})=2^{k^{\prime}}$; for $k=1$, we know that $f(1)=1$ and $f(2)=0$ by \Cref{defining f from F} of $\fof{F}$ (and thus $f$), hence implying $f_{0}(2^{0})=f_{1}(2^{0})$. To summarize, for any $k\in\mathbb{N}\setminus\{0\}$, $f_{0}(2^{k})=f_{1}(2^{k})$.
\end{proof}

\section{Proof of \Cref{mcofTsMisnotcomputable}}

\begin{lemma}\label{mincard of TsM is all defined}
    Every quantifier-free $\TsM$-satisfiable formula is satisfied by a finite $\TsM$-interpretation.
\end{lemma}

\begin{proof}
    Suppose $\phi$ is a quantifier-free formula, and let $\A$ be a $\TsM$-interpretation that satisfies $\phi$: we may assume that it is infinite. The set
    $\{\alpha : \text{$\alpha$ is a term in $\phi$}\}$ is finite, and therefore so is $A=\{\alpha : \text{$\alpha$ is a term in $\phi$}\}^{\A}$. Let 
    \[m_{0}=|\{a\in A: s^{\A}(a)\neq a\}|\quad\text{and}\quad m_{1}=|\{a\in A : s^{\A}(a)=a\}|,\]
    and take a $k\in\mathbb{N}\setminus\{0\}$ such that $2^{k}>\max\{m_{0}, m_{1}\}$. Take as well sets $B$ and $C$ with, respectively, $2^{k}-m_{1}$ and $2^{k}-m_{0}$ elements, disjoint from 
    $A$, and we define a $\TsM$-interpretation $\B$ with: $\s^{\B}=A\cup B\cup C$;
    \[s^{\B}(a)=\begin{cases*}
        a & if $a\in A$ and $s^{\A}(a)=a$, or $a\in B$;\\ 
        s^{\A}(a) & if $s^{\A}(a)\neq a$ but $s^{\A}(a)\in A$;\\
        \text{any element in $B$} & if $s^{\A}(a)\neq a$ and $s^{\A}(a)\notin A$, or $a\in C$        
    \end{cases*};\]
    and $x^{\B}=x^{\A}$ for any variable $x$ in $\phi$, and arbitrary otherwise. $\B$ has $m_{1}$ elements in $A$, plus all $2^{k}-m_{1}$ of those in $B$, 
    satisfying $s^{\B}(a)=a$, adding to a total of $2^{k}$; and $m_{0}$ elements in $A$, plus all $2^{k}-m_{1}$ of the 
    elements in $C$, satisfying $s^{\B}(a)\neq a$, to a grand total of $2^{k}$, meaning $\B$ is a $\TsM$-interpretation with $2^{k+1}$ elements, and thus 
    finite. 
    
    Furthermore, let $s^{j}(x)$ be a term in $\phi$: if $j=0$, $(s^{j}(x))^{\B}=x^{\B}=x^{\A}=(s^{j}(x))^{A}$, so assume the result holds for an arbitrary $j$; in that case, if $s^{j+1}(x)$ is still a term in $\phi$,
    \[(s^{j+1}(x))^{\B}=s^{\B}\big((s^{j}(x)^{\B}\big)=s^{\A}\big((s^{j}(x)^{\A}\big)=(s^{j+1}(x))^{\A},\]
    proving that for any terms 
    $\alpha$ in $\phi$, $\alpha^{\B}=\alpha^{\A}$; since $\phi$ is a quantifier-free formula in a 
    signature without predicates, all of whose terms receive the same value in either $\A$ or $\B$, and $\A$ satisfies $\phi$, we have that $\B$ also satisfies $\phi$, finishing the proof.
\end{proof}

\mctsmnc*

\begin{proof}
We start by proving that $f(n+1)=1$ iff $\mc(\phi_{n})=n+1$, for 
\[\phi_{n}=\bigwedge_{i=1}^{f_{1}(n)+1}\big(s(x_{i})=x_{i}\big)\wedge \bigwedge_{1\leq i<j\leq f_{1}(n)+1}\neg(x_{i}=x_{j})\]
a formula that is only true in a model $\B$ when there exist at least $f_{1}(n)+1$ elements in $\B$ satisfying $s^{\B}(a)=a$.

Notice that the theory $\TsM$ has models of all finite, non-zero cardinalities: indeed, given a $n\in\mathbb{N}\setminus\{0\}$, one 
such model $\A$ has domain $\{a_{1}, \ldots, a_{n}\}$, with $s^{\A}(a_{i})=a_{i}$ for each $1\leq i\leq f_{1}(n)$ (remember $f_{1}(n)\geq 1$
for all $n\in\mathbb{N}\setminus\{0\}$), and $s^{\A}(a_{j})=a_{1}$ for each $f_{1}(n)<j\leq n$, if there are 
any such $j$ (and as long as $n>1$ one has $f_{1}(n)<n$). Notice as well that, for all $p\geq q$, 
$\psi^{=}_{=p}\rightarrow\psi^{=}_{\geq q}$
and
$\psi^{\neq}_{=p}\rightarrow\psi^{\neq}_{\geq q}$
and so a model $\A$ of $\TsM$ must satisfy either $\psi^{=}_{=f_{1}(k)}\wedge\psi^{\neq}_{=f_{0}(k)}$, for some $k\in\mathbb{N}\setminus\{0\}$, or $\psi^{=}_{\geq f_{1}(k)}\wedge\psi^{\neq}_{\geq f_{0}(k)}$ for all $k\in\mathbb{N}\setminus\{0\}$, in 
which case $\A$ is infinite; if our model $\A$ is finite, it must then satisfy $\psi^{=}_{=f_{1}(k)}\wedge\psi^{\neq}_{=f_{0}(k)}$, for some $k\in\mathbb{N}\setminus\{0\}$, and therefore have exactly $f_{1}(k)$ elements satisfying 
$s^{\A}(a)=a$ and $f_{0}(k)$ elements satisfying $s^{\A}(a)\neq a$. Since an element $a$ of $\A$ must satisfy either $s^{\A}(a)=a$ or $s^{\A}(a)\neq a$ and never both of them, we have that $\A$ must have precisely $f_{1}(k)+f_{0}(k)=k$ elements, and so a model $\A$ of $\TsM$ has $k$ elements in its domain iff it contain $f_{1}(k)$ elements satisfying 
$s^{\A}(a)=a$ and $f_{0}(k)$ elements satisfying $s^{\A}(a)\neq a$.

 So, let us prove the implication from left to right in the biconditional $f(n+1)=1\Leftrightarrow \mc(\phi_{n})=n+1$. If $\A$ is a model of $\TsM$ that satisfies $\phi_{n}$
 and has minimal (finite, because of \Cref{mincard of TsM is all defined}) cardinality $\mc(\phi_{n})$ among the models of $\TsM$ that satisfy this formula, we have that $\A$ has at least $f_{1}(n)+1$ elements $a$ that satisfy $s^{\A}(a)=a$, meaning $f_{1}(\mc(\phi_{n}))\geq f_{1}(n)+1$. 
 Since we are assuming $f(n+1)=1$, $f_{1}(n+1)=f_{1}(n)+1$, and since any model of $\TsM$ with $n+1$ elements must have $f_{1}(n+1)=f_{1}(n)+1$ elements satisfying $s^{\A}(a)=a$, and thus actually satisfy $\phi_{n}$, we have $n+1\geq \mc(\phi_{n})$. If $\mc(\phi_{n})$ were strictly less than $n$, we would get $n\geq \mc(\phi_{n})$, and since $f_{1}$ is non-decreasing, $f_{1}(n)\geq f_{1}(\mc(\phi_{n}))\geq f_{1}(n)+1$, what is absurd: we must have instead $\mc(\phi_{n})=n+1$.
 
 Reciprocally, assume $\mc(\phi_{n})=n+1$, and we know that some model $\A$ of $\TsM$ with $n+1$ elements satisfies $\phi_{n}$, and therefore has at least $f_{1}(n)+1$ elements that satisfy $s^{\A}(a)=a$, from what follows that $f_{1}(n+1)\geq f_{1}(n)+1$, and thus the two values are equal (since $f_{1}(n)$ and $f_{1}(n+1)$ can only differ by $0$ or $1$). So $f(n+1)=1$.

To summarize, if $\TsM$ has a computable $\mc$ function and we know the values of $f(1),\ldots, f(n)$ (and therefore of $f_{0}(n)$ and $f_{1}(n)$), we can calculate $f(n+1)$ algorithmically as well.
\end{proof}

\section{Proof of \Cref{Decidability of TsM}}

\thmtrisattsm*

\begin{proof}
If the quantifier-free formula $\phi$ is satisfiable, then it must be satisfied by some $\T$-interpretation $\A$, where $\T$ is the theory with all $\Sigma_{s}$-structures as models, axiomatized by the 
empty set. Take then enumerable sets $A$ and $B$ disjoint from $\s^{\A}$ and each other, and build the interpretation $\B$ with: $\s^{\B}=\s^{\A}\cup A\cup B$;
\[\s^{\B}(a)=\begin{cases*}
    s^{\A}(a) & if $a\in \s^{\A}$;\\
    a & if $a\in A$;\\
    \text{any element of $A$} & if $a\in B$;
\end{cases*}\]
and $x^{\B}=x^{\A}$ for all variables $x$. It is obvious $\B$ is a $\TsM$-interpretation, since it has infinite 
elements satisfying each condition, $s^{\B}(a)=a$ or $s^{\B}(a)\neq a$ (respectively, all of those in $A$ and $B$). Furthermore, 
for any term $\alpha$ in $\phi$, $\alpha^{\B}=\alpha^{\A}$, since $x^{\B}=x^{\A}$ for every variable $x$, and, assuming as induction hypothesis that $(s^{j}(x))^{\B}=(s^{j}(x))^{\A}$ for some $j$, 
\[(s^{j+1}(x))^{\B}=s^{\B}\big((s^{j}(x))^{\B}\big)=s^{\A}\big((s^{j}(x))^{\A}\big)=(s^{j+1}(x))^{\A}.\]
Since $\phi$ is a quantifier-free formula in a signature without predicates, we get that $\phi$ is satisfied by the $\TsM$-interpretation $\B$ (given that is satisfied by $\A$).
\end{proof}



\section{Proofs for \Cref{tab-summary}}
\subsection{$\Tgeqn$}

$\Tgeqn$ is defined by a single axiom, which has the form
$\exists\,\overarrow{x}.\:\psi$ for a quantifier-free formula $\psi$.
Such theories are called \emph{existential} in \cite{DBLP:journals/jar/ShengZRLFB22} and are proven
there to be strongly polite.
Thus we obtain the next lemma:

\begin{lemma}\label{Tgeqn is smooth}
$\Tgeqn$ is smooth, and thus stably-infinite.
It is also
strongly finitely witnessable, and thus finitely witnessable.
\end{lemma}
It is left to show that it is convex:
\begin{lemma}
$\Tgeqn$ is convex.
\end{lemma}

\begin{proof}
Since $\Tgeqn$ is stably infinite from \Cref{Tgeqn is smooth}, \Cref{SI empty theories are convex} guarantees it is also convex.
\end{proof}


 \subsection{$\Tinfty$}

\begin{lemma}
\label{lem:tinftysmooth}
$\Tinfty$ is smooth, and thus stably-infinite.
\end{lemma}

\begin{proof}
Given any $\Tinfty$-interpretation $\A$, the theory is seen to be smooth since every larger interpretation $\B$ must necessarily be a $\Tinfty$-interpretation as well.
\end{proof}

\begin{lemma}\label{Tinfty is not FW}
$\Tinfty$ is not finitely witnessable, and thus not strongly finitely witnessable.
\end{lemma}

\begin{proof}
Suppose $\wit$ is a witness. For any variable $x$, $x=x$ is satisfied by all $\Tinfty$-interpretations $\A$, and therefore $\wit(x=x)$ is satisfied by some 
$\A^{\prime}$, where we only change the value given by $\A$ to variables in 
$\vars(\wit(x=x))\setminus\vars(x=x)$. There must then exist a $\Tinfty$-interpretation $\B$ that satisfies $\wit(x=x)$, where $\sigma^{\B}=\vars(\wit(x=x))^{\B}$. 
Of course, this is impossible: $\vars(\wit(x=x))$ must necessarily be finite, and therefore so is $\vars(\wit(x=x))^{\B}$, while $\B$ is a model of 
$\Tinfty$ if and only if its domain is infinite.
\end{proof}

\begin{lemma}
$\Tinfty$ is convex.
\end{lemma}

\begin{proof}
A corollary of \Cref{SI empty theories are convex,lem:tinftysmooth}.
\end{proof}


\subsection{$\Teven$}

\begin{lemma}\label{Teven is SI}
$\Teven$ is stably-infinite.
\end{lemma}

\begin{proof}
Follows from \Cref{Infinite Model => SI}, given that $\Teven$ is defined on a signature with only one sort and has infinite models.
\end{proof}

\begin{lemma}
$\Teven$ is finitely witnessable, not strongly finitely witnessable and not smooth.
\end{lemma}

\begin{proof}
See Section $3.4$ of \cite{SZRRBT-21}
\end{proof}

\begin{lemma}
$\Teven$ is convex.
\end{lemma}

\begin{proof}
Since $\Teven$ is stably-infinite, from \Cref{Teven is SI}, \Cref{SI empty theories are convex} guarantees $\Teven$ is convex.
\end{proof}


\subsection{$\Tninfty$}

\begin{lemma}\label{Tninfty is SI}
$\Tninfty$ is stably-infinite.
\end{lemma}

\begin{proof}
Given \Cref{Infinite Model => SI}, and the facts that $\Tninfty$ is a theory over the empty signature with only one sort and infinite models, the result follows.
\end{proof}

\begin{lemma}
$\Tninfty$ is not smooth.
\end{lemma}

\begin{proof}
Notice that $\Tninfty$ has models with $n$ elements in their domains, but no models with $m>n$ elements.
\end{proof}

\begin{lemma}\label{Tninfty is not FW}.
$\Tninfty$ is not finitely witnessable, and thus not strongly finitely witnessable.
\end{lemma}

\begin{proof}
Suppose we have a witness $\wit$, and we shall use the quantifier-free formula
\[\phi=\bigwedge_{1\leq i<j\leq n+1}\neg(x_{i}=x_{j});\]
since $\phi$ is satisfied by some infinite $\Tninfty$-interpretations, so is $\wit(\phi)$. There must then exist a $\Tninfty$-interpretation $\A$ that satisfies $\wit(\phi)$ ( and so $\phi$) and $\s^{\A}=\vars(\wit(\phi))^{\A}$. This is, of course, absurd: if $\A$ satisfies $\phi$, it has at least $n+1$ elements in its domain, while any finite $\Tninfty$-interpretation must have precisely $n$ elements in its domain (recall that $\A$ is finite).
\end{proof}

\begin{lemma}
$\Tninfty$ is convex.
\end{lemma}

\begin{proof}
Combine \Cref{SI empty theories are convex} with the fact that $\Tninfty$ is stably-infinite, present in \Cref{Tninfty is SI}.
\end{proof}


\subsection{$\Tone$}

\begin{lemma}
$\Tone$ is not stably-infinite, and thus not smooth.
\end{lemma}

\begin{proof}
Obvious, given it has finite models, but no infinite ones.
\end{proof}

\begin{lemma}
$\Tone$ is strongly finitely witnessable, and thus finitely witnessable.
\end{lemma}

\begin{proof}
Trivially, $\wit(\phi)=\phi$ is a strong witness, given that if $\phi\wedge\delta_{V}$ is satisfied by the $\Tone$-interpretation $\A$, $\A$ is the trivial model and we already have $\vars(\phi\wedge\delta_{V})^{\A}=\s^{\A}$.
\end{proof}

\begin{lemma}
$\Tone$ is convex.
\end{lemma}

\begin{proof}
For any pair of variables $x$ and $y$, trivially one finds that $\dash_{\Tone}x=y$. So, whenever $\dash_{\Tone}\phi\rightarrow\bigvee_{i=1}^{n}x_{i}=y_{i}$, for $\phi$ a conjunction of literals, we have $\dash_{\Tone}x_{i}=y_{i}$, for any $1\leq i\leq n$.
\end{proof}


\subsection{$\Tleqn$}

\begin{lemma}
$\Tleqn$ is not stably-infinite, and thus not smooth.
\end{lemma}

\begin{proof}
Fairly obvious, since it has finite models, but no infinite ones.
\end{proof}

\begin{lemma}
$\Tleqn$ is strongly finitely witnessable, and thus finitely witnessable.
\end{lemma}

\begin{proof}
Consider the function $\wit$ from quantifier-free formulas into themselves such that $\wit(\phi)=\phi$, which is obviously computable. Since $\overarrow{x}=\vars(\wit(\phi))\setminus\vars(\phi)$ is empty, trivially $\phi$ and $\exists\, \overarrow{x}.\:\wit(\phi)=\wit(\phi)$ are $\Tleqn$-equivalent. Now, given a set of variables $V$ and an arrangement $\delta_{V}$ on $V$, suppose that $\A$ is a $\Tleqn$-interpretation that satisfies $\wit(\phi)\wedge\delta_{V}$. Let $W=\vars(\wit(\phi)\wedge\delta_{V})$ and take the $\Tleqn$-interpretation $\B$ with domain $W^{\A}$, and $x^{\B}=x^{\A}$ for every $x\in W$ (and arbitrary otherwise). Of course $\B$ is indeed a $\Tleqn$-interpretation, since $|\s^{\A}|\leq n$, and $\s^{\B}\subseteq \s^{\A}$. Furthermore, since all atomic subformulas of $\wit(\phi)\wedge\delta_{V}$, are necessarily equalities $x=y$ with both $x$ and $y$ in $W$, and $x^{\B}=x^{\A}$ and $y^{\B}=y^{\A}$, $\wit(\phi)\wedge\delta_{V}$ receives the same truth value in $\A$ and $\B$.
\end{proof}

\begin{lemma}
If $n>1$, $\Tleqn$ is not convex.
\end{lemma}

\begin{proof}
Given variables $x_{1}$ through $x_{n+1}$, for any conjunction of literals $\phi$ which is a tautology (such as $x=x$), 
\[\dash_{\Tleqn}\phi\rightarrow \bigvee_{1\leq i<j\leq n+1}x_{i}=x_{j}\]
by the pigeonhole principle. But, since $n>1$, we cannot have $\dash_{\Tleqn}\phi\rightarrow x_{i}=x_{j}$ for any pair $1\leq i<j\leq n+1$, since we can always set, in a $\Tleqn$-interpretation $\A$, $x_{i}^{\A}\neq x_{j}^{\A}$.
\end{proof}


\subsection{$\Tmn$}\label{Tmn}

\begin{lemma}
$\Tmn$ is not stably-infinite, and thus not smooth.
\end{lemma}

\begin{proof}
$\Tmn$ has finite models, but no infinite ones, so it cannot be stably-infinite.
\end{proof}

 \begin{lemma}
If $m>1$ and $n>1$, $\Tmn$ is finitely witnessable.
\end{lemma}

\begin{proof}
Without loss of generality, assume that $m\geq n$. Consider a quantifier-free formula $\phi$, fresh variables $x_{1}$ through $x_{m}$, and define the witness $\wit(\phi)=\phi\wedge\bigwedge_{i=1}^{m}x_{i}=x_{i}$, 
obviously computable. Since $\phi$ and $\wit(\phi)$ are equivalent, given that $\bigwedge_{i=1}^{m}x_{i}=x_{i}$ is a tautology, we have that $\phi$ and $\exists\,\overarrow{x}.\:\wit(\phi)$ are $\Tmn$-equivalent, for $\overarrow{x}=\{x_{1}, ... , x_{m}\}=\vars(\wit(\phi))\setminus\vars(\phi)$.

So suppose that $\wit(\phi)$ is satisfied by a $\Tmn$-interpretation $\A$: let $W=\vars(\wit(\phi))$, $V=\vars(\phi)$ and $k=m-|V^{\A}|$; we define a $\Tmn$-interpretation $\B$ with domain $V^{\A}\cup\{a_{1}, ... , a_{k}\}$, for $a_{i}$ not in $\sigma^{\A}$ (so that 
$|\sigma^{\B}|=m$), $x^{\B}=x^{\A}$ for $x\in V$, and $\{x^{\B} : x\in\overarrow{x}\}=\sigma^{\B}$ (what is possible, given 
$\overarrow{x}$ and $\sigma^{\B}$ both have $m$ elements. Since $\A$ and $\B$ coincide on the variables of $\phi$, the latter satisfies 
$\phi$, and therefore $\wit(\phi)$. Finally, given that $W^{\B}=\sigma^{\B}$, we obtain that $\wit$ is indeed a witness, and $\Tmn$ is finitely witnessable.
\end{proof}

\begin{lemma}
If $|m-n|>1$, $T_{m,n}$ is not strongly finitely witnessable. 
\end{lemma}

\begin{proof}
Suppose $\wit$ is a strong witness, let $x$ be a variable, and take the model $\A^{\prime}$ of $\Tmn$ with $n$ elements in its domain: since $x=x$ is satisfied by any interpretation over this structure, there must be an 
interpretation $\A$ over $\A^{\prime}$ that satisfies $\wit(x=x)$. Let $V=\vars(\wit(x=x))$ and take the equivalence $E$ on $V$ such that $xEy$ iff $x^{\A}=y^{\A}$, with corresponding arrangement $\delta_{V}$, and we have two cases to consider.
\begin{enumerate}
\item If $|V/E|<n$: since $\A$ clearly satisfies $\wit(x=x)\wedge\delta_{V}$, we must have a $\Tmn$-interpretation $\B$ that satisfies $\wit(x=x)\wedge\delta_{V}$ with $\sigma^{\B}=V^{\B}$. Hence $|\sigma^{\B}|<n$, and therefore $\B$ cannot be a $\Tmn$-interpretation, leading to a contradiction.

\item If $|V/E|=n$, take a variable $y\notin\vars(\wit(x=x))$, define $W=V\cup\{y\}$ and consider the equivalence relation $F$ on $W$ such that $xFy$ iff $xEy$ or $x=y$. We state that $\wit(x=x)\wedge \delta_{V}$ is then satisfied by $\B$, a $\Tmn$-interpretation with 
$\sigma^{\B}=\sigma^{\A}\cup\{a_{1}, ... , a_{m-n}\}$ (where $\{a_{1}, ... , a_{m-n}\}\cap \sigma^{\A}=\emptyset$, what implies $\B$ has $m$ elements),  $x^{\B}=x^{\A}$ for every $x\in V$, and $y^{\B}=a_{1}$.

We see that $\B$ satisfies $\wit(x=x)\wedge\delta_{V}$, since this formula is true in $\A$, all its atomic subformulas are equalities of 
variables, and $\A$ and $\B$ coincide on $\vars(\wit(x=x))$; and $\B$ must also satisfy $\wit(\phi)\wedge\delta_{W}$, since $y^{\B}\neq x^{\B}$ for every $x\in V$. Since we are under the assumption that 
$\wit$ is a strong witness, there must exist a $\Tmn$-interpretation 
$\C$ that satisfies $\wit(x=x)\wedge\delta_{W}$ with $\sigma^{\C}=W^{\C}$, and this is impossible: since $|V/E|=n$, 
$|W/F|=n+1$, meaning that if $\C$ satisfies $\delta_{W}$, $\sigma^{\C}=W^{\C}$ must have exactly $n+1$ elements, what
contradicts the fact that $\C$ is a model of $\Tmn$ and therefore should have either $n$ or $m\geq n+2$ elements in its domain.\qedhere
\end{enumerate}
\end{proof}

\begin{lemma}
If $m>1$ and $n>1$, $\Tmn$ is not convex.
\end{lemma}

\begin{proof}
Under the assumption that $m>1$ and $n>1$, $\Tmn$ has no models of cardinality $1$, and by \Cref{Barrett's theorem on convexity} it cannot be convex.
\end{proof}


\subsection{$\Ttwo$}

The following was proven in \cite{SZRRBT-21}:

\begin{lemma}
$\Ttwo$ is smooth, finitely witnessable, but not strongly finitely witnessable w.r.t. $\{\s, \s_{2}\}$.
\end{lemma}


\begin{lemma}
$\Ttwo$ is convex w.r.t. $\{\s, \s_{2}\}$.
\end{lemma}

\begin{proof}
Follows from \Cref{SI empty theories are convex}.
\end{proof}

\subsection{$\Toddtwo$}

\begin{lemma}
$\Toddtwo$ is not stably-infinite, and thus not smooth, w.r.t. $\{\s, \s_{2}\}$.
\end{lemma}

\begin{proof}
$\Toddtwo$ has a model $\A$ where $|\s^{\A}|=1$ and $|\s_{2}^{\A}|=\omega$, but no models $\B$ where both $\s^{\B}$ and $\s_{2}^{\B}$ are infinite.
\end{proof}

\begin{lemma}
$\Toddtwo$ is finitely witnessable w.r.t. $\{\s, \s_{2}\}$.
\end{lemma}

\begin{proof}
So, take a quantifier-free $\Sigma_{2}$-formula $\phi$, and we wish to show
\[\wit(\phi)=\phi\wedge (x=x)\wedge(y=y),\]
where $x$ is a fresh variable of sort $\s$ and $y$ is a fresh variable of sort $\s_{2}$, is a witness for $\Toddtwo$. Of course, if $\overarrow{x}=\vars(\wit(\phi))\setminus\vars(\phi)=\{x, y\}$, $\phi$ and $\exists\,\overarrow{x}.\:\wit(\phi)$ are $\Toddtwo$-equivalent since $\phi$ and $\wit(\phi)$ are, themselves, equivalent.

So assume that a $\Toddtwo$-interpretation $\A$ satisfies $\wit(\phi)$, let $V=\vars_{\s_{2}}(\phi)$ and $E$ be the equivalence on $V$ such that $y_{1}Ey_{2}$ iff $y_{1}^{\A}=y_{2}^{\A}$. There are then two cases to consider.
\begin{enumerate}
\item If $|V/E|$ is odd, we take the $\Toddtwo$-interpretation $\B$ with $\s^{\B}=\s^{\A}$, $\s_{2}^{\B}=V^{\A}$ and $z^{\B}=z^{\A}$ for every variable $z$ except $y$, making instead $y^{\B}$ equal to any value in $\s_{2}^{\B}$. Not only $\B$ satisfies $\wit(\phi)$, but $\s^{\B}=\vars_{\s}(\wit(\phi))^{\B}$ and $\s_{2}^{\B}=\vars_{\s_{2}}(\wit(\phi))^{\B}$. 
\item If $|V/E|$ is even, we take the $\Toddtwo$-interpretation $\B$ with: $\s^{\B}=\s^{\A}$; $\s_{2}^{\B}=V^{\A}\cup\{a\}$, for an element $a\notin \s_{2}^{\A}$, being therefore $|\s_{2}^{\B}|$ odd; and $z^{\B}=z^{\A}$ for every variable $z$ except $y$, where $y^{\B}=a$. Then $\B$ satisfies $\wit(\phi)$, and $\s^{\B}=\vars_{\s}(\wit(\phi))^{\B}$ and $\s_{2}^{\B}=\vars_{\s_{2}}(\wit(\phi))^{\B}$..\qedhere
\end{enumerate}
\end{proof}

\begin{lemma}
$\Toddtwo$ is not strongly finitely witnessable w.r.t. $\{\s, \s_{2}\}$.
\end{lemma}

\begin{proof}
Suppose, for a proof by contradiction, that $\Toddtwo$ does have a strong witness $\wit$. Take the $\Toddtwo$-interpretation $\A$ with one element of sort $\s$ and one element of sort $\s_{2}$ and a variable $y$ of sort $\s_{2}$: $\A$ satisfies $y=y$, and thus $\exists\,\overarrow{x}.\:\wit(y=y)$, for $\overarrow{x}=\vars(\wit(y=y))\setminus\vars(y=y)$. Of course, there is an interpretation $\A^{\prime}$, differing from $\A$ at most on the value assigned to the variables in $\overarrow{x}$, that satisfies $\wit(y=y)$; there must then exist a $\Toddtwo$-interpretation $\B$ that satisfies $\wit(y=y)$ and $\s^{\B}=\vars_{\s}(\wit(y=y))^{\B}$ and $\s_{2}^{\B}=\vars_{\s_{2}}(\wit(y=y))^{\B}$.

Let $V$ be $\vars_{\s_{2}}(\wit(y=y))$, and $E$ be the equivalence relation on $V$ such that $y_{1}Ey_{2}$ iff $y_{1}^{\B}=y_{2}^{\B}$, with corresponding arrangement $\delta_{V}$. Clearly $\B$ satisfies $\delta_{V}$ and $V/E$ has an odd number of equivalence classes; let $y_{0}$ be a fresh variable of sort $\s_{2}$, and take the equivalence $F$ on $V\cup\{y_{0}\}$ such that $y_{1}Fy_{2}$ iff $y_{1}=y_{2}$ or $y_{1}Ey_{2}$, with corresponding arrangement $\delta_{W}$. First, we state that $\wit(y=y)\wedge\delta_{W}$ is $\Toddtwo$-satisfiable. 

In fact, take the interpretation $\C$ with: $\s^{\C}=\s^{\B}$, both then with cardinality $1$; $\s_{2}^{\C}=\s_{2}^{\B}\cup\{a,b\}$, where $a, b\notin\s_{2}^{\B}$, and thus $|\s_{2}^{\C}|=|\s_{2}^{\B}|+2$, an odd number; and $x^{\C}=x^{\B}$ for all variables $x$ except $y_{0}$, where we use instead $y_{0}^{\C}=a$. Not only $\C$ is a $\Toddtwo$-interpretation, but one easily sees it satisfies $\wit(y=y)\wedge\delta_{W}$, meaning we should be able to find a $\Toddtwo$-interpretation $\D$ that satisfies $\wit(y=y)\wedge\delta_{W}$ with $\s^{\D}=\vars_{\s}(\wit(y=y)\wedge\delta_{W})^{\D}$ and$\s_{2}^{\D}=\vars_{\s_{2}}(\wit(y=y)\wedge\delta_{W})^{\D}$, but this is impossible: $\vars_{\s_{2}}(\wit(y=y)\wedge\delta_{W})=W$, but $|W/F|$ is an even number, what would force $\s_{2}^{\D}$ to have an even number of elements.
\end{proof}

\begin{lemma}
$\Toddtwo$ is convex.
\end{lemma}

\begin{proof}
Suppose $\phi$ is a conjunction of literals, and that $\dash_{\Toddtwo}\phi\rightarrow\bigvee_{i=1}^{n}x_{i}=y_{i}$; if some pair $x_{i}$ and $y_{i}$ is of sort $\s$, since all models of $\Toddtwo$ have exactly one element of sort $\s$, it follows that $\dash_{\Toddtwo}\phi\rightarrow x_{i}=y_{i}$. So we can assume that all $x_{i}$ and $y_{i}$ are of sort $\s_{2}$; in addition, we may assume that $\phi$ is not a contradiction, given that in that case $\dash_{\Toddtwo}\phi\rightarrow x_{i}=y_{i}$ for any $1\leq i\leq n$.

Then consider the formula $\phi^{\prime}$, obtained from $\phi$ by removing the literals with variables of sort $\s$. These are necessarily of the form $x=y$, since: we have no functions or predicates; and $\s^{\A}$ has always cardinality $1$ for $\A$ a model of $\Toddtwo$, meaning it cannot satisfy $\neg(x=y)$ if $\phi$ is not a contradiction. Then, we have that $\dash_{\T_{\textit{odd}}}\phi^{\prime}\rightarrow\bigvee_{i=1}^{n}x_{i}=y_{i}$, where $\T_{\textit{odd}}$ is the theory over the one-sorted empty signature (let, in this case, the sort be $\s_{2}$ for clarity) whose models are all structures with an infinite or odd number of elements.

$\T_{\textit{odd}}$ is obviously stably-infinite, and from \Cref{SI empty theories are convex}, it follows it is convex, meaning $\dash_{\T_{\textit{odd}}}\phi^{\prime}\rightarrow x_{i}=y_{i}$ for some $1\leq i\leq n$. Of course, it follows that $\dash_{\Toddtwo}\phi\rightarrow x_{i}=y_{i}$.
\end{proof}


\subsection{$\Tonetwo$}

\begin{lemma}
$\Tonetwo$ is not stably-infinite, and thus not smooth, w.r.t. $\{\s, \s_{2}\}$.
\end{lemma}

\begin{proof}
$\Tonetwo$ has a model $\A$ where $|\s^{\A}|=1$ and $|\s_{2}^{\A}|=\omega$, but no models $\B$ where both $\s^{\B}$ and $\s_{2}^{\B}$ are infinite.
\end{proof}

\begin{lemma}
$\Tonetwo$ is not finitely witnessable, and thus not strongly finitely witnessable, w.r.t. $\{\s, \s_{2}\}$.
\end{lemma}

\begin{proof}
$\Tonetwo$ has a model $\A$ where $|\s^{\A}|=1$ and $|\s_{2}^{\A}|=\omega$, but no models $\B$ where both $\s^{\B}$ and $\s_{2}^{\B}$ are finite.
\end{proof}

\begin{lemma}
$\Tonetwo$ is convex.
\end{lemma}

\begin{proof}
Suppose that $\phi$ is a cube, and that $\dash_{\Tonetwo}\phi\rightarrow\bigvee_{i=1}^{n}x_{i}=y_{i}$. If some pair $(x_{i}, y_{i})$ is of variables of 
 sort $\s$, we already have $\dash_{\Tonetwo}\phi\rightarrow x_{i}=y_{i}$, so assume all $x_{i}$ and $y_{i}$ are of sort $\s_{2}$. 
We then have that $\dash_{\T}\phi\rightarrow\bigvee_{i=1}^{n}x_{i}=y_{i}$, where the antecedent and the consequent may be seen as formulas in an empty signature with 
only one sort $\s_{2}$, and $\T$ is the theory on this signature with only infinite models, which is stably-
infinite. \Cref{SI empty theories are convex} can then be used to obtain that 
$\dash_{\T}\phi\rightarrow x_{i}=y_{i}$ for some $1\leq i\leq n$, and it of course follows that $\dash_{\Tonetwo}\phi\rightarrow x_{i}=y_{i}$: indeed, suppose this is not true, and so there exists a $\Tonetwo$-interpretation $\A$ that satisfies $\phi$ but not $x_{i}=y_{i}$; taking the $\T$-interpretation $\B$ with $\s_{2}^{\B}=\s_{2}^{\A}$ and $x^{\B}=x^{\A}$ for every variable of sort $\s_{2}$, we see that $\B$ must satisfy $\phi$, but not $x_{i}=y_{i}$, contradicting the fact that $\dash_{\T}\phi\rightarrow x_{i}=y_{i}$ for some $1\leq i\leq n$.
\end{proof}


\subsection{$\Ttwotwo$}

\begin{lemma}
$\Tonetwo$ is not stably-infinite, and thus not smooth, w.r.t. $\{\s, \s_{2}\}$.
\end{lemma}

\begin{proof}
$\Ttwotwo$ has a model $\A$ where $|\s^{\A}|=2$ and $|\s_{2}^{\A}|=\omega$, but no models $\B$ where both $\s^{\B}$ and $\s_{2}^{\B}$ are infinite.
\end{proof}

\begin{lemma}
$\Ttwotwo$ is not finitely witnessable, and thus not strongly finitely witnessable, w.r.t. $\{\s, \s_{2}\}$.
\end{lemma}

\begin{proof}
$\Ttwotwo$ has a model $\A$ where $|\s^{\A}|=2$ and $|\s_{2}^{\A}|=\omega$, but no models $\B$ where both $\s^{\B}$ and $\s_{2}^{\B}$ are finite.
\end{proof}

\begin{lemma}
$\Ttwotwo$ is not convex.
\end{lemma}

\begin{proof}
If it were, \Cref{Barrett's theorem on convexity} would guarantee that $\Ttwotwo$ is also stably-infinite, since it has no models $\A$ where either $|\s^{\A}|$ or $|\s_{2}^{\A}|$ equals $1$.
\end{proof}


\subsection{$\TsM$}

\begin{lemma}\label{TsM is smooth}
$\TsM$ is smooth, and thus stably-infinite.
\end{lemma}

\begin{proof}
Take a quantifier-free formula $\phi$, a $\TsM$-interpretation $\A$ that satisfies $\phi$ and a cardinal $\kappa\geq |\s^{\A}|$. If $\kappa$ is infinite, we take two sets $A$ and $B$ with $\kappa$ elements, disjoint from each other and $\s^{\A}$, and define an interpretation $\B$ as follows: $\s^{\B}=\s^{\A}\cup A\cup B$; $s^{\B}$ equal to $s^{\A}$ when restricted to $\s^{\A}$, equal to the identity when restricted to $A$, and if $a\in B$, we only require that $s^{\B}(a)\neq a$, the specific value of the function at this element being irrelevant; and, for all variables $x$, $x^{\B}=x^{\A}$.

It is clear that $\B$ is a $\TsM$-interpretation since it has infinite elements satisfying $s(a)=a$ (namely, all those in $A$) and infinite elements that satisfying $s(a)\neq a$ (those in $B$); furthermore, $|\s^{\B}|=|\s^{\A}|+|A|+|B|=\kappa$, as $|A|=|B|=\kappa$ and $\kappa$
is infinite.
Finally, $\B$ validates $\phi$ since the atomic formulas in the signature $\Sigma_{s}$ are of the form $s^{i}(x)=s^{j}(y)$, for $i,j\in\mathbb{N}$, and these maintain their truth value in $\B$ if $x, y\in \vars(\phi)$ since $s^{\B}$ equals $s^{\A}$ when restricted to $\s^{\A}$.

Now, assume that $\kappa$, and so $|\s^{\A}|$ as well, is finite, and let $\kappa=m$ and $|\s^{\A}|=n$. We consider two sets $A$ and $B$ disjoint from $\s^{\A}$ with, respectively, $f_{1}(m)-f_{1}(n)$ and $f_{0}(m)-f_{0}(n)$ elements, and define an interpretation $\B$ as follows: 
$\s^{\B}=\s^{\A}\cup A\cup B$; $s^{\B}$ as equal to $s^{\A}$ when restricted to $\s^{\A}$, as equal to the identity when restricted to $A$, and as equal to any function from $B$ to $\s^{\A}$ when 
restricted to $B$; and $x^{\B}=x^{\A}$ for every variable $x$. It is easy to see that $\B$ is then a $\TsM$-interpretation: it has 
\[n+(f_{1}(m)-f_{1}(n))+(f_{0}(m)-f_{0}(n))=m\]
elements in its domain; $f_{1}(n)$ elements in $\s^{\A}$, and all $f_{1}(m)-f_{1}(n)$ elements in $A$, satisfy $s^{\B}(a)=a$, to a total of $f_{1}(m)$; and $f_{0}(n)$ elements in $\s^{\A}$, plus all $f_{0}(m)-f_{0}(n)$ elements in $B$ (remember $B$ and $\s^{\A}$ are disjoint), satisfy $s^{\B}(a)\neq a$, to a total of $f_{0}(m)$.

Furthermore, $|\s^{\B}|=|\s^{\A}|+|A|+|B|=m$; since, again, $\B$ and $\A$ agree on 
the interpretation of the variables in $\vars(\phi)$ and the value given by $s$ to the elements that may occur in $\phi$, and so $\phi$ is true in $\B$.
\end{proof}

\begin{lemma}\label{TsM is FW}
$\TsM$ is finitely witnessable.
\end{lemma}

\begin{proof}
For a quantifier-free formula $\phi$, consider the witness
\[\wit(\phi)=\phi\wedge\bigwedge_{i=1}^{n}\bigwedge_{j=0}^{M_{i}+1}y_{i,j}=s^{j}(z_{i})\wedge\bigwedge_{i=1}^{2^{k+1}}x_{i}=x_{i},\]
for: $\vars(\phi)=\{z_{1}, \ldots , z_{n}\}$; $M_{i}$ the maximum of the indexes $j$ such that the term $s^{j}(z_{i})$ appears in $\phi$; $2^{k}$ the smallest power of two equal to or larger than $2M$, where $M=\sum_{i=1}^{n}(M_{i}+2)$; and $x_{i}$ and $y_{i,j}$ fresh variables. One can easily convince themselves that
$\wit(\phi)$ is computable, and it is obvious that $\phi$ and $\exists\,\overarrow{x}.\:\wit(\phi)$ are $\TsM$-equivalent, for $\overarrow{x}=\vars(\wit(\phi))\setminus\vars(\phi)$, since: if $\phi$ is true in a $\TsM$-interpretation $\A$, if we change the assignment on $\A$ so that 
$y_{i,j}^{\A^{\prime}}=(s^{\A})^{j}(z_{i}^{\A})$, it is clear that $\wit(\phi)$ is satisfied by $\A^{\prime}$; and if $\exists\,\overarrow{x}.\:\wit(\phi)$ is satisfied by a $\TsM$-interpretation $\A$, given that the variables in $\overarrow{x}$ do not occur in $\phi$, 
\[\exists\,\overarrow{x}.\:\wit(\phi)\quad\text{and}\quad \phi\wedge\exists\,\overarrow{x}.\:\big[\bigwedge_{i=1}^{n}\bigwedge_{j=0}^{M_{i}+1}y_{i,j}=s^{j}(z_{i})\wedge\bigwedge_{i=1}^{2^{k+1}}x_{i}=x_{i}\big]\]
are equivalent, and thus $\A$ satisfies $\phi$.

So, assume that the $\TsM$-interpretation $\A$ satisfies $\wit(\phi)$. Let: $m_{1}<M$ be the number of terms $s^{j}(z_{i})$ (for $1\leq i\leq n$ and $0\leq j\leq M_{i}$) appearing in $\phi$ with $(s^{\A})^{j+1}(z_{i}^{\A})=(s^{\A})^{j}(z_{i}^{\A})$; $m_{0}<M$ be the number of such terms $s^{j}(z_{i})$ satisfying 
instead $(s^{\A})^{j+1}(z_{i}^{\A})\neq(s^{\A})^{j}(z_{i}^{\A})$; and $m_{1}^{*}\leq n$ be the number of elements $(s^{\A})^{M_{i}+1}(z_{i}^{\A})$ that are not in $\{s^{j}(z_{i}) : 1\leq i\leq n, 0\leq j\leq M_{i}\}^{\A}$. 
We then have that $m_{0}+m_{1}+m_{1}^{*}\leq M<2^{k}$, and so we can take sets $A$ and $B$ disjoint from $\s^{\A}$ and each other with, respectively, $2^{k}-m_{1}-m_{1}^{*}$ and $2^{k}-m_{0}$ elements. Finally, we define a $\TsM$-interpretation $\B$, starting by setting its domain to 
\[\s^{\B}=\{s^{j}(z_{i}) : 1\leq i\leq n, 0\leq j\leq M_{i}+1\}^{\A}\cup A\cup B.\]
To define $s^{\B}$, we take an element $a_{0}\in\s^{\A}$ (and there is at least one), and make
\[s^{\B}(a)=\begin{cases*}
s^{\A}(a) & if $a=(s^{\A})^{j}(z_{i}^{\A})$ for $1\leq i\leq n$ and $0\leq j\leq M_{i}$;\\
a & if $a=(s^{\A})^{M_{i}+1}(z_{i}^{\A})$, for $1\leq i\leq n$,\\
& and $a\neq (s^{\A})^{j}(z_{l}^{\A})$ for any $1\leq l\leq n$ and $0\leq j\leq M_{i}$;\\
a & if $a\in A$;\\
a_{0} & if $a\in B$;
\end{cases*}\]
notice that the first two cases are disjoint and cover all of $\{s^{j}(z_{i}) : 1\leq i\leq n, 0\leq j\leq M_{i}+1\}^{\A}$, while the two last are disjoint from each other and from the previous ones (since $A$ and $B$ are 
disjoint from each other and from $\s^{\A}$), and finish covering $\s^{\B}$; notice, as well, that $s^{\B}(a)$ is well defined since $\s^{\A}$ is not empty, and $B$ and $\s^{\A}$ are disjoint. And $x^{\B}=x^{\A}$ for every variable in $\phi$, 
$y_{i,j}^{\B}=(s^{\B})^{j}(z_{i}^{\B})$, and $x\mapsto x^{\B}$ a bijection between $\{x_{1}, ... , x_{2^{k+1}}\}$ and $\s^{\B}$, what is possible given both sets have $2^{k+1}$ elements (and 
arbitrarily for other variables). To see that this is true for $\s^{\B}$, where this fact is not obvious, notice the set $\alpha(\phi)^{\A}$, for
$\alpha(\phi)=\{s^{j}(z_{i}) : 1\leq i\leq n, 0\leq j\leq M_{i}\}$ has $m_{0}+m_{1}$ elements, from the definition of $m_{0}$ as the number of elements $a$ in $\alpha(\phi)^{\A}$ satisfying $s^{\A}(a)=a$, and of $m_{1}$ as the number of elements $a$ in $\alpha(\phi)^{\A}$ satisfying instead $s^{\A}(a)\neq a$ (and, of course, each element of $\alpha(\phi)^{\A}$ must satisfy either $s^{\A}(a)=a$ or $s^{\A}(a)\neq a$, and never both of them); therefore, $\{s^{j}(z_{i}) : 1\leq i\leq n, 0\leq j\leq M_{i}+1\}^{\A}$ has $m_{0}+m_{1}+m_{1}^{*}$ elements, from the definition of $m_{1}^{*}$ as being the number of elements of the form $(s^{\A})^{M_{i}}(z_{i}^{\A})$ which are not in $\alpha(\phi)^{\A}$. Using $A$ has 
$2^{k}-m_{1}-m_{1}^{*}$, and $B$ has $2^{k}-m_{0}$, we obtain the aforementioned total of $2^{k+1}$.

$\B$ is a $\TsM$-interpretation since: it has $2^{k+1}$ elements in its domain; half of these elements satisfy $s^{\B}(a)=a$, explicitly $m_{1}$ of those in $\{s^{j}(z_{i}) : 1\leq i\leq n, 0\leq j\leq M_{i}\}^{\A}$ that satisfy $s^{\A}(a)=a$, 
all $m_{1}^{*}$ elements $(s^{\A})^{M_{i}+1}(z_{i}^{\A})$ that are not in $\{s^{j}(z_{i}) : 1\leq i\leq n, 0\leq j\leq M_{i}\}^{\A}$, and all those 
$2^{k}-m_{1}-m_{1}^{*}$ elements in $A$; and half satisfying $s^{\A}(a)\neq a$, 
explicitly $m_{0}$ of those in $\{s^{j}(z_{i}) : 1\leq i\leq n, 0\leq j\leq M_{i}\}^{\A}$, and $2^{k}-m_{0}$ more elements in $B$. Since $f$ is defined so 
that $f_{1}(2^{k})=f_{0}(2^{k})$ and $f(m)=f_{0}(m)+f_{1}(m)$ for all $k, m\in\mathbb{N}\setminus\{0\}$,  it is then clear that $\B$ is indeed a $\TsM$-interpretation.

Furthermore, any term $\alpha$ (necessarily of the form $s^{j}(z_{i})$ with $1\leq i\leq n$ and $0\leq j\leq M_{i}$) that appears in $\phi$ receives the same value in either $\A$ or $\B$, what we now prove by induction. Indeed, for any variable $z_{i}$ that appears in $\phi$, $z_{i}^{\B}=z_{i}^{\A}$ (and thus $(s^{\B})^{0}(z_{i}^{\B})=(s^{\A})^{0}(z_{i}^{\A})$); and then, for each $0\leq j< M_{i}$, assuming as induction hypothesis that $(s^{\B})^{j}(z_{i}^{\B})=(s^{\A})^{j}(z_{i}^{\A})$ (call that element $a$, for convenience), by the definition of $s^{\B}$ we have $s^{\B}(a)=s^{\A}(a)$, and so   
\[(s^{\B})^{j+1}(z_{i}^{\B})=s^{\B}\big((s^{\B})^{j}(z_{i}^{\B})\big)=s^{\A}\big((s^{\A})^{j}(z_{i}^{\A})\big)=(s^{\A})^{j+1}(z_{i}^{\A}).\]
Since the underlying signature has no predicates, we have that all atomic formulas of $\phi$ receive the same truth-value in either $\A$ or $\B$; since $\phi$ has no quantifiers and is satisfied by $\A$, this means that $\phi$ is also satisfied by $\B$. We also have that $\B$ satisfies $\wit(\phi)$, since $y_{i,j}^{\B}=(s^{\B})^{j}(z_{i}^{\B})$. Finally, $\vars(\wit(\phi))^{\B}=\s^{\B}$, since $\{x_{1}, \ldots , x_{2^{k+1}}\}^{\B}\subseteq \vars(\wit(\phi))^{\B}$ and, given that $x\mapsto x^{\B}$ is a bijection between $\{x_{1}, \ldots , x_{2^{k+1}}\}$ and $\s^{\B}$, $\{x_{1}, \ldots , x_{2^{k+1}}\}^{\B}=\s^{\B}$, proving that $\wit$ is indeed a witness.
\end{proof}

\begin{lemma}\label{TsM is not SFW}
$\TsM$ is not strongly finitely witnessable.
\end{lemma}

\begin{proof}
Suppose $\wit$ is a strong witness; we start by proving that, given a quantifier-free formula $\phi$,
\[\mc(\phi)=\min\{|V/E| : \text{$E\in eq$ and $\wit(\phi)\wedge\delta_{V}^{E}$ is $\TsM$-satisfiable}\},\]
where $eq$ is the set of all equivalence relations $E$ on $V=\vars(\wit(\phi))$, being the corresponding arrangements denoted by $\delta_{V}^{E}$. To prove this identity, suppose $\A$ is a $\TsM$-interpretation that satisfies $\phi$ with minimal cardinality of the domain, and so $|\s^{\A}|=\mc(\phi)$; because $\phi$ and $\exists\,\overarrow{x}.\:\wit(\phi)$ are $\TsM$-equivalent (for $\overarrow{x}=\vars(\wit(\phi))\setminus\vars(\phi)$), $\A$ also satisfies $\exists\,\overarrow{x}.\:\wit(\phi)$, and by changing the value given by $\A$ to the variables in $\overarrow{x}$, we obtain a new $\TsM$-interpretation $\A^{\prime}$ that satisfies $\wit(\phi)$ and has the same underlying structure as $\A$.

Let $E$ be the equivalence on $V$ such that $xEy$ iff $x^{\A^{\prime}}=y^{\A^{\prime}}$, and we have that $\A^{\prime}$ satisfies $\wit(\phi)\wedge\delta_{V}^{E}$ and $|V/E|\leq |\s^{\A^{\prime}}|$; since $\wit$ is supposed to be a strong witness, there must then exist a $\TsM$-interpretation $\B$ that satisfies $\wit(\phi)\wedge\delta_{V}^{E}$ and $\s^{\B}=V^{\B}$, and so $|\s^{\B}|=|V/E|$. But, again since $\phi$ and $\exists\,\overarrow{x}.\:\wit(\phi)$ are $\TsM$-equivalent, $\B$ also satisfies $\exists\,\overarrow{x}.\:\wit(\phi)$ and therefore $\phi$, meaning that $|\s^{\B}|\geq |\s^{\A}|=|\s^{\A^{\prime}}|$. With all of that, $|\s^{\B}|\geq |\s^{\A^{\prime}}|\geq |V/E|=|\s^{\B}|$, implying all are equal and thus $|V/E|=|\s^{\A}|=\mc(\phi)$. Of course, we then get
\[\mc(\phi)\in \{|V/E| : \text{$E\in eq$ and $\wit(\phi)\wedge\delta_{V}^{E}$ is $\TsM$-satisfiable}\},\]
so now suppose that there exists an equivalence $E$ on $V$ such that $\wit(\phi)\wedge\delta_{V}^{E}$ is $\TsM$-satisfiable, but $|V/E|<\mc(\phi)$, and we shall reach a contradiction. Since $\wit(\phi)\wedge\delta_{V}^{E}$ is $\TsM$-satisfiable and $\wit$ is a strong witness, there exists a $\TsM$-interpretation $\A$ that satisfies $\wit(\phi)\wedge\delta_{V}^{E}$ and $\s^{\A}=V^{\A}$, and so $|\s^{\A}|=|V/E|$. But, since $\A$ also satisfies $\exists\,\overarrow{x}.\:\wit(\phi)$, and $\phi$ and $\exists\,\overarrow{x}.\:\wit(\phi)$ are $\TsM$-equivalent, $\A$ satisfies $\phi$, contradicting the fact that the smallest $\TsM$-interpretation to satisfy $\phi$ has domain of cardinality $\mc(\phi)>|V/E|=|\s^{\A}|$. So our identity for $\mc(\phi)$ is true. 

But the right side of the identity is indeed computable: in fact, finding the set $eq$ is trivial; testing whether $\wit(\phi)\wedge\delta_{V}^{E}$ is $\TsM$-satisfiable is also decidable, since it is equivalent to testing whether the same formula is not contradictory according to \Cref{Decidability of TsM}; and finding the number of equivalence classes of $E$ is also straightforward. 
Of course, this contradicts the fact that $\TsM$ does not have a computable $\mc$ function, as proven in \Cref{mcofTsMisnotcomputable}, proving that this theory is not strongly finitely witnessable.
\end{proof}

\begin{lemma}
$\TsM$ is convex.
\end{lemma}

\begin{proof}
Assume that $\phi$ is a cube, that $\dash_{\TsM}\phi\rightarrow\bigvee_{i=1}^{n}x_{i}=y_{i}$, but $\TsM$ is not convex, and so $\not\dash_{\TsM}\phi\rightarrow x_{i}=y_{i}$ for every $1\leq i\leq n$, meaning we can find $\TsM$-interpretations $\A_{i}$ that satisfy $\phi$ and $\neg(x_{i}=y_{i})$. We can conclude, then, that $\dash_{\T}\phi\rightarrow\bigvee_{i=1}^{n}x_{i}=y_{i}$, where $\T$ is the theory with all $\Sigma_{s}$-structures as models, axiomatized by the empty set: 
if this were not true, then there would exist a $\T$-interpretation $\A$ satisfying $\phi$ but not $\bigvee_{i=1}^{n}x_{i}=y_{i}$, meaning $\phi\wedge\neg\bigvee_{i=1}^{n}x_{i}=y_{i}$ is $\T$-satisfiable, and thus non-contradictory and $\TsM$-satisfiable, given \Cref{Decidability of TsM}, what is absurd given our assumptions. However, given that \Cref{uninterpretedfunctionsis convex} states $\T$ is convex, we must have $\dash_{\T}\phi\rightarrow x_{j}=y_{j}$, for some $1\leq j\leq n$, and so $\phi\wedge\neg(x_{j}=y_{j})$ is $\T$-unsatisfiable, itself contradicting the fact that $\A_{j}$ is a $\TsM$-interpretation (and so a $\T$-interpretation as well) that satisfies $\phi$ and $\neg(x_{j}=y_{j})$.
\end{proof}


\subsection{$\TSM$}

 \begin{lemma}
 The $\mc$ function of $\TSM$ is not computable.
 \end{lemma}
 
 \begin{proof}
 The proof is the same as the one of \Cref{mcofTsMisnotcomputable}.
 \end{proof}

 \begin{lemma}
 $\TSM$ is smooth, and thus stably-infinite.
 \end{lemma}

 \begin{proof}
We can slightly adapt the proof of \Cref{TsM is smooth}: the only difference in the proofs is that now one must require that $s^{\B}$ maps elements of $B$ into those of $A$, so that $s^{\B}(s^{\B}(a))=s^{\B}(a)$.
\end{proof}

\begin{lemma}
$\TSM$ is finitely witnessable.
\end{lemma}

\begin{proof}
Take a quantifier-free formula $\phi$ and consider the witness
\[\wit(\phi)=\phi\wedge\bigwedge_{i=1}^{n}y_{i}=s(w_{i})\wedge z_{i}=s(y_{i})\wedge\bigwedge_{i=1}^{2^{k+1}}x_{i}=x_{i},\]
where $\vars(\phi)=\{w_{1},\ldots, w_{n}\}$, $2^{k}$ is the smallest power of two greater than $2n$, and $x_{i}$, $y_{i}$ and $z_{i}$ are fresh variables. 
If 
a $\TSM$-interpretation $\A$ satisfies $\phi$, by making $y_{i}^{\A^{\prime}}=s^{\A}(w_{i}^{\A})$ and $z_{i}^{\A^{\prime}}=s^{\A}(y_{i}^{\A^{\prime}})$, we obtain a second $\TSM$-interpretation $\A^{\prime}$ that satisfies $\wit(\phi)$, since, for $\overarrow{x}=\vars(\wit(\phi))\setminus\vars(\phi)$, $\phi$ and $\exists\,\overarrow{x}.\:\wit(\phi)$ are $\TSM$-equivalent.

Now, suppose that the $\TSM$-interpretation $\A$ satisfies $\wit(\phi)$, and let 
\[V=\{w_{i} : 1\leq i\leq n\}\cup\{y_{i} : 1\leq i\leq n\}\cup\{z_{i} : 1\leq i\leq n\}.\]
Let $m_{1}$ be the cardinality of $\{a\in V^{\A} : s^{\A}(a)=a\}$, $m_{0}=|V^{\A}|-m_{1}$, and 
then $2^{k}>\max\{m_{0}, m_{1}\}$: 
this is because either $y_{i}^{\A}=w_{i}^{\A}$ or $z_{i}^{\A}=w_{i}^{\A}$ for each $1\leq i\leq n$, implying $|V^{\A}|\leq 2n$, and thus $m_{0}, m_{1}\leq 2n$. Given sets $A$ and $B$, disjoint from 
$\s^{\A}$, with respectively $2^{k}-m_{1}$ and $2^{k}-m_{0}$ elements, we define a new $\TSM$-interpretation $\B$ by making:
\begin{enumerate}
\item $\s^{\B}=V^{\A}\cup A\cup B$ (which has $2^{k+1}$ elements);
\item $s^{\B}$ equal to $s^{\A}$ when restricted to $V^{\A}$ (what is well-defined, since $s^{\A}(w_{i}^{\A})=y_{i}^{\A}$, $s^{\A}(y_{i}^{\A})=z_{i}^{\A}$ and $s^{\A}(z_{i}^{\A})$ equals either $w_{i}^{\A}$ or $y_{i}^{\A}$), equal to the identity when restricted to $A$, and equal to any function from $B$ to $A$ when restricted to $B$ (this way, all elements of $A$ satisfy $s^{\B}(a)=a$, while all of $B$ satisfy $s^{\B}(a)\neq a$ and $s^{\B}(s^{\B}(a))=s^{\B}(a)$);
\item and $x^{\B}=x^{\A}$ for any variable in $V$, $x_{i}\mapsto x_{i}^{\B}$ a bijection between $\{x_{1}, \ldots , x_{2^{k+1}}\}$ and $\s^{\B}$, and arbitrary otherwise.
\end{enumerate}

Now we prove that $\B$ is a $\TSM$-interpretation: it has $2^{k+1}$ elements; $2^{k}$ of them, specifically $m_{1}$ in $\{a\in V^{\A} : s^{\A}(a)=a\}$ and another $2^{k}-m_{1}$ in $A$, satisfy $s^{\B}(a)=a$; the other half, specifically 
$m_{0}$ in $V^{\A}\setminus\{a\in V^{\A} : s^{\A}(a)=a\}$ and $2^{k}-m_{1}$, satisfy $s^{\B}(a)\neq a$ instead; so $\B$ is at least a $\TsM$-interpretation. 
Furthermore: for any element $a$ of $V^{\A}$, $s^{\B}(a)=s^{\A}(a)$ and $s^{\B}(s^{\B}(a))=^{\A})s^{\A}(a))$, and we 
already have either $s^{\A}(s^{\A}(a))=a$ or $s^{\A}(s^{\A}(a))=s^{\A}(a)$; for any $a\in A$, $s^{\B}(a)=a$, and so $s^{\B}(s^{\B}(a))=a$; and for any $a\in B$, 
$s^{\B}(a)\in A$, meaning $s^{\B}(s^{\B}(a))=s^{\B}(a)$; so $\B$ also satisfies 
$\psiv$, and is therefore a $\TSM$-interpretation.

Now, let $\alpha$ be a term in $\phi$, necessarily of the form $s^{j}(w_{i})$: we know $w_{i}^{\A}$ and $s^{\A}(w_{i}^{\A})$ are in $V^{\A}$, and therefore so are $(s^{j}(w_{i}))^{\A}$ for all 
$j\in\mathbb{N}$, since $(s^{\A})^{j}(w_{i}^{\A})$ must equal either $w_{i}^{\A}$ or $s^{\A}(w_{i}^{\A})$; since $s^{\B}$ coincides with $s^{\A}$ on
$V^{\A}$, and $w_{i}^{\B}=w_{i}^{\A}$, we have that $\alpha^{\B}=\alpha^{\A}$. Given that $\wit(\phi)$ is quantifier and predicate-free, we have that it must 
receive the same value in either $\A$ or $\B$, and is therefore true in $\B$; furthermore, $x\mapsto x^{\B}$ is a 
bijection between $\{x_{1}, \ldots , x_{2^{k+1}}\}$ and $\s^{\B}$, and since $\{x_{1}, \ldots , x_{2^{k+1}}\}\subseteq \vars_{\s}(\wit(\phi))$, we have  $\vars_{\s}(\wit(\phi))^{\B}=\s^{\B}$. Hence $\wit$ is indeed a witness.
\end{proof}

\begin{theorem}\label{Decidability of TSM}
Given a quantifier-free formula $\phi$ in the signature $\Sigma_{s}$ with $\vars(\phi)=\{w_{1},\ldots, w_{n}\}$, if 
\[\overline{\phi}=\phi\wedge\bigwedge_{i=1}^{n}s^{2}(w_{i})=w_{i}\vee s^{2}(w_{i})=s(w_{i})\]
is satisfiable, then $\phi$ is $\TSM$-satisfiable.
\end{theorem}

\begin{proof}
Suppose $\overline{\phi}$ is satisfiable, and from \Cref{Decidability of TsM} there exists a $\TsM$-interpretation $\A$ that satisfies $\overline{\phi}$. We produce a $\TSM$-interpretation $\B$ as follows: $\s^{\B}=\s^{\A}$; if 
\[a\in\s^{\A}\setminus\vars(\phi)^{\A}\quad\text{and}\quad a\in \s^{\A}\setminus\{ s^{\A}(a) : a\in\vars(\phi)^{\A}\},\]
and $s^{\A}(a)\neq a$, we make $s^{\B}(a)$ equal to any element $b\in\s^{\A}$ such that $s^{\A}(b)=b$, and otherwise $s^{\B}$ equals $s^{\A}$; for all variables $x$, $x^{\B}=x^{\A}$. Since the interpretation of all variables in $\B$ is the same as in 
$\A$, for every variable $w_{i}$ of $\overline{\phi}$ one has $w_{i}^{\B}=w_{i}^{\A}$; and since $s^{\B}$ agrees with $s^{\A}$ in $\vars(\phi)^{\A}\cup\{ s^{\A}(a) : a\in\vars(\phi)^{\A}\}$, and for every 
$j\in\mathbb{N}$ one finds that $(s^{\A})^{j}(w_{i}^{\A})$ because $\A$ satisfies $\bigwedge_{i=1}^{n}s^{2}(w_{i})=w_{i}\vee s^{2}(w_{i})=s(w_{i})$, we get that for any term $\alpha$ in $\overline{\phi}$, $\alpha^{\B}=\alpha^{\A}$. Since $\Sigma_{s}$ does not have predicates (other than identity) and $\overline{\phi}$ is quantifier-free, we reach the conclusion that $\B$ clearly satisfies $\overline{\phi}$, and so $\phi$.

It is less clear, however, that $\B$ is indeed a $\TSM$-interpretation. We start by noticing $\B$ is at least a $\TsM$-
interpretation, since $s^{\B}(a)=a$ iff $s^{\A}(a)=a$: indeed, begin by assuming that $s^{\B}(a)=a$; we cannot have that $a\in\s^{\A}\setminus\vars(\phi)^{\A}$, 
$a\in \s^{\A}\setminus\{ s^{\A}(a) : a\in\vars(\phi)^{\A}\}$ and $s^{\A}(a)\neq a$, since in that case $s^{\B}(a)=b$ for an element $b$ such that $s^{\A}(b)=b$ 
(which must forcibly be different from $a$, since $s^{\A}(a)\neq a$); so $s^{\B}$ must coincide with $s^{\A}$, implying that $s^{\A}(a)=a$. Reciprocally, assume $s^{\A}(a)=a$: then we are not in the case that $a\in\s^{\A}\setminus\vars(\phi)^{\A}$, 
$a\in \s^{\A}\setminus\{ s^{\A}(a) : a\in\vars(\phi)^{\A}\}$, and so $s^{\B}$ must coincide with $s{\A}$, meaning $s^{\B}(a)=a$, as we wished to show.

Furthermore:
\begin{enumerate}
\item if $a\in \vars(\phi)^{\A}$ or $a\in\{s^{\A}(a) : a\in\vars(\phi)^{\A}\}$, by the fact that $\B$ satisfies $\overline{\phi}$ we get that either $s^{\B}(s^{\B}(a))=a$ or $s^{\B}(s^{\B}(a))=s^{\B}(a)$; 
\item if $a\in \s^{\A}\setminus\vars(\phi)^{\A}$, $a\in\s^{\A}\setminus\{s^{\A}(a) : a\in\vars(\phi)^{\A}\}$, and $s^{\A}(a)=a$, $s^{\B}(a)=a$;
\item and if $a\in \s^{\A}\setminus\vars(\phi)^{\A}$ and $a\in\s^{\A}\setminus\{s^{\A}(a) : a\in\vars(\phi)^{\A}\}$, but $s^{\A}(a)\neq a$, there is a $b\in \s^{\A}$ with $s^{\A}(b)=b$ such that $s^{\B}(a)=b$, and therefore $s^{\B}(s^{\B}(a))=s^{\B}(a)$,
\end{enumerate}
proving $\B$ is indeed a $\TSM$-interpretation.
\end{proof}

\begin{lemma}
$\TSM$ is not strongly finitely witnessable.
\end{lemma}

\begin{proof}
We proceed as in \Cref{TsM is not SFW}, showing that
\[\mc(\phi)=\min\{|V/E| : \text{$E\in eq$ and $\wit(\phi)\wedge\delta_{V}^{E}$ is $\TSM$-satisfiable}\},\]
what is absurd given the right side is computable: indeed, given $\phi$, finding $V$ and $eq$ can be easily done, as well as $|V/E|$; obtaining 
$\wit(\phi)\wedge\delta_{V}^{E}$ can also be done algorithmically, if $\wit$ is computable; and, thanks to \Cref{Decidability of TSM}, testing whether $\wit(\phi)\wedge\delta_{V}^{E}$ is $\TSM$-satisfiable is equivalent to testing whether the formula 
\[\overline{\wit(\phi)\wedge\delta_{V}^{E}}=\wit(\phi)\wedge\delta_{V}^{E}\wedge\bigwedge_{w\in V}s^{2}(x)=x\vee s^{2}(x)=s(x)\]
is satisfiable, something that can be achieved algorithmically. \end{proof}

\begin{lemma}
$\TSM$ is not convex.
\end{lemma}

\begin{proof}
This proof is the same as the one for the non-convexity of $\Taddnc$ in \Cref{Tvee is ...}.
\end{proof}


\subsection{$\Tneqodd$}\label{CV and FW, but not SI, SM nor SFW}

\begin{lemma}
$\Tneqodd$ is not stably-infinite, and thus not smooth.
\end{lemma}

\begin{proof}
While $s(x)=x$ is satisfied by the trivial $\Tneqodd$-interpretation, any infinite model of $\Tneqodd$ must satisfy $\forall\, x.\:\neg[s(x)=x]$ instead.
\end{proof}

\begin{lemma}
$\Tneqodd$ is finitely witnessable.
\end{lemma}

\begin{proof}
Given a quantifier-free formula $\phi$ with $\vars(\phi)=\{w_{1}, \ldots, w_{n}\}$ (notice $n\geq 1$, since the signature over which 
$\Tneqodd$ is defined has neither function nor predicate constants, being thus impossible to define a 
formula with no variables) and, for every $1\leq i\leq n$, $M_{i}=\max\{j : \text{$s^{j}(w_{i})$ is in $\phi$}\}$, we define its witness as
\[\wit(\phi)=\phi\wedge(y=y)\bigwedge_{i=1}^{n}\bigwedge_{j=0}^{M_{i}+1}y_{i,j}=s^{j}(w_{i}),\]
where $y$ and $y_{i,j}$ are all fresh variables. For $\overarrow{x}=\vars(\wit(\phi))\setminus\vars(\phi)$, 
we prove that $\phi$ and $\exists\, \overarrow{x}.\:\wit(\phi)$ are $\Tneqodd$-equivalent: 
the right-to-left is trivial. For the converse,
if 
$\A$ is a $\Tneqodd$-interpretation that satisfies $\phi$, by changing the value given by $\A$ to the 
variables $y_{i,j}$ so that $y_{i,j}^{\A^{\prime}}=(s^{\A})^{j}(w_{i}^{\A})$, we obtain a second $\Tneqodd$-interpretation $\A^{\prime}$ that satisfies 
$\wit(\phi)$; of course, this means that $\A$ itself satisfies $\exists\,\overarrow{x}.\:\wit(\phi)$.

Now, assume that the $\Tneqodd$-interpretation $\A$ satisfies $\wit(\phi)$. Let $V=\vars(\phi)\cup\{y_{i,j} : 1\leq i\leq n, 0\leq j\leq M_{i}+1\}$. 
We have three cases to consider.
\begin{enumerate}
\item If $|V^{\A}|=1$, it is clear that $\A$ is the $\Tneqodd$-interpretation with only one element (since $M_{i}+1\geq 1$, and so $(s^{\A})^{1}(w_{i}^{\A})=(s^{\A})^{0}(w_{i}^{\A})$), and therefore $\A$ is already a $\Tneqodd$-interpretation that satisfies $\wit(\phi)$ with $\vars_{\s}(\wit(\phi))^{\A}=\s^{\A}$, so there is nothing we need to do.

\item If $|V^{\A}|$ is an odd number greater than $1$, we make a second $\Tneqodd$-interpretation $\B$ by proceeding as follows. Regarding the domain of $\B$, $\s^{\B}=V^{\A}$. If $a=(s^{\A})^{j}(w_{i}^{\A})$ for some
$1\leq i\leq n$ and $0\leq j\leq M_{i}$, $s^{\B}(a)=s^{\A}(a)$ (and, this way, $s^{\B}(a)\neq a$, since $s^{\A}(a)\neq a$); and if $a=(s^{\A})^{M_{i}+1}(w_{i}^{\A})$, but $a$ does not equal $(s^{\A})^{j}(w_{k}^{\A})$ for any $1\leq k\leq n$ and $0\leq j\leq M_{k}$, we simply make $s^{\B}(a)$ equal any element from 
$\s^{\B}$ different from $a$ (and there is one, since $|V^{\A}|>1$). Finally, $x^{\B}=x^{\A}$ for all variables $x\in V$, and arbitrarily otherwise (what includes $y$ and all $y_{i,j}$). It is then easy to see 
that not only $\B$ is a $\Tneqodd$-interpretation that satisfies $\wit(\phi)$, but also $\s^{\B}=\vars(\wit(\phi))^{\B}$.

\item Finally, suppose that $|V^{\A}|$ is an even number. We then take an element $b\notin V^{\A}$ and 
define a new $\Tneqodd$-interpretation $\B$ as follows. For the domain, we use $\s^{\B}=V^{\A}\cup\{b\}$. If $a=(s^{\A})^{j}(w_{i}^{\A})$, for some $1\leq i\leq n$
and $0\leq j\leq M_{i}$, again we make $s^{\B}(a)=s^{\A}(a)$; and if $a=b$ or $a=(s^{\A})^{M_{i}+1}(w_{i}^{\A})$ but $a$ does not equal $(s^{\A})^{j}(w_{k}^{\A})$, for any $1\leq k\leq n$ and $0\leq j\leq M_{k}$, $s^{\B}(a)$ may be an arbitrary element 
from $\s^{\B}\setminus\{a\}$. And $x^{\B}=x^{\A}$ for all variables $x$ in $V$, $y^{\B}=b$, and arbitrarily 
otherwise. Again one easily obtains that $\B$ is a $\Tneqodd$-interpretation that satisfies both $\wit(\phi)$ and $\s^{\B}=\vars(\wit(\phi))^{\B}$.\qedhere

\end{enumerate}
\end{proof}

\begin{lemma}
$\Tneqodd$ is not strongly finitely witnessable. 
\end{lemma}

\begin{proof}
Suppose that $\wit$ is a strong witness. We begin by noticing that there are infinite $\Tneqodd$-interpretations $\A^{\prime}$, such as the one with domain $\mathbb{N}$ and $s^{\A^{\prime}}(n)=n+1$ for all $n\in\mathbb{N}$. Since $\phi=(w=w)$ is satisfied by all $\Tneqodd$-interpretations, including the infinite ones, there must exist an infinite $\Tneqodd$-interpretation $\A$ that satisfies $\wit(w=w)$ (since 
$\phi$ and $\overarrow{x}\wit(\phi)$ are $\Tneqodd$-equivalent, for $\overarrow{x}=\vars(\wit(\phi))\setminus\vars(\phi)$, and $\A^{\prime}$ satisfies $\phi$, there must exist a $\Tneqodd$-interpretation $\A$, differing from $\A^{\prime}$ at most on $\overarrow{x}$, that satisfies $\wit(\phi)$).

Consider now the set $W=\vars(\wit(w=w))$ and the equivalence relation $F$ on $W$ such that $xFy$ iff $x^{\A}=y^{\A}$, with corresponding arrangement 
$\delta_{W}$: of course $\A$ satisfies $\wit(\phi)\wedge\delta_{W}$, and we have now two cases to consider.

\begin{enumerate}
\item If $W/F$ has an even number of equivalence classes, we know there must exist a $\Tneqodd$-interpretation $\B$ that satisfies 
$\wit(w=w)\wedge\delta_{W}$ with $\s^{\B}=W^{\B}$, what 
is absurd: if $\B$ satisfies $\delta_{W}$, $W^{\B}$ will have as many elements as $W/F$, and therefore have an 
even number of them, contradicting the fact that $\B$ is a $\Tneqodd$-interpretation.

\item So, assume that $W/F$ has an odd number of equivalence classes, take some $z\notin W$, and define the equivalence relation $E$ on $V=W\cup\{z\}$ such that 
$xEy$ iff $xFy$ or $x=y$, with corresponding arrangement $\delta_{V}$. 

To see that $\wit(w=w)\wedge\delta_{V}$ is still $\Tneqodd$-satisfiable, remember that $\A$ not only is a 
$\Tneqodd$-interpretation that satisfies $w=w$, but is also infinite: since $W$, and thus $W^{\A}$, must be finite, there exists an element $a\in\s^{\A}\setminus W^{\A}$; we then define 
$\A^{\prime\prime}$ to be the same interpretation as $\A$, except that $z^{\A^{\prime\prime}}=a$ (and, for all other variables $x$, $x^{\A^{\prime\prime}}=x^{\A}$). Given $\A$ and $\A^{\prime\prime}$ agree on the variables of the quantifier-free 
formula $\wit(\phi)\wedge\delta_{W}$, and $\A$ satisfies $\wit(\phi)\wedge\delta_{W}$, it follows that $\A^{\prime\prime}$ also satisfies that formula and, additionally, that 
$z^{\A^{\prime\prime}}\neq x^{\A^{\prime\prime}}$ for all $x\in W$; this, of course, means $\A^{\prime\prime}$ is a $\Tneqodd$-interpretation that satisfies $\wit(\phi)\wedge\delta_{V}$.

So there must exist a $\Tneqodd$-interpretation $\B$ that satisfies 
$\wit(w=w)\wedge\delta_{V}$ with $\s^{\B}=V^{\B}$. Of course, this is absurd: if $W/F$ has an odd number of equivalence classes, $V/E$ has an even number of 
equivalence classes, forcing $\B$ to have an even number of elements in its domain since it validates $\delta_{V}$.\qedhere
\end{enumerate}
\end{proof}

In the following proof, we need to use the fact that the theory 
$\T^{\prime}$, axiomatized by the set of formulas
\[\{\neg\psi_{=1}\}\cup\{\forall\, x.\:\neg(s(x)=x)\}\cup\{\neg\psi_{2k} : k\in\mathbb{N}\},\]
is stably-infinite. This is actually easy to see: take a quantifier-free formula $\phi$ and a $\T^{\prime}$-interpretation $\A$ that satisfies $\phi$. Consider then a set $A=\{a_{n}, b_{n} : n\in\mathbb{N}\}$ disjoint from $\s^{\A}$, and define a $\T^{\prime}$-interpretation $\B$ with: $\s^{\B}=\s^{\A}\cup A$; $s^{\B}(a)=s^{\A}(a)$ if $a\in\s^{\A}$, $s^{\B}(a_{n})=b_{n}$ and $s^{\B}(b_{n})=a_{n}$; and $x^{\B}=x^{\A}$ for all variables $x$. Then $\B$ is infinite, meaning it satisfies $\{\neg\psi_{=1}\}\cup\{\neg\psi_{2k} : k\in\mathbb{N}\}$, and in addition satisfies $\forall\, x.\:\neg(s(x)=x)$; furthermore, it satisfies $\phi$, implying it is an infinite $\T^{\prime}$-interpretation that satisfies this formula, what makes of the theory stably-infinite.

\begin{lemma}\label{Tneqodd is convex}
$\Tneqodd$ is convex.
\end{lemma}

\begin{proof}
Suppose we have 
\[\dash_{\Tneqodd}\phi\rightarrow\bigvee_{k=1}^{n}x_{k}=y_{k}\quad\text{but}\quad \not\dash_{\Tneqodd}\phi\rightarrow x_{k}=y_{k},\quad\text{for every $1\leq k\leq n$},\]
where $\phi$ is a conjunction of literals. There must then exist $\Tneqodd$-interpretations $\A_{k}$ that 
satisfy $\phi$ but not $x_{k}=y_{k}$, for every $1\leq k\leq n$; notice that, since $\A_{k}$ does not satisfy $x_{k}=y_{k}$, it cannot be the $\Tneqodd$-interpretation with only one element. Notice as well
that, if we remove the model with domain of cardinality $1$ from the class of models of $\Tneqodd$, we obtain the theory $\T^{\prime}$ axiomatized by the set of 
formulas
\[\{\neg\psi_{=1}\}\cup\{\forall\, x.\:\neg(s(x)=x)\}\cup\{\neg\psi_{2k} : k\in\mathbb{N}\},\]
which, unlike $\Tneqodd$, is stably-infinite. Because, for each $1\leq k\leq n$, the $\A_{k}$ are $\Tneqodd$-interpretations, and so $\T^{\prime}$-interpretations, that satisfy $\phi$ but not $x_{k}=y_{k}$, using the fact that $\T^{\prime}$ is stably-infinite (and $\phi\wedge\neg(x_{k}=y_{k})$ is quantifier-free) we obtain $\T^{\prime}$-interpretations, that are necessarily $\Tneqodd$-interpretations as well, $\A^{\prime}_{k}$ that satisfy $\phi$ but not $x_{k}=y_{k}$, and that are infinite. 

Appealing to \Cref{LowenheimSkolemDownwards} once again, there must 
exist $\Tneqodd$-interpretations $\B_{k}$ with countably infinite domains that satisfy $\phi$ but not 
$x_{k}=y_{k}$. Let $\{z_{1}, \ldots, z_{M}\}=\vars(\phi)\cup\{x_{k}, y_{k} : 1\leq k\leq n\}$ and define $M_{i}$, for each $1\leq i\leq M$, as either the maximum of $j$ such that $s^{j}(z_{i})$ appears in $\phi$ or, if no $s^{j}(z_{i})$ is a subterm of $\phi$, as equal to $0$. We then take a fresh set of 
variables $V=\{x_{i,j} : 1\leq i\leq M, 0\leq j\leq M_{i}\}$ and for each $k$ define $E_{k}$ as the smallest equivalence relation on $V$ such that:
\begin{enumerate}
\item if $x_{i,j}E_{k}x_{p,q}$, $0\leq j<M_{i}$ and $0\leq q<M_{p}$, then $x_{i,j+1}E_{k}x_{p, q+1}$;
\item if $(s^{\B_{k}})^{j}(z_{i}^{\B_{k}})=(s^{\B_{k}})^{q}(z_{p}^{\B_{k}})$, then $x_{i,j}E_{k}x_{p,q}$.
\end{enumerate}
Notice that none of these equivalence relations is ill-defined, since they are precisely the equivalences induced by $\B_{k}$ on the set $V$ once we identify 
$(s^{\B_{k}})^{j}(z_{i}^{\B_{k}})$ with $x_{i,j}$ (observe that the first defining property of $E_{k}$ comes from the fact that, if $a=b$, then $s^{\B_{k}}(a)=s^{\B_{k}}(b)$): since $s^{\B_{k}}(a)\neq a$ for all $a\in \s^{\B_{k}}$, we also easily derive that, if $0\leq j<M_{i}$, $x_{i,j}\overline{E_{k}}x_{i,j+1}$, where $\overline{E_{k}}$ is the complement of $E_{k}$. We then finally define the equivalence $E$ on $V$ by setting $x_{i,j}Ex_{p,q}$ iff 
$x_{i,j}E_{k}x_{p,q}$ for all $1\leq k\leq n$.

We will denote by $[x_{i,j}]$ the equivalence class with representative $x_{i,j}$, and proceed now to define a $\Tneqodd$-interpretation $\B$ as follows.
\begin{enumerate}
\item $\s^{\B}=(V/E)\cup \mathbb{N}$ (which is infinite).

\item If there is an $x_{i,j}$ in $[x_{p,q}]$ such that $j<M_{i}$, we make $s^{\B}([x_{p,q}])=[x_{i,j+1}]$, otherwise $s^{\B}([x_{p,q}])=0$; $s^{\B}(a)=a+1$ for 
every $a\in\mathbb{N}$ (notice that we always have $s^{\B}(a)\neq a$, being obviously true if $a\in\mathbb{N}$ or if $s^{\B}(a)=0$; if $a\in V/E$, $a=[x_{i,j}]$ and $s^{\B}([x_{i,j}])=[x_{i,j+1}]$, the fact 
that $s^{\B}(a)=a$, and therefore $[x_{i,j}]=[x_{i,j+1}]$, would imply that $x_{i,j}E_{k}x_{i,j+1}$, what is not possible).

So, to prove that $s^{\B}$ is well-defined, suppose that $[x_{i,j}]=[x_{p,q}]$, $j<M_{i}$ and $q<M_{p}$: because of the first defining property of $E_{k}$, we have that
$x_{i,j}E_{k}x_{p,q}$ implies $x_{i,j+1}E_{k}x_{p,q+1}$, and of course this will extend to $E$, meaning $s^{\B}([x_{i,j}])=s^{\B}([x_{p,q}])$ as we needed to show.

\item For every variable $z_{i}$ in $V$, $z_{i}^{\B}=[x_{i,0}]$ (and we may define $x^{\B}$ arbitrarily for variables $x$ not in $V$).
\end{enumerate}

Now, we state that $\B$ validates $\phi$: in fact, let $s^{j}(z_{i})=s^{q}(z_{p})$ be a literal of $\phi$ not preceded by negation; in this case, $x_{i,j}E_{k}x_{p,q}$ for every $1\leq k\leq n$ (by the third defining
property of $E_{k}$), and so $[x_{i,j}]=[x_{p,q}]$, meaning $(s^{\B})^{j}(z_{i}^{\B})=(s^{\B})^{q}(z_{i}^{\B})$. If, however, $\neg[s^{j}(z_{i})=s^{q}(z_{p})]$ is the literal in $\phi$, since $E_{k}$ is the 
smallest equivalence with its properties, $x_{i,j}\overline{E_{k}}x_{p,q}$, and so $(s^{\B})^{j}(z_{i}^{\B})\neq(s^{\B})^{q}(z_{i}^{\B})$.

Finally, we see that $\B$ does not satisfy any $x_{k}=y_{k}$, leading to a contradiction: if $x_{k}=z_{i}$ and $y_{k}=z_{p}$,since 
$x_{i,0}\overline{E_{k}}x_{p,0}$, by construction of $E_{k}$, this means that $x_{i,0}\overline{E}x_{p,0}$, and so $\B$ does not satisfy $x_{k}=y_{k}$, for each 
$1\leq k\leq n$; of course, this would imply that $\B$ does not satisfy $\bigvee_{k=1}^{n}x_{k}=y_{k}$, contradicting the fact that is satisfies $\phi$ and we have 
$\dash_{\Tneqodd}\phi\rightarrow\bigvee_{k=1}^{n}x_{k}=y_{k}$.
\end{proof}


 \subsection{$\Tneqoneinfty$}

\begin{lemma}
$\Tneqoneinfty$ is not stably-infinite, and thus not smooth.
\end{lemma}

\begin{proof}
Notice that the quantifier-free formula $s(x)=x$ is satisfied by the trivial $\Tneqoneinfty$-interpretation, but by no infinite interpretations of this theory.
\end{proof}

\begin{lemma}
$\Tneqoneinfty$ is not finitely witnessable, and thus not strongly finitely witnessable.
\end{lemma}

\begin{proof}
If there existed a witness $\wit$, there would also exist a finite $\Tneqoneinfty$-interpretation $\A$ satisfying both $\wit(\phi)$ and $\vars(\wit(\phi))^{\A}=\s^{\A}$, for $\phi=\neg(x=y)$. This is absurd, since $\A$ would be a finite $\Tneqoneinfty$-interpretation that satisfies $\phi$, and therefore has at least $2$ elements.
\end{proof}

\begin{lemma}
$\Tneqoneinfty$ is convex.
\end{lemma}

\begin{proof}
This proof follows that of \Cref{Tneqodd is convex}.
\end{proof}


\subsection{$\Tneqtwoinfty$}

\begin{lemma}\label{Tneqoneinfty is not SI}
$\Tneqtwoinfty$ is not stably-infinite, and thus not smooth.
\end{lemma}

\begin{proof}
Since the quantifier-free formula $s(x)=x$ is satisfied by the $\Tneqtwoinfty$-interpretation with two elements, but by no infinite such interpretation, the theory cannot be stably-infinite.
\end{proof}

\begin{lemma}
$\Tneqtwoinfty$ is not finitely witnessable, and thus not strongly finitely witnessable.
\end{lemma}

\begin{proof}
If $\Tneqtwoinfty$ were finitely witnessable with witness $\wit$, we would have that there is a finite $\Tneqtwoinfty$-interpretation $\A$ satisfying $\wit(\phi)$ and $\vars(\wit(\phi))^{\A}=\s^{\A}$, for
\[\phi=\neg(x=y)\wedge\neg(x=z)\wedge\neg(y=z).\]
This is absurd, since $\A$ would be a finite $\Tneqtwoinfty$-interpretation that satisfies $\phi$, and therefore has at least $3$ elements.
\end{proof}

\begin{lemma}
$\Tneqtwoinfty$ is not convex.
\end{lemma}

\begin{proof}
Given \Cref{Barrett's theorem on convexity}, and the facts that $\Tneqtwoinfty$ has no trivial models and is not stably-infinite (see \Cref{Tneqoneinfty is not SI}), the theory cannot be convex.
\end{proof}


\section{Proof of \Cref{newcombinationtheorem}}

For the proof of \Cref{newcombinationtheorem}, we rely on the
following theorem, which may be found, with a proof, in \cite{JB10-TR} as Theorem $2.5$.

\begin{theorem}\label{LemmafromPolitetheoriesrevisited}
Let $\Sigma_{1}$ and $\Sigma_{2}$ be disjoint signatures, $\T_{1}$ be a $\Sigma_{1}$-theory and $\T_{2}$ be a $\Sigma_{2}$-theory. Consider $\Sigma=\Sigma_{1}\cup\Sigma_{2}$, $\T=\T_{1}\oplus \T_{2}$ and $S$ the set of sorts shared by $\Sigma_{1}$ and $\Sigma_{2}$. Let $\phi_{1}$ and $\phi_{2}$ be quantifier-free, respectively, $\Sigma_{1}$ and $\Sigma_{2}$-formulas, and let $U_{\s}=\vars_{\s}(\phi_{1})\cap\vars_{\s}(\phi_{2})$. 

If there exists a $\T_{1}$-interpretation $\A$ and a $\T_{2}$-interpretation $\B$, and an arrangement $\delta_{U}$ on $U$ such that $\A$ satisfies $\phi_{1}\wedge\delta_{U}$, $\B$ satisfies $\phi_{2}\wedge\delta_{U}$ and, for all sorts $\s\in S$, $|\s^{\A}|=|\s^{\B}|$, then there exists a $\T$-interpretation $\C$ such that $\C$ satisfies $\phi_{1}\wedge\phi_{2}\wedge\delta_{U}$, $\s^{\C}=\s^{\A}$ for all $\s\in\S_{\Sigma_{1}}$, and $\s^{\C}=\s^{\B}$ for all $\s\in\S_{\Sigma_{2}}\setminus S$.
\end{theorem}

The key ingredient of our proof is the following lemma, which relaxes the need for smoothness in polite theory combination. 
\begin{restatable}{lemma}{sisplussfweqcsms}
\label{SI+SFW=CS}
Let $\Sigma$ be a signature with $S\subseteq\S_{\Sigma}$, and $\T$ a theory over $\Sigma$. If $\T$ is a stably-infinite and strongly finitely witnessable theory, both w.r.t. the set of sorts $S$, then: for every quantifier-free $\Sigma$-formula $\phi$; $\T$-interpretation $\A$ that satisfies $\phi$; and function $\kappa$ from $S_{\omega}^{\A}=\{\s\in S : |\s^{\A}|\leq \omega\}$ to the class of cardinals such that $|\sigma^{\A}|\leq \kappa(\sigma)\leq \omega$ for every $\s\in S_{\omega}^{\A}$, there exists a $\T$-interpretation $\B$ that satisfies $\phi$ with $|\s^{\B}|=\kappa(\s)$ for every $\s\in S_{\omega}^{\A}$, and $|\s^{\B}|=\omega$ for every $\s\in S\setminus S_{\omega}^{\A}$.
\end{restatable}

\begin{proof}
Suppose that $\T$ is stably-infinite and strongly finitely witnessable w.r.t. a set of sorts $S$; let $\phi$ be a quantifier-free formula, $\A$ a $\T$-interpretation that satisfies $\phi$, and take a function $\kappa$ from $S_{\omega}^{\A}$ to the class of cardinals such that $|\sigma^{\A}|\leq \kappa(\sigma)\leq \omega$ for every $\s\in S_{\omega}^{\A}$. For simplicity, we also define $S^{-}_{\omega}$ to be the set of sorts $\sigma\in S_{\omega}^{\A}$ such that $\kappa(\sigma)<\omega$, while $S^{+}_{\omega}=S\setminus S^{-}_{\omega}$ will be its complement in $S$, that is, those sorts in $S_{\omega}^{\A}$ with $\kappa(\sigma)=\omega$, and those sorts in $S$ with $\s^{\A}>\omega$.

Suppose $\wit$ is our strong witness: since $\A$ satisfies $\phi$, it must also satisfy $\exists\, \overarrow{x}.\: \wit(\phi)$, for $\overarrow{x}=\vars(\wit(\phi))\setminus\vars(\phi)$; by changing $\A$ at most on these variables, we obtain a second $\T$-interpretation $\A^{\prime}$ that satisfies $\wit(\phi)$.
For each $\sigma\in S$, let
 $W_{\sigma}=\vars_{\sigma}(\wit(\phi))$, equipped with the equivalence relations $F_{\sigma}$ such that $xF_{\sigma}y$ iff $x^{\A^{\prime}}=y^{\A^{\prime}}$; the corresponding arrangement will be 
 \[\delta_{W}=\bigwedge_{\s\in S}\big[\bigwedge_{xF_{\s}y}x=y\wedge\bigwedge_{x\overline{F_{\s}}y}\neg(x=y)\big],\]
 where $\overline{F_{\s}}$ is the complement of $F_{\s}$. 

Now, take a positive integer $M$ and consider, for every $\sigma\in S$, a set of fresh variables $U_{\sigma}$ of sort $\sigma$ with:
\[|U_{\sigma}|=\begin{cases*}
\kappa(\sigma)-|W_{\sigma}/F_{\sigma}| & if $\sigma\in S^{-}_{\omega}$;\\
M & if $\sigma\in S^{+}_{\omega}$
\end{cases*}\]
(notice that $|W_{\s}/F_{\s}|\leq |\s^{\A^{\prime}}|$, by definition of $F_{\s}$, $|\s^{\A^{\prime}}|=|\s^{\A}|$ and $|\s^{\A}|\leq \kappa(\s)$, all for each $\s\in S$, meaning $\kappa(\sigma)-|W_{\sigma}/F_{\sigma}|$ is always non-negative). We also define the relation $E_{\sigma}$ on $V_{\sigma}=U_{\sigma}\cup W_{\sigma}$, such that $xE_{\sigma}y$ iff $xF_{\sigma}y$ or if $x=y$, and we will denote the corresponding arrangement by $\delta_{V}$.

Now, because $\T$ is stably-infinite w.r.t. $S$, and $\wit(\phi)\wedge\delta_{W}$ is quantifier-free (and satisfied by the $\T$-interpretation 
$\A^{\prime}$), there must exist a $\T$-interpretation $\B$, with $\sigma^{\B}$ infinite for every sort $\sigma\in S$, that satisfies $\wit(\phi)\wedge\delta_{W}$; since the variables 
in $U_{\sigma}$ are fresh and thus not in $\wit(\phi)$, we can change the value of $\B$ on $U_{\sigma}$ in a way that different variables of 
this set are mapped to different elements without affecting the satisfaction of 
$\wit(\phi)\wedge\delta_{W}$, thus obtaining a 
$\T$-interpretation $\B^{\prime}$ 
such that $\B^{\prime}$ satisfies $\wit(\phi)\wedge\delta_{V}$,
$x^{\B^{\prime}}\neq y^{\B^{\prime}}$ whenever $x\neq y$
for every $x,y\in U_{\sigma}$ and for every $\sigma$
(and, of course, $\sigma^{\B^{\prime}}$ is infinite for every sort $\s\in S$). 

Since we have now established that $\wit(\phi)\wedge\delta_{V}$ is $\T$-satisfiable, and this theory is strongly finitely witnessable w.r.t. $S$,  there must exist a $\T$-interpretation $\C_{M}$ that satisfies $\wit(\phi)\wedge\delta_{V}$ (and, since $\C_{M}$ satisfies $\wit(\phi)$, it must satisfy $\exists\,\overarrow{x}.\:\wit(\phi)$ and thus $\phi$ as well) and $\sigma^{\C_{M}}=V_{\sigma}^{\C_{M}}$ for every $\s\in S$; but, since the definition of $\delta_{V}$ forces $V_{\sigma}/E_{\sigma}$ to have $\kappa(\sigma)$ equivalence classes for $\sigma\in S^{-}_{\omega}$, and $M+|W_{\s}/F_{\s}|$ equivalence classes for $\sigma\in S^{+}_{\omega}$, we have that $\sigma^{\C_{M}}$ has $\kappa(\sigma)$, respectively $M+|W_{\s}/F_{\s}|$, elements for $\sigma\in S^{-}_{\omega}$, respectively $S^{+}_{\omega}$.

 From the $\T$-interpretaions $\C_{M}$ we constructed it is clear that 
\[\Gamma_{M}=\{\phi\}\cup\{\psi^{\sigma}_{=\kappa(\sigma)} : \sigma\in S^{-}_{\omega}\}\cup\{\psi^{\sigma}_{\geq M} : \sigma\in S^{+}_{\omega}\}\]
is $\T$-satisfiable for all $M$. We now state that, through \Cref{Compactness}, this implies that $\Gamma=\bigcup_{M\in\mathbb{N}}\Gamma_{M}$ is also $\T$-satisfiable, what will finish proving our theorem. Indeed, if $\C$ is a $\T$-interpretation that satisfies $\Gamma$, each $\sigma^{\C}$, for $\sigma\in S^{-}_{\omega}$, has cardinality $\kappa(\sigma)$; and each $\sigma^{\C}$, for $\sigma\in S^{+}_{\omega}$, will be infinite, given $\Gamma$ contains the formulas $\psi^{\sigma}_{\geq M}$ for all $M\in\mathbb{N}$. Using \Cref{LowenheimSkolemDownwards} with the union of $\Gamma$ and $\ax(\T)$ (which is satisfied by $\C$, given that that is a $\T$-interpretation), we obtain a $\T$-interpretation $\D$ that satisfies $\Gamma$, where $\sigma^{\D}$ has cardinality $\omega$ whenever it is infinite, that is, whenever $\sigma\in S^{+}_{\omega}$. Since $|\sigma^{\D}|=\kappa(\sigma)$ for $\sigma\in S^{-}_{\omega}$, given that $\D$ satisfies $\Gamma$ and therefore $\psi^{\sigma}_{=\kappa(\sigma)}$, and $|\sigma^{\D}|=\omega=\kappa(\sigma)$ for $\sigma\in S^{+}_{\omega}$, we see that $\D$ is the interpretation we wished to build.

Now, to see that $\Gamma$ is $\T$-satisfiable, suppose that it is not: by \Cref{Compactness}, there must exist finite sets $\ax_{0}\subseteq \ax(\T)$, $S^{-}_{0}\subseteq S^{-}_{\omega}$ and $S^{+}_{0}\subseteq S^{+}_{\omega}\times\mathbb{N}$ such that 
\[\ax_{0}\cup\{\phi\}\cup\{\psi^{\sigma}_{=\kappa(\sigma)} : \sigma\in S^{-}_{0}\}\cup\{\psi^{\sigma}_{\geq k} : (\sigma,k)\in S^{+}_{0}\}\]
is unsatisfiable. But, by taking $M=\max\{k : (\sigma,k)\in S^{+}_{0}\}$, we see that $\C_{M}$ is a $\T$-interpretation that satisfies $\Gamma_{M}$, the latter set including $\phi$, $\{\psi^{\sigma}_{=\kappa(\sigma)} : \sigma\in S^{-}_{\omega}\}$ and $\{\psi^{\sigma}_{\geq M} : \sigma\in S^{+}_{\omega}\}$. Of course, this last set implies $\{\psi^{\sigma}_{\geq k} : \sigma\in S^{+}_{\omega}, k\leq M\}$, contradicting the fact that the aforementioned set of formulas is supposed to be itself contradictory.
\end{proof}

We are now able to prove \Cref{newcombinationtheorem}:

\newcombteo*

\begin{proof}
Start by assuming that $\phi_{1}\wedge\phi_{2}$ is $\T$-satisfiable and take $\overarrow{x}=\vars(\psi)\setminus\vars(\phi_{2})$. Since $\phi_{2}$ and $\exists\, \overarrow{x}.\: \psi$ are $\T_{2}$-equivalent, we have that $\phi_{1}\wedge\psi$ is also $\T$-satisfiable; let $\A$ be a $\T$-interpretation which satisfies that formula.
For each $\s\in S$, take the equivalence relation $E_{\s}$ on $V_{\s}$ such that $xE_{\s}y$ iff $x^{\A}=y^{\A}$, and 
set $\delta_{V}$ to be the corresponding arrangement. 
Then $\A$, when restricted to the signature $\Sigma_{1}$, is a $\T_{1}$-interpretation that satisfies $\phi_{1}\wedge\delta_{V}$, 
and when restricted to the signature $\Sigma_{2}$, is a $\T_{2}$-interpretation that satisfies $\psi\wedge\delta_{V}$.

Now, for the reciprocal: let $\A$ be a $\T_{1}$-interpretation satisfying $\phi_{1}\wedge \delta_{V}$, and let $\B$ be a $\T_{2}$-interpretation satisfying $\psi\wedge\delta_{V}$. Using 
\Cref{LowenheimSkolemDownwards}, we can build a second $\T_{1}$-interpretation $\A^{\prime}$ that satisfies $\phi_{1}\wedge\delta_{V}$ and where $\sigma^{\A^{\prime}}$ is at most countable, for every sort $\s$; and, since 
$\T_{2}$ is strongly finitely witnessable and $\psi=\wit(\phi_{2})$, we can obtain a $\T_{2}$-interpretation $\B^{\prime}$ which satisfies $\psi\wedge\delta_{V}$ and $\sigma^{\B^{\prime}}=\vars_{\s}(\psi\wedge\delta_{V})^{\B^{\prime}}=V_{\s}^{\B^{\prime}}$ for every $\s\in S$, where $\vars_{\s}(\psi\wedge\delta_{V})=V_{\s}$ because 
$\delta_{V}$ is a formula where only variables of $V$ occur. Therefore, 
\[|\s^{\B^{\prime}}|=|V_{\s}^{\B^{\prime}}|=|V_{\s}^{\A^{\prime}}|\leq |\s^{\A^{\prime}}|\leq \omega,\]
where the second equality comes from the fact that both $\A^{\prime}$ and $\B^{\prime}$ satisfy $\delta_{V}$. Using \Cref{SI+SFW=CS} for the theory $\T_{2}$ and the facts that: $\psi\wedge\delta_{V}$ is a quantifier-free formula; $\B^{\prime}$ is a $\T_{2}$-interpretation that satisfies $\psi\wedge\delta_{V}$; and setting $\kappa(\s)=|\s^{\A^{\prime}}|$ 
for every $\s\in S$
as a function from $\{\s\in S : |\s^{\B^{\prime}}|\leq \omega\}=S$ to the class of cardinals such that $|\s^{\B^{\prime}}|\leq \kappa(\s)\leq \omega$, we can obtain a $\T_{2}$-interpretation $\C$ that satisfies $\psi\wedge\delta_{V}$ and where $|\s^{\C}|=|\s^{\A^{\prime}}|$ for every $\s\in S$. 

In other words, $\A^{\prime}$ is a $\T_{1}$-interpretation that satisfies $(\phi_{1}\wedge \delta_{V})\wedge\delta_{V}$, $\C$ is a $\T_{2}$-interpretation satisfying $(\psi\wedge\delta_{V})\wedge\delta_{V}$, and for all sorts $\s\in S$ we have $|\s^{\A^{\prime}}|=|\s^{\C}|$. We  duplicate $\delta_{V}$ when writing the formulas $(\phi_{1}\wedge \delta_{V})\wedge\delta_{V}$ and $(\psi\wedge\delta_{V})\wedge\delta_{V}$ (equivalent, more simply, to $\phi_{1}\wedge \delta_{V}$ and $\psi\wedge\delta_{V}$, respectively) as we wish to meet the conditions of \Cref{LemmafromPolitetheoriesrevisited}: if we try to apply the theorem to $\phi_{1}$ and $\psi$, we would need an arrangement over the set of variables $U$, where $U_{\s}=\vars_{\s}(\phi_{1})\cap\vars_{\s}(\psi)$, but instead we have one over $V_{\s}=\vars_{\s}(\psi)$, which may contain other variables than those of sort $\s$ shared by $\phi_{1}$ and $\psi$; however, $V_{\s}=\vars_{\s}(\phi_{1}\wedge\delta_{V})\cap\vars_{\s}(\psi\wedge\delta_{V})$, and we therefore can apply the theorem verbatim to the formulas $\phi_{1}\wedge\delta_{V}$ and $\psi\wedge\delta_{V}$ instead.

So, after applying \Cref{LemmafromPolitetheoriesrevisited}, we obtain a $\T$-interpretation $\D$ which satisfies $\phi_{1}\wedge\psi\wedge\delta_{V}$. Since $\exists\, \overarrow{x}.\:\psi$ and $\phi_{2}$ are $\T_{2}$-equivalent, it follows that $\D$ satisfies $\phi_{1}\wedge\phi_{2}$, as we wished to show.
\end{proof}

\end{document}